\documentclass[oneside,reqno,11pt,a4paper]{amsart}
\usepackage[T1]{fontenc}
\usepackage[margin=2.5cm]{geometry}
\usepackage{amsmath,amsfonts,amssymb,amscd,amsthm}
\usepackage[usenames,dvipsnames,svgnames]{xcolor}
\usepackage[utf8x]{inputenc}
\usepackage{graphicx}
\graphicspath{{}{graphs/}{results/}}
\usepackage{caption}
\usepackage{subcaption}
\usepackage[numbers,sort&compress]{natbib}
\usepackage{doi}
\usepackage{autonum}
\usepackage{makecell}
\usepackage{hyperref}
\usepackage{acronym}
\usepackage{algorithm,algorithmic}
\usepackage{textgreek}
\usepackage{url}
\usepackage{color} 
\usepackage{framed}
\usepackage{verbatim}
\usepackage{fancyvrb}
\usepackage{calc}
\usepackage{newtxtext}
\usepackage{newtxmath}
\usepackage{microtype} 
\usepackage[final]{pdfpages}

\usepackage{ragged2e}
\usepackage{enumitem}

\usepackage{booktabs}
\usepackage{tikz}
\definecolor{myellow}{RGB}{255,230,128}
\definecolor{gray20}{RGB}{204,204,204}
\definecolor{mygray}{RGB}{204,204,204}
\definecolor{mygreen}{RGB}{138,203,95}
\definecolor{myblue}{RGB}{77,151,214}
\newif\iflistings
\listingsfalse

\iflistings
   \usepackage{minted}

\newcommand{\julia}[1]{\mint{julia}|#1|}
\fi

\hypersetup{
    breaklinks=true,
    bookmarksopen=true,
    pdftitle={High order unfitted finite element discretizations for explicit boundary representations},    
    pdfauthor={Pere A. Martorell, Santiago Badia},     
    colorlinks=true,       
    linkcolor=black,          
    citecolor=blue,        
    filecolor=black,      
    urlcolor=blue           
}

\definecolor{bg}{rgb}{0.93,0.93,0.93}

\newtheorem{theorem}{Theorem}[section]

\newtheorem{proposition}[theorem]{Proposition}

\newtheorem{assumption}[theorem]{Assumption}
\newtheorem{remark}[theorem]{Remark}

\theoremstyle{definition}

\acrodef{aa}[AA]{axis-aligned}
\acrodef{sa}[SA]{side-aligned}
\acrodef{fsi}[FSI]{fluid-structure interaction}
\acrodef{amr}[AMR]{adaptive mesh refinement}
\acrodef{cad}[CAD]{computer-aided design}
\acrodef{brep}[BREP]{boundary representation}
\acrodef{cae}[CAE]{computer-aided engineering}
\acrodef{pde}[PDE]{partial differential equation}
\acrodef{cae}[CAE]{computer-aided engineering}
\acrodef{stl}[STL]{Stereolithography}
\acrodef{fe}[FE]{finite element}
\acrodef{vof}[VOF]{volume of fluid}
\acrodef{fem}[FEM]{finite element method}
\acrodef{xfem}[XFEM]{extended finite element method}
\acrodef{dof}[DOF]{degree of freedom}
\acrodefplural{dof}[DOFs]{degrees of freedom}
\acrodef{agfe}[agFE]{aggregated finite element}
\acrodef{agfem}[AgFEM]{aggregated finite element method}
\acrodef{cg}[CG]{continuous Galerkin}
\acrodef{dg}[DG]{discontinuous Galerkin}
\acrodef{amg}[AMG]{algebraic multi-grid}
\acrodef{iga}[IGA]{isogeometric analysis}
\acrodef{obb}[OBB]{oriented bounding box}
\acrodef{nurbs}[NURBS]{non-uniform rational basis splines}
\acrodef{bspline}[B-spline]{basis spline}
\acrodef{csg}[CSG]{constructive solid geometry}
\acrodef{aabb}[AABB]{axis-aligned bounding box}
\acrodef{obb}[OBB]{oriented bounding box}
\acrodefplural{aabb}[AABBs]{axis-aligned bounding boxes}
\acrodefplural{obb}[OBBs]{oriented bounding boxes}
\acrodef{kdop}[k-DOP]{discrete orientation polytope}
\acrodef{step}[STEP]{standard for the exchange of product model data}
\acrodef{occt}[OCCT]{Open CASCADE Technology}

\newcommand{\fig}[1]{Fig.~\ref{#1}}
\newcommand{\sect}[1]{Sec.~\ref{#1}}
\newcommand{\alg}[1]{Alg.~\ref{#1}}

\newcommand{\alglin}[1]{line~\ref{#1}}
\newcommand{\thm}[1]{Th.~\ref{#1}}
\newcommand{\prop}[1]{Prop.~\ref{#1}}
\newcommand{\rmk}[1]{Rmk.~\ref{#1}}

\definecolor{shadecolor}{gray}{.92}
\definecolor{incolor}{rgb}{0,0,.7}
\definecolor{outcolor}{rgb}{.65,0,0}
\definecolor{syntaxcolor}{rgb}{.65,0,0}

 \newcommand{\rev}[1]{{#1}}

\newif\ifsvgs
\svgsfalse

\newcommand{\includefig}[3][\tiny]{%
    \def\svgwidth{#2}
    #1
    
  \ifsvgs
    \updatepdffromsvg{#3}
  \fi
  \input{#3.pdf_tex}

}

\newcommand{\updatepdffromsvg}[1]{
  \executeiffilenewer{#1.svg}{#1.pdf}%
  {inkscape -z -C --file=#1.svg %
  --export-pdf=#1.pdf --export-latex && %
  sed "s|$(basename #1).pdf|$(echo #1).pdf|g" -i #1.pdf_tex }%
}

\newcommand{\executeiffilenewer}[3]{%
  \IfFileExists{#2}{}{\immediate\write18{#3}}
  \ifnum\pdfstrcmp{\pdffilemoddate{#1}}%
  {\pdffilemoddate{#2}}>0%
  {\immediate\write18{#3}}\fi%
}

\begin{document}

\title{High order unfitted finite element discretizations for explicit boundary representations}

\author[P. Martorell]{Pere A. Martorell$^{1,*}$}

\author[S. Badia]{Santiago Badia$^{1,2}$}

\thanks{\null\\
$^{1}$ CIMNE, Centre Internacional de M\`etodes Num\`erics a l'Enginyeria, Campus Nord, 08034 Barcelona, Spain.\\
$^{2}$ School of Mathematics, Monash University, Clayton, Victoria, 3800, Australia.\\
$^*$ Corresponding author.\\
E-mails: 
{\tt pamartorell@cimne.upc.edu} (PM)
{\tt santiago.badia@monash.edu} (SB) 
}

\begin{abstract}
  When modeling scientific and industrial problems, geometries are typically modeled by explicit boundary representations obtained from \acl{cad} software. Unfitted (also known as embedded or immersed) \acl{fe} methods offer a significant advantage in dealing with complex geometries, eliminating the need for generating unstructured body-fitted meshes. However, current unfitted \aclp{fe} on nonlinear geometries are restricted to implicit (possibly high-order) level set geometries. In this work, we introduce a novel automatic computational pipeline to approximate solutions of \aclp{pde} on domains defined by explicit nonlinear \aclp{brep}. For the geometrical discretization, we propose a novel algorithm to generate quadratures for the bulk and surface integration on nonlinear polytopes required to compute all the terms in unfitted \acl{fe} methods. The algorithm relies on a nonlinear triangulation of the boundary, a kd-tree refinement of the surface cells that simplify the nonlinear intersections of surface and background cells to simple cases that are diffeomorphically equivalent to linear intersections, robust polynomial root-finding algorithms and surface parameterization techniques. We prove the correctness of the proposed algorithm. We have successfully applied this algorithm to simulate \aclp{pde} with unfitted \aclp{fe} on nonlinear domains described by \acl{cad} models, demonstrating the robustness of the geometric algorithm and showing high-order accuracy of the overall method.
\end{abstract}

\maketitle


\newcommand{\mapF}{\pmb{{\phi}}_F}
\newcommand{\mapE}{\pmb{{\phi}}_E}

\section{Introduction}\label{sec:intro}

A wide range of industrial and scientific applications requires solving \acp{pde} on complex domains. These domains are commonly enclosed by \ac{brep} models and generated in \ac{cad} software. \ac{cad} models are described by \ac{nurbs} and boolean operations like \ac{csg}. Due to the high-order nature of \ac{nurbs}, it becomes essential to employ specialized tools capable of efficiently handling numerical simulations on these complex domains.

Despite the enduring popularity of body-fitted meshes in traditional simulation pipelines, they exhibit significant limitations. The generation of body-fitted unstructured mesh generation relies on manual intervention, resulting in significant bottlenecks in the process \cite{hughes_isogeometric_2005}, especially when dealing with high-order representations. Additionally, simulating \acp{pde} on body-fitted meshes with distributed memory machines necessitates mesh partitioning strategies based on graph partition techniques. These algorithms are inherently sequential and demand extensive memory resources \cite{Karypis_1997}. Consequently, the mesh partitioning process represents a major bottleneck in the simulation pipeline and cannot be automated in general.

Unfitted \acp{fem}, also known as embedded or immersed \acp{fem}, offers a solution to the mesh generation bottlenecks by eliminating the need for body-fitted meshes. Unfitted methods rely on a simple background mesh, such as a uniform or adaptive Cartesian mesh. Traditionally, unfitted methods utilize implicit descriptions (level sets) of geometries \rev{or handle explicit representations using low order quadratures (e.g., using composite quadratures on adapted meshes \cite{Dster2008}) that do not preserve sharp features. The algorithms in \cite{Badia_2022-stl} compute quadratures for oriented linear triangulations as boundary representations (a.k.a. \ac{stl}) without perturbing the geometry, preserving sharp features up to machine precision.} 
 
On the other side, several approaches have been introduced to combine unfitted methods with high-order implicit representations of geometries. \rev{One notable example is the methodology in \cite{Fries_2015}, which has been gradually improved in subsequent works,} including \cite{Legrain_2018, Fries2017, Fries_2018, Stanford_2019b}.
These advancements have even extended to handling \ac{brep} models in \cite{Stanford_2019a}.  
\rev{Other important contributions are \cite{Lehrenfeld2016,Saye2015} and references therein.} However, it is worth noting that these methods still rely on level set representations.

In embedded \acp{fem} the small cut cell problem is a significant limitation extensively discussed in the literature~\cite{DePrenter2023}. This problem arises when the intersection between physical and background domain cells becomes arbitrarily small, leading to ill-conditioning issues in the numerical solution.
Although various techniques have been proposed to address this problem, only a few have demonstrated robustness and optimal convergence.
One approach is the ghost penalty method \cite{burman2010ghost}, which is utilized in the CutFEM packages \cite{burman_cutfem_2015}. Alternatively, cell agglomeration techniques present a viable option to ensure robustness concerning the cut cell location, \rev{which were  originally applied to \ac{dg} methods~\cite{Bastian2011,Johansson2012}.}
Extensions to the $\mathcal{C}^0$ Lagrangian \ac{fe} have been introduced in \cite{Badia2018c}, while mixed methods have been explored in \cite{Badia2018a}, where the \ac{agfem} term was coined.
\ac{agfem} exhibits good numerical qualities, including stability, bounds on condition numbers, optimal convergence, and continuity concerning data. Distributed implementations have been exploited in \cite{Verdugo2019, Badia2020Jun}, also AgFEM has been extended to $h$-adaptive meshes \cite{Neiva2021} and higher-order \ac{fe} with modal $\mathcal C^0$ basis in \cite{Badia_2022-highorder}. In \cite{Badia_2022-ghost}, a novel technique combining ghost penalty methods with \ac{agfem} was proposed, offering reduced sensitivity to stabilization parameters. 

The development of \ac{iga} over the past two decades has been driven by the goal of improving the interaction between \ac{cad} and \ac{cae} \cite{hughes_isogeometric_2005}. While \ac{iga} techniques are suitable for \acp{pde} on boundaries, they cannot readily handle \acp{pde} in the volume of a \ac{cad} representation of the domain. Standard \ac{cad} representations are 2-variate (boundary representations) and they do not provide a parameterization of the volume. To overcome this limitation, some works \cite{Engvall_2016, Engvall_2017,Xia_2017,Xia_2018}  propose constructing volume parametrizations based on Bernstein-B\'ezier basis using the B\'ezier projection techniques described in \cite{Thomas_2015}.  Nevertheless, these approaches still rely on high-order unstructured meshes, thereby inheriting the known limitations associated with them, such as tangling issues, lack of parallelization, and global graph partition bottlenecks.

Similarly to unfitted \ac{fe} methods, immersed \ac{iga} \cite{Wei_2021, Antolin_2022a, Antolin_2022b}  eliminates the need for unstructured meshes by utilizing the intersection of a background mesh.
These methods utilize integration techniques for complex domains, including dimension reduction of integrands \cite{Chin_2020, Gunderman_2021}, i.e.,  integrating over lines and surfaces. The precision of these methods is bounded by the approximation algorithms used on the \emph{trimming curves}, which represent surface-surface intersections \cite{Stanford_2019b, Park_2020} resulting from boolean \ac{csg} on the \ac{cad} models.

One of the primary challenges in the \ac{cad} to \ac{cae} paradigm is the approximation of trimming curves \cite{Patrikalakis2010, Shen_2016, Li_2009}. Trimming curves, in general, cannot be represented in the intersected surface patches. In the literature, there is a wide range of strategies to approximate trimming curves, see \cite{Marussig_2017, Beer_2019} and references therein. These strategies can be categorized into analytical methods, lattice evaluation methods, subdivision methods, marching methods, or a combination thereof. The representation of the approximated trimming curves is also extensively studied. Once the trimming curves are approximated, various techniques can be employed. Untrimming techniques \cite{Xia_2017, Xia_2018, Massarwi_2018, Massarwi_2019, Antolin_2019} are one approach, which involves a conformal reparametrization of the original surface. Another approach is the direct integration onto trimmed surfaces \cite{Scholz_2019, Gunderman_2021}.

In this work, we propose a computational framework that combines unfitted \ac{fem} and implicit CAD representations of the geometry. In order to do this, we propose a novel approach to numerically integrate on unfitted cut cells that are intersected by domains bounded by a high-order \acp{brep}. Our method involves approximating the geometry using a set of B\'ezier patches, utilizing B\'ezier projection methods \cite{Borden_2010}. By leveraging the properties of B\'ezier curves, we can efficiently perform intersections. We can reduce the complexity of the collision interrogations and tangling prediction through the convex hull property. The variational diminishing property of B\'ezier curves enables root isolation for determining the intersection points \cite{Mourrain2004}. We can employ efficient multivariate root-finding techniques \cite{Reuter_2007, Mourrain_2009} for polynomials on nonrational B\'ezier patches. To build an intersection method for nonlinear {boundary representations}, we combine the intersection techniques with partition techniques typically used in level set methods \cite{Fries_2015} and linear polytopal intersection \cite{Badia_2022-stl}. We propose a kd-tree refinement of the surface B\'ezier triangulation that reduces nonlinear intersections against background cells to simple situations that are diffeomorphically equivalent to linear intersections. This allows us to handle the complexity of intersecting high-order geometries efficiently. Furthermore, we can approximate these intersections with a set of B\'ezier patches using least-squares techniques based on \cite{Borges_2002}.

With the intersection method established, we are then able to integrate on the surface of these polytopes to solve \acp{pde} on high-order unfitted \ac{fe} meshes. We employ moment-fitting techniques based on Stokes theorem \cite{Chin_2020, Badia_2022-highorder} to ensure accurate and stable integration. By utilizing these techniques, we can effectively handle the integration process on high-order unfitted \ac{fe} meshes, allowing for the solution of \acp{pde} in domains bounded by complex high-order \acp{brep}.

The outcomes of this work are as follows:
\begin{itemize}
  \item An automatic computational framework that relies on a robust and accurate intersection algorithm for background cells and B\'ezier patches of arbitrary order (briefly described above), and the mathematical analysis of the correctness of the algorithm.
  \item Accuracy and robustness numerical experimentation of the intersection algorithm. The algorithms exhibit optimal convergence rates of the surface and volume integration. These errors are robust concerning the relative position of the background cells and the \ac{brep}.
  \item The numerical experimentation of a high-order unfitted \ac{fe} method for high-order \acp{brep} with analytical benchmarks. The results demonstrate optimal $hp$-convergence of the error norms.
  \item The demonstration of the application of the methods for problems defined in \ac{cad} geometries.
\end{itemize}

The outline of this article is as follows. Firstly, in \sect{sec:agfem}, we introduce the unfitted \ac{fe} methods and their requirements for handling high-order \acp{brep}. Next, in \sect{sec:algorithms}, we provide the proposed geometric algorithms for computing the nonlinear intersections between background cells and oriented high-order \acp{brep}, along with a surface parametrization method for integration purposes. Then, in \sect{sec:experiments}, we present the numerical results obtained from applying the proposed method, including accuracy and robustness of the intersections, benchmark tests for validation of the unfitted \ac{fe} pipeline, and simulations on \ac{cad} geometries. Finally, in \sect{sec:conclusions}, we draw the main conclusions and future work lines.

\section{Unfitted \acl{fe} method}\label{sec:agfem}

\subsection{Unfitted \acl{fe} formulations}

Let us consider an open Lipschitz domain $\Omega \in \mathbb R ^3$ in which we want to approximate a system of \acp{pde}. An oriented high-order surface mesh $\mathcal{B}$ defines the domain boundary $\partial \Omega$ and encloses the domain interior.
The \acp{pde} usually involve Dirichlet boundary conditions on $\Gamma_D $ and Neumann boundary conditions on $\Gamma_N$, where $\partial \Omega \doteq \Gamma_D \cup \Gamma_N$ and $\Gamma_D \cap \Gamma_N = \emptyset$. These subsets, $\Gamma_D$ and $\Gamma_N$, correspond to geometric discretizations of $\mathcal{B_D}$ and $\mathcal{B_N}$, resp., such that $\mathcal{B} \equiv \mathcal{B_D} \cup \mathcal{B_N}$.

The principal motivation of this work is to enable the utilization of grid-based unfitted numerical schemes that can be automatically generated from the oriented high-order surface mesh $\mathcal B$.
This approach is valuable in industrial and scientific applications.
We can alleviate the geometric constraints associated with body-fitted meshes
by employing embedded discretization techniques.
These techniques utilize a background partition $\mathcal T^\mathrm{bg}$ defined over an arbitrary artificial domain $\Omega^\mathrm{art} \supseteq  \Omega$.
The artificial domain can be a simple bounding box containing $\Omega$, dramatically simplifying the computation of $\mathcal T^\mathrm{bg}$ compared to a body-fitted partition of $\Omega$.
In this work, we adopt a Cartesian mesh $\mathcal T^\mathrm{bg}$ for the sake of simplicity, although the proposed approach can readily be used on other types of background meshes, such as tetrahedral structured meshes obtained through simplex decomposition or adaptive mesh refinement (AMR) techniques.

The abstract exposition of unfitted formulations considered in this work is general and encompasses various unfitted \ac{fe} techniques from the literature. These techniques include the \ac{xfem} \cite{belytschko_arbitrary_2001}, designed for handling unfitted interface problems. In order to have robustness with respect to small cut cells, the cutFEM method \cite{burman_cutfem_2015} and the finite cell method \cite{Schillinger_2016} add \rev{additional terms to enhance stability}. The AgFEM \cite{Badia2018c} provides robustness via a discrete extension operator from interior to cut cells. Since DG methods can work on polytopal meshes, combined with cell aggregation~\cite{muller2017high}, they are also robust unfitted \ac{fe} techniques.

The definition of \ac{fe} spaces on unfitted meshes
requires a cell classification.
The background cells $K\in\mathcal{T}^\mathrm{bg}$ with a null intersection with $\Omega$ are classified as exterior cells and are denoted as $\mathcal{T}^\mathrm{out}$. These exterior cells, which have no contribution to functional discretization and can be discarded. The active mesh, denoted as $\mathcal{T} = \mathcal{T}^\mathrm{bg} \setminus \mathcal{T}^\mathrm{out}$, represents the relevant mesh for the problem (\fig{fig:unfitted-fe}). The unfitted \ac{fe} techniques stated above utilize \ac{fe} spaces defined on $\mathcal{T}$ to construct the finite-dimensional space $V$. This space approximates the solution and tests the weak form of \acp{pde}.
An abstract unfitted \ac{fe} problem reads as follows: find $u\in V$ such that
\begin{equation}
  a(u,v) = l(v), \quad \forall v \in V,
\end{equation}
where
\begin{equation}
  a(u,v) = \int _\Omega L_\Omega(u,v) d\Omega + \int _{\Gamma _D } L_D (u,v)d\Gamma + \int _{\mathcal F} L_\mathrm{sk} (u,v)d\Gamma,
\end{equation}
and
\begin{equation}
  l(v) = \int _\Omega F_\Omega(v) d\Omega + \int _{\Gamma _N } F_N (v)d\Gamma + \int _{\Gamma _D } F_D (v)d\Gamma.
\end{equation}

The bulk terms $L_\Omega$ and $F_\Omega$ consist of the weak form of the differential operator, the source term, and possibly other numerical stabilization terms. On $\Gamma_D$, the operators $L_D$ and $F_D$ represent the enforcement of the Dirichlet boundary conditions,
often implemented using Nitsche's method in unfitted formulations.
The term $F_N$ on $\Gamma_N$ represents the Neumann boundary conditions of the given problem. The skeleton of the active mesh $\mathcal F$ corresponds to the interior faces of $\mathcal{T}$, while the term $L_\mathrm{sk}$ collects additional penalty terms, such as weak continuity enforcement in \ac{dg} methods or ghost penalty stabilization techniques.

\begin{figure}[http]
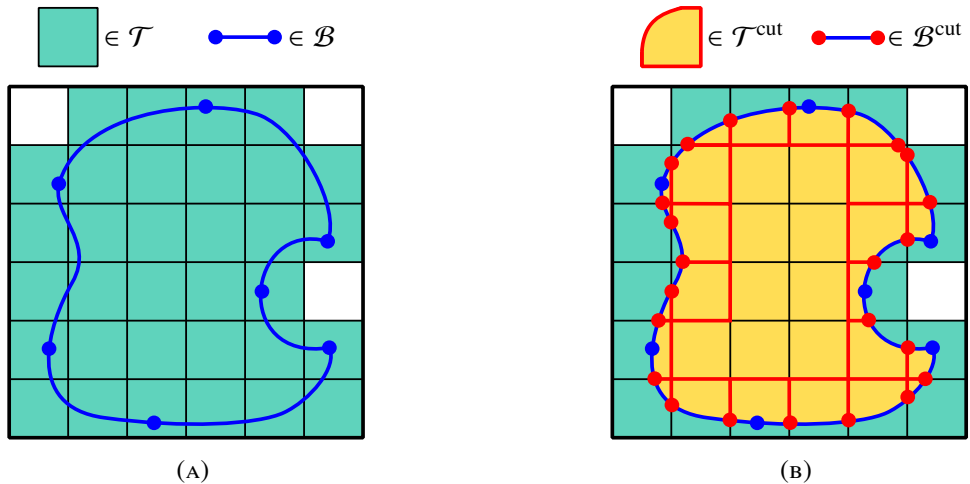

  \centering
  \begin{subfigure}[b]{0.49\textwidth}
    \centering
    \includefig[\normalsize]{0.6\textwidth}{unfitted_fe_a}
    \caption{}
  \end{subfigure}
  \begin{subfigure}[b]{0.49\textwidth}
    \centering
    \includefig[\normalsize]{0.6\textwidth}{unfitted_fe_b}
    \caption{}
  \end{subfigure}
  \caption[Example of the embedded nonlinear domain in 2D.]{Example of the embedded nonlinear domain in 2D. Figure (a) presents a nonlinear oriented skin mesh $\mathcal B$ embedded in an active mesh $\mathcal T$. The intersections, computed with the techniques proposed in this work, result in the two-level partitions $\mathcal T^\mathrm{cut}$ and $\mathcal B^\mathrm{cut}$ shown in (b). These partitions are utilized for integrating unfitted formulations.
   It is important to note that the intersections in 2D are points that can be represented exactly. However, the intersections in 3D are trimming curves that must be approximated in general.}
  \label{fig:unfitted-fe}
\end{figure}
Since \ac{fe} methods are piecewise polynomials,
integrating these terms requires a cell-wise decomposition for both bulk and surface contributions. However, to accurately respect the geometry and solve the \ac{pde} on the proper domain,
it is necessary to perform these integrals within the interiors of the domains.
In particular, we have
\begin{equation}\label{eq:bulk_integral}
  \int_\Omega (\cdot)d\Omega = \sum _{K\in\mathcal T}  \int _{K\cap \Omega} (\cdot)d\Omega.
\end{equation}

Due to the inherent characteristics of unfitted \ac{fe} methods,
the surface mesh $\mathcal B$ is not the boundary restriction of the background mesh $\mathcal T$, which defines the cell-wise polynomial \ac{fe} functions.
Consequently, to integrate the boundary terms on $\mathcal B$,  we must compute a cell-wise integral as follows:

\begin{equation}\label{eq:boundary_integral}
  \int_{\Gamma_*} (\cdot)d\Gamma = \sum _{K\in\mathcal T} \sum _ {F\in \mathcal B _*}  \int _{F\cap K} (\cdot)d\Gamma, \quad * \in \{ D, N\}.
\end{equation}

It is important to note that, in general, we do not need to restrict the integrals on the skeleton terms $L_\mathrm{sk}$ to the domain interior.
\rev{In most formulations, e.g., ghost penalty and \ac{dg} methods,
one integrates these terms on the entire face skeleton.}
Therefore, in the unfitted \ac{fe} methods, we must give particular attention to the geometrical operations involving $K\cap\Omega$ and $K \cap \mathcal B$ (or more specifically, $\mathcal B_ N$ and $\mathcal B_D$ ). However, the methodology described below also handles $F \cap \Omega$ for $F \in \mathcal{F}$, in case it is required.

\subsection{Geometrical ingredients for unfitted \aclp{fe}}\label{sec:geo-fe}

This section aims to describe the geometrical entities that we need to compute the integrals \eqref{eq:bulk_integral}-\eqref{eq:boundary_integral} of unfitted \ac{fe} formulations. These entities are easier to compute than unstructured meshes in body-fitted formulations. Our geometrical framework involves intersection algorithms that are cell-wise, and so, embarrassingly parallel. Since the \ac{fe} spaces are defined on background meshes, the cut meshes do not require to be conforming or shape-regular.

The input of our geometrical framework is an oriented mesh of non-rational triangular B\'ezier elements $\mathcal B$ whose interior is the domain $\Omega$.
Each triangle $F \in \mathcal{B}$ is the image of a B\'ezier map $\mapF : \hat{F} \rightarrow F$ of order $q$ acting on a reference triangle $\hat{F} \subset \mathbb{R}^{2}$.

In the following exposition, we use $\hat{\pmb{\xi}}$ to represent the coordinate system of the reference space of $\hat{F}$ and $\pmb{x}$ to represent the coordinate system of the physical space. Given a set of points $\hat{f}$ in the reference space of $\hat F$ we can compute  $f \doteq \mapF(\hat{f})$. Since $\hat F$ is diffeomorphic, we can also define $f$ and compute $\hat f \doteq \mapF^{-1}({f})$. Given a function $\hat{\gamma}(\hat{\pmb{\xi}})$ in the reference space of $\hat{F}$, we can compute $\gamma(\pmb{x}) \doteq \hat{\gamma}\circ \mapF^{-1}(\pmb{x})$. Reversely, we can define $\gamma(\pmb{x})$ in the physical domain and compute its pull-back $\hat \gamma(\pmb{\xi}) \doteq {\gamma}\circ \mapF(\pmb{\xi})$. We will heavily use this notation to transform sets and functions between the reference and physical spaces.

In practical applications, geometries are described through a \ac{cad} model $\mathcal B^\mathrm{CAD}$. Transforming $\mathcal B^\mathrm{CAD}$ into $\mathcal B$ generally involves an approximation process, where we can utilize least-squares methods and third-party libraries, e.g., \texttt{gmsh} \cite{Geuzaine_2009}. Next, we perform intersection algorithms between the surface representation and the background mesh, which consist of two steps as follows.

In the first step, we consider the intersection of the triangular B\'ezier elements that compose the surface $F\in \mathcal B$ with the background cells $K \in \mathcal T$.
According to \eqref{eq:boundary_integral}, the resulting mesh $\mathcal B^\mathrm{cut}$ can be represented as a two-level mesh:
\begin{equation}\label{eq:b-cut}
  \mathcal B^\mathrm{cut} \doteq \bigcup _{K\in\mathcal T} \mathcal B_K ^\mathrm{cut},
  \qquad \hbox{where} \quad
  \mathcal B^\mathrm{cut}_K \doteq \{ F \cap K : F \in \mathcal B\}.
\end{equation}
$\mathcal B^\mathrm{cut}$ is a partition of $\mathcal B$ composed of general \emph{nonlinear polytopes}
(see \fig{fig:sphere_cut}).
\rev{We define a nonlinear polytope as the image of a polytope under a diffeomorphic map.}

\begin{figure}[http]
  \centering
  \begin{subfigure}[b]{0.32\textwidth}
    \includegraphics[width=\textwidth]{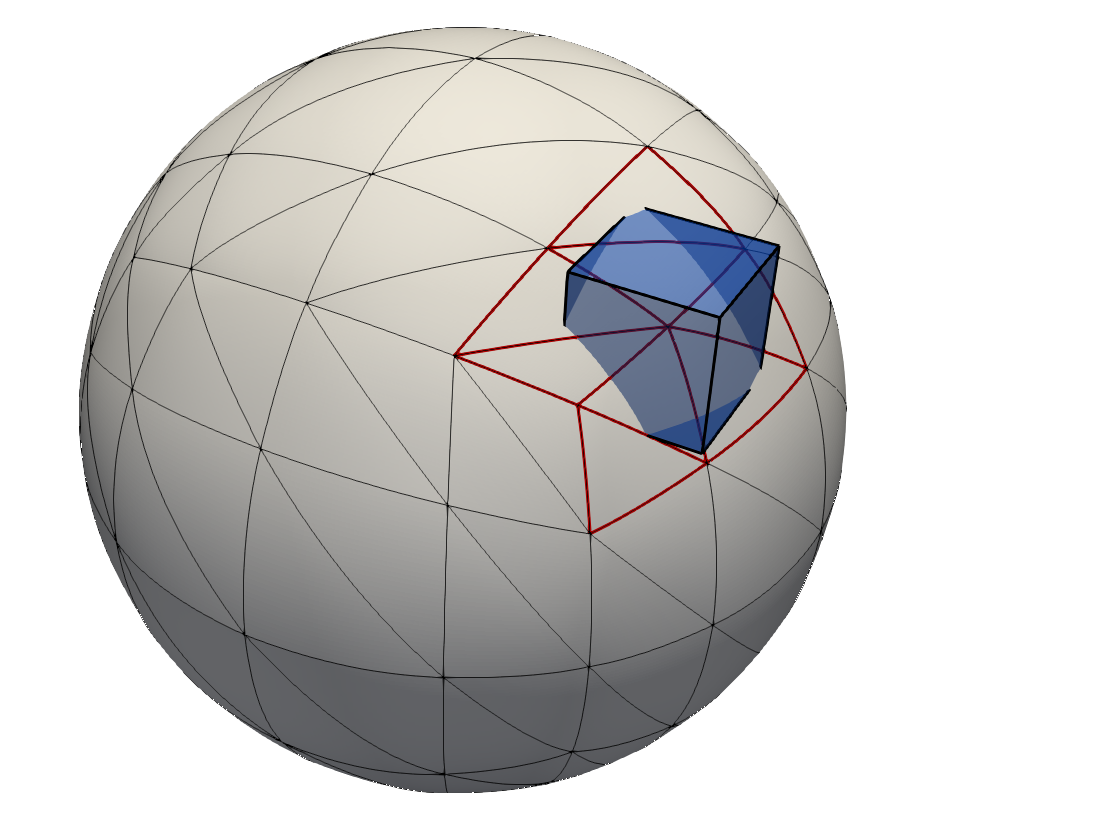}
    \caption{$\mathcal B$, $\mathcal B_K$ and $K$}
  \end{subfigure}
  \hspace{-1cm}
  \begin{subfigure}[b]{0.32\textwidth}
    \centering
    \includegraphics[width=0.8\textwidth]{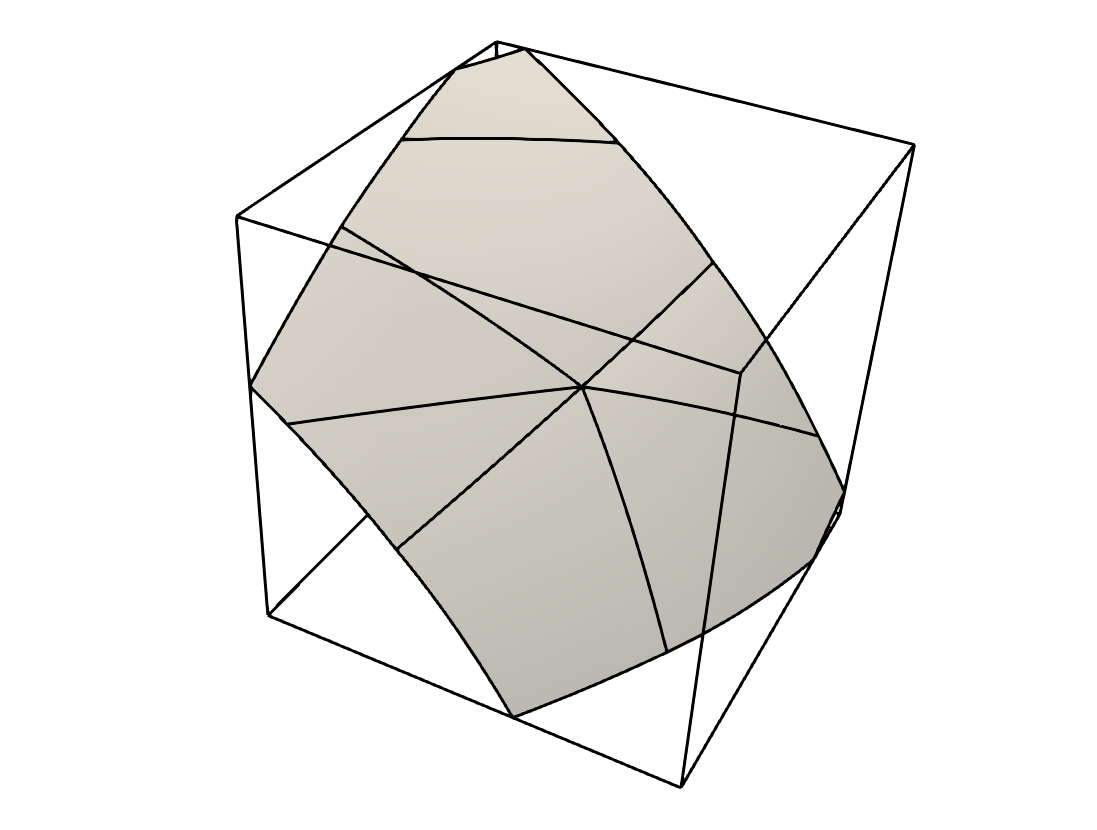}
    \caption{$\mathcal{B}^{\mathrm{cut}}_K$}
  \end{subfigure}
  \hspace{-1cm}
  \begin{subfigure}[b]{0.32\textwidth}
    \includegraphics[width=\textwidth]{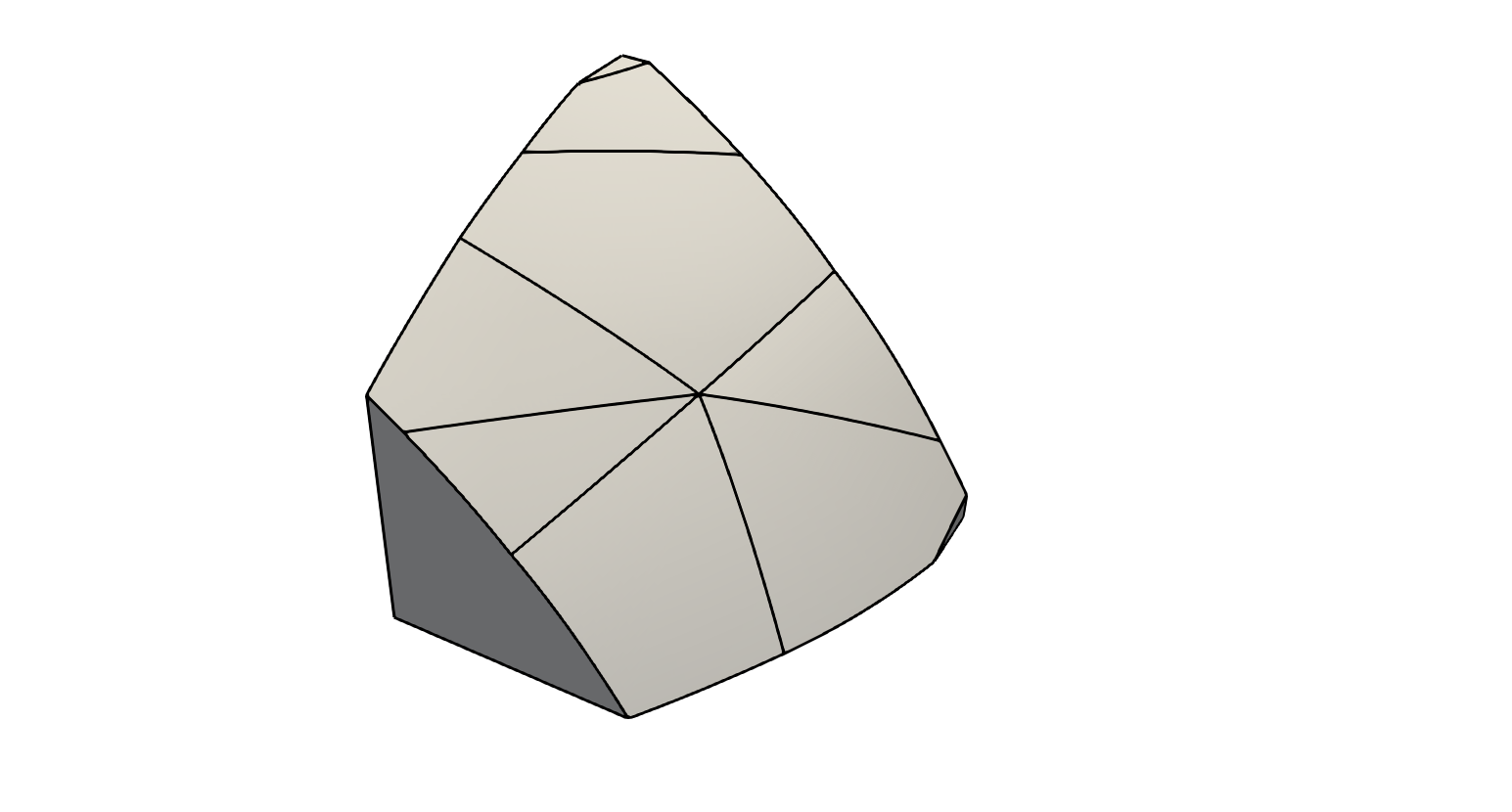}
    \caption{$K^{\mathrm{cut}} $}
  \end{subfigure}
  \caption[Representation of the surface $\mathcal{B}\doteq \partial \Omega$, its intersection $\mathcal{B}^{\mathrm{cut}}_K \doteq \mathcal{B} \cap K$ for a background cell $K \in \mathcal{T}$, and the domain interior of the cell $K^\mathrm{cut} \doteq K \cap \Omega$.]{Representation of the surface $\mathcal{B}\doteq \partial \Omega$ (see (a)), its intersection $\mathcal{B}^{\mathrm{cut}}_K \doteq \mathcal{B} \cap K$ for a background cell $K \in \mathcal{T}$ (see (b)), and the domain interior of the cell $K^\mathrm{cut} \doteq K \cap \Omega$ (see (c)).
  In order to compute $\mathcal{B}^{\mathrm{cut}}_K$, we identify first the subset of cells $\mathcal{B}_K \subseteq \mathcal{B}$ touching $K$ (see cells with red edges in (a)). Next, for each triangle $F \in \mathcal{B}_K$, we compute the intersection of $F \cap K$ at the reference \ac{fe}, i.e., we compute $\mapF^{-1}(F \cap K)$. Finally, we intersect $K$ with the surface portion $\mathcal B^\mathrm{cut}_K$ to obtain $K^\mathrm{cut}$. It is worth to note that $\mathcal{B}^{\mathrm{cut}}_K \subset \partial K^\mathrm{cut}$.
  }
  \label{fig:sphere_cut}
\end{figure}

In the second step, we compute a mesh of general nonlinear polyhedra that represent the domain interior of background cells (see \sect{sec:cell-int}):
\[
\mathcal{T}^{\mathrm{cut}} \doteq \bigcup_{K \in \mathcal{T}} K^{\mathrm{cut}}, \qquad \hbox{where} \quad K^\mathrm{cut}\doteq K \cap \Omega.
\]
The surface cut cells in $\mathcal{B}_{K}^{\mathrm{cut}}$  are nonlinear polygons that can be split into simplices.
A simplex decomposition of the volumetric cells in $\mathcal{T}^{\mathrm{cut}}$ is much more complex and expensive. Instead, we avoid a trivariate representation of $\mathcal{T}^{\mathrm{cut}}$ and instead rely on a bivariate representation of $\partial \mathcal{T}^{\mathrm{cut}}$.
We consider a boundary representation (often abbreviated B-rep or BREP in solid modeling and computer-aided design) of ${K}^{\mathrm{cut}}$ as the interior of an oriented closed surface represented as the collection of connected oriented surface elements, i.e.:
\begin{align}\label{eq:T-brep}
  \partial K^{\mathrm{cut}} = (\partial {K})^{\mathrm{cut}} \cup \mathcal{B}^{\mathrm{cut}}_K, \qquad  (\partial K)^{\mathrm{cut}} \doteq \bigcup_{F \in \Lambda^2(K)}  F \cap \Omega.
\end{align}

\subsection{Integration methods for cut cells}\label{sec:integration-methods}
For the computation of the volume integrals in (\ref{eq:bulk_integral}), since we define cut cells by its boundary representation in (\ref{eq:T-brep}), we compute the integration of polynomials on volume cut cells by transforming them into surface integrals via Stokes theorem \cite{Chin_2020, Badia_2022-highorder}. Let us denote $r$th order polynomial differential $k$-forms in $\mathbb{R}^3$ with $\mathcal{P}_r\Lambda^k(\mathbb{R}^3)$ and $d$ the exterior derivative (see, e.g., \cite{Arnold2018}).  For any $\omega(\pmb{x}) \in \mathcal{P}_r \Lambda^{3}(\mathbb{R}^3)$, one can readily find $\sigma \in \mathcal{P}_{r+1} \Lambda^{2}(\mathbb{R}^3)$ such that $\omega(\pmb{x}) = d \sigma(\pmb{x})$, due to the exactness of the polynomial de Rham complex. Using Stokes theorem, one gets:
\begin{align}\label{eq:div-integral}
  \int_{{K}^{\mathrm{cut}}} \omega
  &= \int_{\partial {K}^{\mathrm{cut}}} \sigma
  = \int_{\mathcal{B}_{K}^{\mathrm{cut}}} \sigma +
  \int_{(\partial{K})^{\mathrm{cut}}} \sigma
  = \sum_{F \in \mathcal{B}_{K}} \int_{F \cap K} \sigma + \sum_{F \in \Lambda^2(K)} \int_{F \cap \Omega} \sigma
  \\
 &
   = \sum_{F \in \mathcal{B}_{K}} \int_{\mapF^{-1}(F \cap K)} \mapF^*(\sigma) + \sum_{F \in  \Lambda^2(K)} \int_{\mapF^{-1}(F \cap \Omega)} \mapF^*(\sigma).
\end{align}
where $\mapF^*(\sigma)$ denotes the pull-back of $\sigma$ by $\mapF$. In order to create a quadrature on ${K}^{\mathrm{cut}}$, we combine this expression with a moment-fitting technique. First, we integrate all the elements in the monomial basis for $\mathcal{P}_{r}(\mathbb{R}^3)$; e.g., given a monomial $3$-form $\omega(\pmb{x}) = x^\alpha y^\beta z^\gamma dx \wedge dy \wedge dz$, we define $\sigma(\pmb{x}) = \frac{1}{\alpha+1} x^{\alpha+1} y^\beta z^\gamma dy \wedge dz$ (which holds $\omega = d \sigma$) and use (\ref{eq:div-integral}). After computing these integrals for the monomial basis at each face $F \in \mathcal{B}_K$ , we can combine them to compute the integrals of Lagrangian polynomials, use the corresponding nodal interpolation, and end up with a quadrature on the cut cell $K^{\mathrm{cut}}$.

To compute the surface integrals in (\ref{eq:boundary_integral}) and the right-hand side of (\ref{eq:div-integral}), we have several options. One approach is to generate a simplex nonlinear mesh of the boundary of the nonlinear polytopes in $\mathcal{T}^{\mathrm{cut}}$, which is feasible on surfaces. We can use a moment-fitting method to compress the resulting quadrature. E.g., for the integral on $\mathcal{B}_{K}^{\mathrm{cut}}$, we consider a nonlinear triangulation $\mathcal{S}^{\mathrm{cut}}_{\hat{F}}$ of $\pmb{\phi}^{-1}_F(F \cap K)$ and compute:
\begin{align}\label{eq:boundary-integral-simplexify}
  \int_{\mathcal{B}_{K}^{\mathrm{cut}}} \sigma = \sum_{F \in \mathcal{B}_{K}} \sum_{S \in \mathcal{S}^{\mathrm{cut}}_{\hat{F}}} \int_{S} \mapF^*(\sigma)
  = \sum_{F \in \mathcal{B}_{K}} \sum_{S \in \mathcal{S}^{\mathrm{cut}}_{\hat{F}}} \int_{\hat{F}} \pmb{\phi}_S^* \circ \mapF^*(\sigma)
\end{align}
where we rely on a map $\pmb{\phi}_S: \hat{F} \rightarrow S$ to transform the integral to the reference triangle $\hat{F}$; $S$ is a nonlinear triangle in $\mathbb{R}^2$ because it can have nonlinear edges. We proceed analogouly for the surface integral on $(\partial K)^{\mathrm{cut}}$; this is a simpler case, since the surface is already contained in a plane in the physical space.

The edges of $S \in \mathcal{S}_{\hat{F}}^{\mathrm{cut}}$ that result from intersections are implicitly defined. Thus, we must consider an approximation $\tilde \phi_S$ of the map $\phi_S$ in (\ref{eq:boundary-integral-simplexify}). This approximation is detailed in the next section. Since $F$ (and as a result $S \in \mathcal{T}^{\mathrm{cut}}$) is a smooth manifold, the polynomial approximation $\tilde \phi_S$ will introduce an error that is reduced by increasing the polynomial order of the approximation or by reducing the diameter of $S$. In turn, the diameter goes to zero with both the background mesh $\mathcal{T}$ and surface mesh $\mathcal{B}$ characteristic sizes.

Another approach is to integrate the pull-back $\mapF^*(\sigma)$ of the polynomial differential form $\sigma$ in (\ref{eq:div-integral}) on the surface to $\hat{F}$ (which are polynomials of higher order due to the factor $\mathrm{det}(\frac{\partial \mapF}{\partial \pmb{\xi}})$ that comes from the pullback) and transform the integral to $\partial \hat{F}$ using Stokes theorem:
\begin{align}\label{eq:boundary-integral-edge}
\int_{F \cap K} \sigma &= \int_{{\mapF}^{-1}(F \cap K)}
\mapF^*(\sigma) 
 = \int_{\mapF^{-1}(F \cap K)} d \varrho \\
& = \int_{\partial \mapF^{-1}(F \cap K)} \varrho = \int_{\mapF^{-1}(\partial F \cap K)} \varrho, \qquad \forall \varrho  \in \mathcal{P}_{r+2q} \Lambda^1(\mathbb{R}^2) \ : \ d \varrho =
\mapF^*(\sigma). 
\end{align}
We note that $\mapF^*(\sigma) 
\in \mathcal{P}_{r+2q-1}\Lambda^{2}(\mathbb{R}^2)$ and the existence of $\varrho$ is assured by the exactness of the polynomial de Rham complex. In this case, we can extract the (nonlinear) edges $E$ of $\partial F \cap K$ and define a map $\mapE : \hat{E} \rightarrow E$ from $\mapF$. Next, we can compute the integrals in a reference segment, i.e.,
\begin{align}\label{eq:boundary-integral-edge-2}
  \int_{F \cap K} \sigma & = \sum_{E \in \partial F \cap K}  \int_{\pmb{\phi}_{E}^{-1}(E)} \varrho.
\end{align}

As above, since some edges $E \in \partial F \cap K$ are only implicitly defined, we must consider approximations $\tilde{\pmb{\phi}}_E$ of $\mapE$ in (\ref{eq:boundary-integral-edge}). However, we do not need to compute approximated surface maps or triangulations of $\mapF^{-1}(F \cap K)$. Since we approximate each edge separately, and both $F$ and the faces of $K$ are smooth manifolds, a polynomial approximation provides error bounds that vanish by increasing the polynomial order or by reducing the edge sizes; edge sizes also go to zero with both the background mesh $\mathcal{T}$ and surface mesh $\mathcal{B}$ characteristics sizes.

\section{Intersection algorithm}\label{sec:algorithms}
In this section, we describe an algorithm that takes a background mesh $\mathcal T$ and a high-order oriented surface $\mathcal B$ as inputs and returns $\mathcal{T}^\mathrm{cut}$ and $\mathcal{B}^{\mathrm{cut}}$. This algorithm is robust to the relative position of $\mathcal T$ and $\mathcal B$. Moreover, it provides an accurate description of the intersections. We list below the different steps considered in the definition of the algorithm:

\begin{enumerate}
  \item In \sect{sec:intersection-points}, we describe the nonlinear computations of the singular points utilized in the intersection algorithms. There, we utilize multivariate root-finding techniques \cite{Mourrain_2009}.
  \item In \sect{sec:surface-trimming} and \sect{sec:connection-algorithm}, we propose an accurate intersection algorithm for nonlinear polyhedra. It combines refinement strategies (to simplify the nonlinear intersection to simple cases) and linear clipping methods \cite{Badia_2022-stl}.
  \item \sect{sec:surface-partition} describes a partition method for the intersected surface. This partition is composed of standard polytopes and is ready to be parametrized.
  \item In \sect{sec:surf-approx}, we introduce a parametrization method for intersected nonlinear polyhedra. This parametrization consists of a combination of least-squares methods
  \cite{Borges_2002} and sampling strategies \cite{Fries_2015}. It returns a set of nonrational B\'ezier patches.
  \item In \sect{sec:cell-int}, we build the polyhedral representation of the intersected cells. In addition, we provide the tool to parametrize the resulting surface.
\end{enumerate}
  We describe a global algorithm that combines the previous algorithms in \sect{sec:global-alg}.

\subsection{Intersection points}\label{sec:intersection-points}
In this section, we perform some intersection algorithms that will be required in our geometrical framework.
In order to compute surface intersections, we
require the computation of three types of intersections, described below.

Consider an oriented plane $\pi$ defined by a point $\pmb{x}_\pi$ on the plane and its outward normal $\pmb{n}_\pi$. The plane-point distance function can be computed as
\begin{align}\label{eq:distance-function}
  \gamma_\pi(\pmb{x}) \doteq \mathtt{dist}(\pi,\pmb{x}) \doteq (\pmb{x}-\pmb{x}_\pi)\cdot \pmb{n}_\pi, \qquad \pmb{x} \in \mathbb{R}^{3}.
\end{align}
In the following, we will also consider this distance function parametrized in the reference face $\hat{F}$ as $\hat{\gamma}_\pi(\hat{\pmb{\xi}}) \doteq \gamma_\pi \circ \mapF(\hat{\pmb{\xi}})$. Let us denote by $\mathrm{int}(\gamma_\pi)=\{\pmb{x}:\gamma_\pi(\pmb{x})<0\}$ the interior of a level set. The zero level set of the distance function is represented by \begin{align}\label{eq:level-set-distance-function}\hat \gamma^0_\pi \doteq \{ \  \hat{\pmb{\xi}} \in \hat F \ : \ \mathtt{dist}(\pi,\mapF(\hat{\pmb{\xi}})) =0\ \}.\end{align}

\subsubsection{Curve-plane intersection}
First,
we define the curve-plane intersections as the intersection of an edge $E = \mapE(\hat{E})$ (which can be described as a diffeomorphic polynomial map on a reference segment $\hat{E}$) with a plane $\pi$ (see \fig{fig:intersection-edge}). We compute the intersection points in the reference segment $\hat{E}$ as the result of finding:
\begin{equation}\label{eq:curve-plane}
\hat{t} \in\hat E \ : \ \gamma_\pi \circ \mapE(\hat{t}) =  0.
\end{equation}
If $E \in \Lambda^1(F)$, the map $\mapE: \hat{E} \rightarrow E$ can be readily obtained from the corresponding face map $\mapF$. Since $\mapE$ is a polynomial, the computation of $\hat{t}$ involves finding the roots of a univariate polynomial system. One can utilize root isolation techniques combined with standard iterative solvers, e.g., the Newton-Rapson method. As the solution may be not unique, the root isolation techniques in  \cite{Mourrain2004} provide the means to find multiple roots by leveraging the variation diminishing property of B\'ezier curves.
\begin{figure}[http]
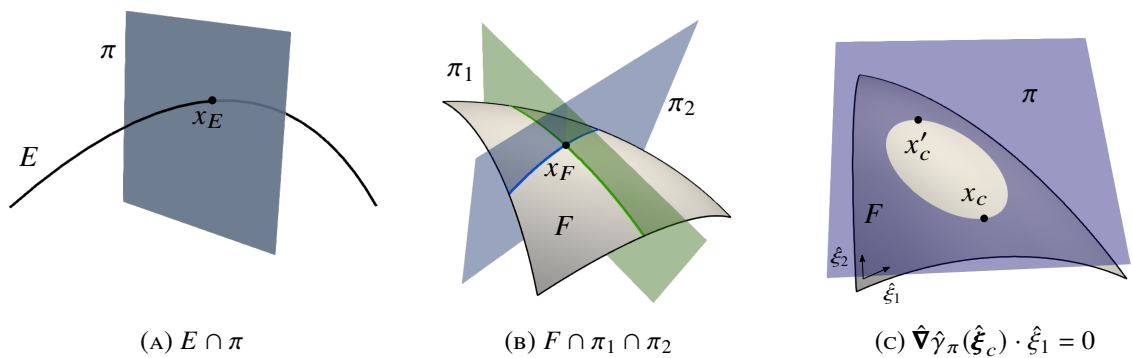

  \centering
  \begin{subfigure}[b]{0.32\textwidth}
    \includefig[\normalsize]{\textwidth}{edge_plane_intersection}
    \caption{$E \cap \pi$}
    \label{fig:intersection-edge}
  \end{subfigure}
  \begin{subfigure}[b]{0.32\textwidth}
    \centering
    \includefig[\normalsize]{0.8\textwidth}{plane_plane_intersection}
    \caption{$F \cap \pi_1 \cap \pi_2$}
    \label{fig:intersection-plane}
  \end{subfigure}
  \begin{subfigure}[b]{0.32\textwidth}
    \centering
    \includefig[\normalsize]{0.8\textwidth}{aa-critical-point_phys}
    \caption{$\hat{\pmb{\nabla}}\hat \gamma_\pi (\hat {\pmb{\xi}}_c) \cdot \hat \xi_1 =0$}
    \label{fig:intersection-critical}
  \end{subfigure}

  \caption[Definition of intersection points.]{Definition of intersection points. The intersection of curves and planes (see (a)) is computed through univariate root-finding methods.  The surface-plane-plane intersection points in (b) also represent the intersection of surface-ine intersections.
  The \acs{aa} critical points in (c) are defined in the reference space of $F$. These points split the surface-plane intersection curves into monotonic curves in the reference space (see an analogous reference space in \fig{fig:refinement-rules-a}).
  Both computations, surface-plane-plane intersection points and \acs{aa} critical points, require solving a bivariate root system.}
\end{figure}

\subsubsection{Surface-line intersection}
\label{ssub:surface-line-intersection}
Next, we consider line-surface intersections. Since a line can be represented as the intersection of two half-spaces, we compute surface-line intersections as surface-plane-plane intersections (see \fig{fig:intersection-plane}). To compute these intersections, we find:
\[ \hat{\pmb{\xi}}\in\hat F \  : \  \gamma_{\pi_i} \circ \mapF(\hat{\pmb{\xi}}) = 0, \, \quad \forall \ i \in \{1,2\},\]
where $\pi_1$, $\pi_2$ are the planes whose intersection is the line we want to intersect. We note that $\hat{\pmb{\xi}} \in \mathbb{R}^2$. It can be checked that this is a system of two bivariate polynomial equations of order $q$. The roots can be computed using bivariate root-finding algorithms for polynomials. The algorithms described in \cite{Mourrain_2009} bound and subdivide the reference domain until the roots are isolated with a given precision. In this work, we adapt the implementation to simplices using the bounding techniques defined in \cite{Reuter_2007}.

\subsubsection[Critical points of the zero isosurface of a distance function]{Critical points of the zero isosurface of a distance function with respect to an edge.}
  Finally, we compute the points in which the zero let set curve
  $\hat \gamma^0_\pi$  has a zero directional derivative with respect to a given direction $\pmb{t}$ (see \fig{fig:intersection-critical}). When the direction $\pmb{t}$ is \ac{aa}, we call these points \emph{\ac{aa} critical points}. These points are obtained as the solution  of the following system:
  \begin{equation}
  \hat{\pmb{\xi}} \in \hat F \ : \ \hat{\pmb{\nabla}} \hat  \gamma_\pi (\hat{\pmb{\xi}})  \cdot \pmb{t} = 0, \quad   \hat  \gamma_\pi(\hat{\pmb{\xi}}) = 0,
\end{equation}
which is again a system of two bivariate polynomial equations of order $q$ and can be computed as above. {We note that the critical points are defined in the reference space. Thus, we take the gradient in the reference space.}

\subsection{Nonlinear trimming surface}\label{sec:surface-trimming}
Given a face $F\in\mathcal{B}_K$, we consider its intersection against a cell $K \in \mathcal{T}_h$. The objective is to create a partition of $\mathcal{B}_K$ by splitting its faces by intersection against the background cells $K \in \mathcal{T}_h$. We consider $\mathcal{T}_h$ to be a partition of $\Omega_\mathrm{art}$ in a pointwise sense, i.e., every $\pmb{x} \in \Omega_\mathrm{art}$ belongs to only one $K \in \mathcal{T}_h$. $K$ is neither open nor closed in general, \rev{i.e., each background facet belongs to a single background cell. In particular, \rev{given a cell whose closure is $\overline{K} = [x^-,x^+]\times [y^-,y^+] \times [z^-,z^+]$, we define the cell as $K = (x^-,x^+]\times (y^-,y^+] \times (z^-,z^+]$.}}
We refer to \rev{\cite[Sec. 3.6]{Badia_2022-stl}} for more details. The pointwise partition is required to properly account for the case in which $F \in \mathcal{B}_K$ is (in machine precision) on $K$; otherwise, we could be counting more than once the same face $F$. We utilize the notation $K^\mathrm{\circ}$ and $\overline{K}$ to refer to the interior and closure of $K$, resp.
\begin{assumption}
    We assume that $F \cap K^\circ \neq \emptyset$.
\end{assumption}
In \cite{Badia_2022-stl}, we propose an algorithm that determines whether the assumption holds, i.e., checks if $F\cap \pi_f = F$ for some $f \in \Lambda^2 (K)$. Here, $f \in \Lambda^2 (K)$ represents a face bounding $K$ and $\pi_f$ its corresponding half-space. If the assumption does not hold, $F\cap K = F \cap f$, and we perform the intersection in one dimension less, which is a simplification of the one below.
\begin{assumption}
  We assume that all intersection queries are well-posed, i.e., they return a finite number of points.
\end{assumption}
For simplicity in the exposition, we postpone the ill-posed cases till the end of the section.

For each $f \in \Lambda^2(K)$ and its corresponding half-space $\pi_f$, we can define the distance function $\hat{\gamma}_{\pi_f}$ in the reference space using (\ref{eq:distance-function}) , which we denote with $\hat{\gamma}_{f}$ for brevity. We can also consider the zero level set curves $\{\hat{\gamma}^0_{f}\}_{f \in \Lambda^2(K)}$, using (\ref{eq:level-set-distance-function}), and their intersections:
\[
    \hat{\gamma}^0_{ f} \cap \hat{\gamma}^0_{{f}'}, \quad f, f' \in \Lambda^2(K).
\]
We note that the intersection points of these curves are the intersections $\pmb{\phi}_F^{-1}( \pi_f \cap \pi_{f'})$; $\pi_f \cap \pi_{f'}$ contains the edge $e \in \Lambda^1(K)$ such that $e \in f, f'$. We also note that $\hat \gamma_f^0 \cap \overline{\hat{F}},\ f \in \Lambda^2(K)$, can have several disconnected parts. We represent with $C_f(\hat{F})$  the set of mutually disconnected components of $\hat \gamma_f^0 \cap \overline{\hat{F}}$.

In the following, we want to refine $F$ in such a way that the nonlinear case is \emph{diffeomorphically} equivalent to the linear one. Thus, we can use a linear clipping algorithm to compute the cyclic graph representation of the nonlinear polytope (see \cite{Badia_2022-stl}).  For this purpose, we consider the \ac{aa} refinement $\mathrm{ref}_1$ of $F$ presented in \alg{alg:aa-refinement} and illustrated in \fig{fig:refinement-rules}. \rev{This \ac{aa} refinement generates a $kd$-tree partition. A $kd$-tree is a binary tree that splits a $k$-dimensional space of each non-leaf node through a hyperplane (see, e.g., \cite{wald2006building} for details).}

Let us consider the level set curve $\hat \gamma^0_{\varepsilon_d}$ associated with the diagonal edge $\varepsilon_d \in \Lambda^1(\hat F)$. After $\mathrm{ref}_1$, these edges cannot contain intersection points in their interior. Having intersection points only on \ac{aa} edges simplifies the exposition of the next algorithms. We cannot do the same for \ac{aa} edges since new \ac{aa} edges are generated after each refinement step and new intersections can appear.

\begin{algorithm}
  \caption{$\mathrm{ref}_1$: \acl{aa} refinement of invariants}\label{alg:aa-refinement}
  \justify
  Let us consider a $kd$-tree partition $\mathrm{ref}_1^0(\hat F)$ of $\hat F$ by refining against
  \begin{enumerate}
    \item
    $\hat{\xi}_1$-aligned lines that contain the $\hat{\xi}_2$-aligned critical points of $\hat \gamma_f^0$, $f\in \Lambda^2(K)$. Proceed analogously for $\hat{\xi}_2$-aligned lines and $\hat{\xi}_1$-aligned critical points.
    \item $\hat{\xi}_1$ and $\hat{\xi}_2$-aligned lines containing each intersections $\hat \gamma_f^0 \cap \hat \gamma_{f'}^0$ and $\hat \gamma_f^0 \cap \hat \gamma^0_{\varepsilon_d}$, for all $f,f' \in \Lambda^2 (K)$.
  \end{enumerate}
  After this refinement, if $\hat{\gamma}^0_f \cap \partial \hat F_{R},\ f\in\Lambda^2(K),\ \hat F_{R} \in \mathrm{ref}^0_1(\hat F)$, has more than two intersection points, consider the partition $\mathrm{ref}_1^1(\hat F_{R}$) by refining against
  \begin{enumerate}
    \setcounter{enumi}{2}
    \item \ac{aa} lines perpendicular to the \ac{aa} edge $\varepsilon \in \Lambda^1(\hat F_{R})$ containing each of these intersection points.
  \end{enumerate}
  Next, if $\hat{\gamma}_{f} \cap \hat{\gamma}_{f'}$  or $\hat \gamma_f \cap \hat \gamma^0_{\varepsilon_d}$, $f,f' \in \Lambda^2(K)$, intersect more than once $\hat F _{R} \in \mathrm{ref}^1_1 \circ \mathrm{ref}^0_1(\hat{F})$,
  \begin{enumerate}
    \setcounter{enumi}{3}
    \item Consider a partition of $\hat{F}_{R}$ by an \ac{aa} line that passes between two vertices.
  \end{enumerate}
  Return the final $kd$-tree partition $\mathrm{ref}_1(\hat F)$ after these four steps.
\end{algorithm}

\begin{figure}
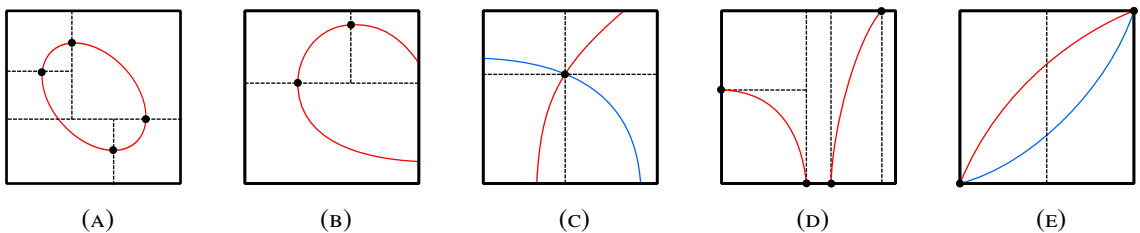

  \centering
  \begin{subfigure}[b]{.19\textwidth}
    \centering
    \includefig[\tiny]{.8\textwidth}{invariant_1}
    \caption{}
    \label{fig:refinement-rules-a}
  \end{subfigure}
  \begin{subfigure}[b]{.19\textwidth}
    \centering
    \includefig[\tiny]{.8\textwidth}{invariant_2}
    \caption{}
    \label{fig:refinement-rules-b}
  \end{subfigure}
  \begin{subfigure}[b]{.19\textwidth}
    \centering
    \includefig[\tiny]{.8\textwidth}{invariant_3}
    \caption{}
    \label{fig:refinement-rules-c}
  \end{subfigure}
  \begin{subfigure}[b]{.19\textwidth}
    \centering
    \includefig[\tiny]{.8\textwidth}{invariant_6}
    \caption{}
    \label{fig:refinement-rules-d}
  \end{subfigure}
  \begin{subfigure}[b]{.19\textwidth}
    \centering
    \includefig[\tiny]{.8\textwidth}{invariant_4}
    \caption{}
    \label{fig:refinement-rules-e}
  \end{subfigure}
  \caption[Refinement steps of \alg{alg:nonlinear:aa-refinement}.]{
  Refinement steps of \alg{alg:aa-refinement}. Step 1 is represented in (a) and (b).
  Steps 2-4 are represented in (c)-(e), resp.
  Red and blue dashed lines represent the $\hat \gamma_f^0$ and $\hat \gamma_{f^\prime}^0$-curves,  resp., $f,f^\prime \in \Lambda^2(K)$. Interior black lines represent the proposed partition of $\hat F$ for a given invariant. Solid black lines are the edges of $\Lambda^1( \hat F)$. There are multiple possible partitions, depending on the order the intersections are processed.}
  \label{fig:refinement-rules}
\end{figure}

\begin{proposition}\label{prop:bounded-refinement}
The number of faces in $\mathrm{ref}_1(\hat{F})$ is bounded.
\end{proposition}
\begin{proof}
First, we note that there is a limited number of intersections and critical points since $\mapF$ is a polynomial. Thus, the number of lines in the refinement strategy is bounded, so the resulting number of faces is also bounded. Furthermore, these points do not change with refinement (the sides of the children faces can only be parallel to the sides of the parent face). Only the intersection points in step (3) do change with refinement. However, we only need to perform step (3) once for the purpose of \prop{prop:monotonicity}.
\end{proof}

\begin{proposition}\label{prop:monotonicity}
  After the refinement strategy above, for any $f\in\Lambda^2(K)$, $C_f(\hat{F}_R)$ is a set of smooth connected curves such 
  that: given $\hat{F}_R \in \mathrm{ref}_1(\hat{F})$
  \begin{enumerate}[label=(\roman*)]
    \item the curves in $C_f(\hat{F}_R)$ cannot be tangent to any \ac{aa} line in $\hat F_R \setminus \Lambda^0 (\hat{F}_R)$;
    \item every curve in $C_f(\hat{F}_R)$ is strictly monotonic with respect to the coordinates $(\hat\xi_1,\hat\xi_2)$, can only intersect once \ac{aa} segments in $\hat F_R \setminus \Lambda^0(\hat{F}_R)$ and can only intersect diagonal edges on $\Lambda^0(\hat{F}_R)$;
    \item all the curves in $C_f(\hat{F}_R)$ are strictly increasing with respect to the coordinates ($\hat \xi_1,\hat \xi_2$) (or all strictly decreasing).
  \end{enumerate}
\end{proposition}
\begin{proof}

First, we note that \ac{aa} critical points do not change by refinement. The curves are smooth since they are the intersection of a smooth surface (image of a polynomial map) with a plane.

Let us prove (i) by contradiction. We assume that some $\hat \gamma^0_f$, $f \in \Lambda^2(K)$ is tangent to a $\hat \xi_1$-aligned \ac{aa} line on $\hat{F}_R \setminus \Lambda^0(\hat{F})$. The point was a $\hat\xi_1$-critical point before the refinement. As a result, after refinement, this point belongs to a $\hat\xi_2$-aligned axis of $\hat{F}_R$ (by (1) in $\mathrm{ref}_1$). Then, the point can only be in $\Lambda^0(\hat{F}_R)$. We proceed analogously for $\hat \xi_2$-aligned lines.

Next, we prove (ii). By definition of the refinement rule, there cannot be \ac{aa} critical points in $\hat{F}_R^\circ$. We also know that the curves cannot be tangent to \ac{aa} sides. Thus, by the implicit function theorem, $\hat \gamma _{ f}(\hat{\pmb{ \xi}}) = 0 $ can be expressed as the graph of a function $d_1(\hat \xi_1) = \hat \xi_2$ and $d_2(\hat \xi_2) = \hat \xi _1$ in $\hat F \setminus \Lambda^0 (\hat F)$. Since these functions are smooth, by the mean value theorem, they must be strictly monotonic in $\hat F$ with respect to $\hat \xi_1$ and $\hat \xi_2$.
Thus, they can intersect \ac{aa} lines at most once. Besides, all intersections against diagonal sides are on $\Lambda^0(\hat{F}_R)$ after refinement (by (2) in $\mathrm{ref}_1$); note that we are not inserting new non-\ac{aa} edges by refinement.

{Finally, we prove (iii). We note that the partition in step (3) or $\mathrm{ref}_1$ prevents two curves $\hat \alpha_f, \hat \alpha_f^\prime \in C_f(\hat F_R)$, strictly increasing and decreasing, resp., to be in the same $\hat F_R$. Let us assume that these two curves are in $\hat F_R$ before step (3). Then, by refining through the intersection points of $\hat \alpha_f$, we create the \ac{aabb} of the curve; \rev{an \ac{aabb} is the smallest axis-aligned hypercube that contains a given entity (see, e.g., \cite{ericson2004real}).} $\hat \alpha_f$ (which cannot touch the face boundary as proven above) splits the face into two parts, including the bottom-left and top-right vertices of this face, resp. A strictly decreasing curve in this face must connect the left-top and bottom-right edges intersecting $\hat \alpha_f$. However, this is not possible since $\hat\alpha_f$ and $\hat \alpha_f^\prime$ are disconnected components of a non-self-intersecting curve $\hat\gamma_f^0$ (since $\mapF$ is a diffeomorphism).}

\end{proof}

\begin{proposition}\label{prop:ref1}
{Given $f, f' \in \Lambda^2(K)$, $f \neq f'$, and connected components $\hat{\alpha}_f \in C_f(\hat{F}_R)$ and $\hat{\alpha}_{f'} \in C_{f'}$, satisfy the following properties in $\hat{F}_R \in \mathrm{ref}_1(\hat F)$,
  \begin{enumerate}[label=(\roman*)]
      \item $\hat \alpha_f$ intersects at most once an edge $\hat{\varepsilon} \in \Lambda^1(\hat F_R)$ and is not tangent to the edge;
      \item {if $\hat \alpha _f$ intersects $\hat F_R^\circ$, $\hat \alpha_f$ intersects twice $\partial\hat F_R$;}
      \item $\hat \alpha_f$ and $\hat \alpha_{f^\prime}$ do not intersect in $\overline{\hat{F}_R} \setminus \Lambda^0(\hat F_R)$;
      \item $\hat \alpha_f$ and $\alpha_{f^\prime}$ intersects at most once in $\overline{\hat{F}_R}$.
  \end{enumerate}
  }
\end{proposition}
\begin{proof}
{Statement (i) is a direct consequence of \prop{prop:monotonicity}. Let us prove (ii) by contradiction. We assume that there exists a $\hat \alpha_f\in C_f(\hat F_R)$, $f\in\Lambda^2(K)$, that does not intersect any edge $\hat{\varepsilon} \in \Lambda^1 (\hat F_R)$. $\hat \alpha_f$ is a closed curve in $\mathbb R ^2$. This curve is not monotonic, which is in contradiction with \prop{prop:monotonicity}. If we assume that $\hat \alpha_f\in C_f(\hat F_R)$, $f\in\Lambda^2(K)$ intersects only once $\partial \hat F_R$, then $\hat \alpha_f$ is tangent to an edge $\hat\varepsilon \in \Lambda^1(\hat F_R)$, which is in contradiction with (i). $\hat \alpha_f$ cannot intersect more than twice $\partial \hat F_R$ since $\hat\alpha_f$ is a connected and non-self-intersecting curve.}

All intersections are on vertices of refined faces by step (2) in $\mathrm{ref}_1$, which proves (iii). Step (3) in the refinement rule explicitly enforces (iv).

\end{proof}

We note that we have not provided yet an algorithm that determines $C_f(\hat{F}_R)$, $f \in \Lambda^1(\hat{F}_R)$, i.e., the $\hat{\alpha}$-curves after $\mathrm{ref}_1$. We can compute all the intersections (vertices) but it still remains open how to connect these vertices.  We develop the details of the connectivity of the intersection points in the next section. In any case, assuming that the $\hat{\alpha}$-curves are connected, we can define the following refinement rule $\mathrm{ref}_2$ to compute the intersection $\hat{F}_R\cap K$, where $\hat{F}_R \in \mathrm{ref}_1 (\hat{F})$. The $\mathrm{ref}_2$ is not \ac{aa} but it does not introduce new intersection points, as stated in \thm{th:clipping} (see \fig{fig:clipping-surface}).

\begin{algorithm}[H]
  \caption{$\mathrm{ref}_2$: partition by connecting intersections}\label{alg:split-connect}
  \justify
  Split $\hat F_R \in \mathrm{ref}_1(\hat F)$ through each edge $\alpha \in C_f(\hat{F}_R)$, $ \forall  f \in \Lambda^2(K)$.
\end{algorithm}

\begin{theorem}\label{th:clipping}
    The previous refinement rule is such that $\hat F_R \in \mathrm{ref}_2\circ\mathrm{ref}_1(\hat F)$ has the following properties:
    \begin{enumerate}[label=(\roman*)]
      \item $\Lambda^1( \hat F_R)$ is composed of \ac{aa} edges and strictly monotonic nonlinear edges (including potential diagonal edges);
      \item The vertices of $\hat F_R$ lay on the boundary of its \ac{aabb};
      \item All the nonlinear edges are strictly increasing or all strictly decreasing;
      \item $\hat F_R$ is diffeomorphically equivalent to a convex linear polygon $P^\mathrm{ln}$;
    \end{enumerate} 
\end{theorem}
\begin{proof}
  {We note that the refinement $\mathrm{ref}_1$ is a $kd$-tree partition of the square that contains the reference triangle. Let us consider the set of \acp{aabb} obtained by this refinement. By construction of the partition, $\hat F_R\in\mathrm{ref}_2\circ\mathrm{ref}_1(\hat F)$ is the clipping of an \ac{aabb} by strictly monotonic curves that do not intersect among themselves (see \prop{prop:monotonicity} (ii) and \prop{prop:ref1}). Thus, (i) and (ii) readily hold, while (iii) holds due to \prop{prop:monotonicity} (iii). 

  The nonlinear polygon can be represented as a cyclic graph, where the vertices are the intersection points of the edges. We can replace the nonlinear edges with linear ones connecting the vertices, creating a closed linear polygon. Since the nonlinear edges are smooth, the vertices are preserved, and the topology of the surface is preserved, the linear polygon is diffeomorphic to the nonlinear one. Thus, (iv) holds.}
\end{proof}

After refining $F$,
the $\hat \gamma^0$-curves do not intersect $\hat F^\circ_R$, i.e., they belong to $\partial \hat F_R$, for $\hat F_R \in \mathrm{ref}_2\circ\mathrm{ref}_1(\hat F)$.
Thus, $F_R^\circ$ can only be inside or outside $K$.
Now, we can classify the faces with respect to $K$ as {interior},
\begin{equation}
  \label{eq:Fcut}
  \mathcal{F}^\mathrm{cut} \doteq \{ \ \hat{F}_R \in \mathrm{ref}_2\circ\mathrm{ref}_1(\hat{F}) \ :\  \hat F_R \subseteq \bigcap_{f\in\Lambda^2(K)} \mathrm{int} (\hat \gamma_f)  \ \},
\end{equation}
and {exterior, $\mathcal{F}^\mathrm{ext} \doteq \{ \ \hat{F}_R \in \mathrm{ref}_2\circ\mathrm{ref}_1(\hat{F}) \} \setminus \mathcal{F}^\mathrm{cut}$}.
Remember that $K$ is not open or closed in general, as stated at the beginning of this section. Thus,  $\mathrm{int}(\hat \gamma_f)$ (where $\mathrm{int}(\cdot)$ denotes the interior), $f\in\Lambda^2(K)$ in \eqref{eq:Fcut} can be either open or closed (see more details in \cite{Badia_2022-stl}).
Once the refined faces are classified, we define \alg{alg:intersect_F_K} to obtain $F\cap K$ as a partition into nonlinear polygons (see \fig{fig:clipping-surface}).
{We note that, due to \thm{th:clipping}, the linearization of $\mathcal F^\mathrm{ref}$ can be represented with a cyclic graph (see \cite{Badia_2022-stl}). Thus, $\mathcal F^\mathrm{cut}$ can utilize such representation.}

\begin{algorithm}[H]
  \caption{$F\cap K$}\label{alg:intersect_F_K}
  \justify
  Compute the partition $\mathcal{F}^\mathrm{ref} \doteq \mathrm{ref}_2\circ\mathrm{ref}_1(F)$. \\
  Return the faces $\mathcal{F}^\mathrm{cut}$ in $\mathcal{F}^\mathrm{ref}$ that are inside $K$ (see \eqref{eq:Fcut}).
\end{algorithm}

\begin{figure}
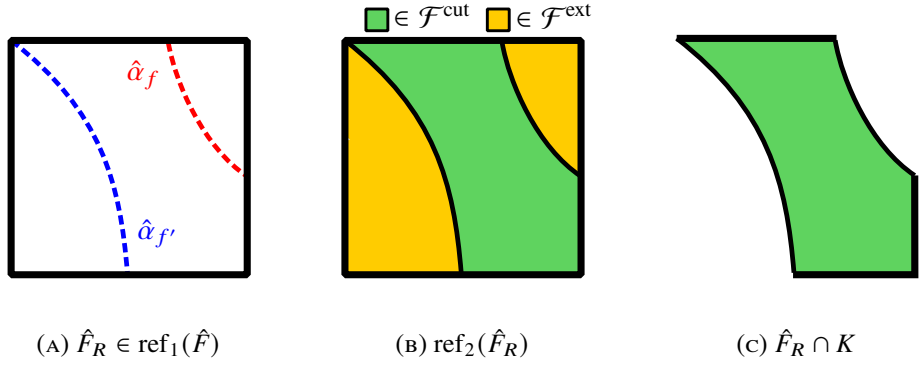

  \centering
  \includefig[\normalsize]{0.15\textwidth}{refinement-connection_legend}

  \begin{subfigure}[b]{0.2\textwidth}
    \includefig[\normalsize]{\textwidth}{refinement-connection_1_1}
    \caption{$\hat F_R \in \mathrm{ref}_1(\hat F)$}
  \end{subfigure}
  \hspace{1cm}
  \begin{subfigure}[b]{0.2\textwidth}
    \includefig[\normalsize]{\textwidth}{refinement-connection_2_1}
    \caption{$\mathrm{ref}_2(\hat F_R)$}
  \end{subfigure}
  \hspace{1cm}
  \begin{subfigure}[b]{0.2\textwidth}
    \includefig[\normalsize]{\textwidth}{refinement-connection_3}
    \caption{$\hat F_R\cap K$}
  \end{subfigure}
  \caption[Representation of clipping algorithm $F\cap K$.]{Representation of clipping algorithm $F\cap K$. In (a), we start with a refined face $\hat F_R\in \mathrm{ref}_1(\hat F)$ such that the $\hat \alpha$-curves are strictly monotonic with respect to the axes in the reference space and do not intersect in $\hat F_R\setminus \Lambda^0(\hat F_R)$.
  Here, the $\hat \alpha$-curves are $\hat \alpha_f \in C_f(\hat F_R)$ (red) and $\hat \alpha_{f^\prime} \in C_{f^\prime}(\hat F_R)$ (blue) for $f,f^\prime \in \Lambda^2(K)$.
  Then, in (b), we generate a partition $\mathrm{ref}_2(\hat F_R)$
  through the $\hat \alpha$-curves, in which we classify the faces with respect to $K$.
  Finally, in (c), we restrict
  $\mathrm{ref}_2(\hat F_R)$ inside $K$ (see \alg{alg:split-connect}). } 
  \label{fig:clipping-surface}
\end{figure}

\begin{remark}\label{rmk:ill-posed}
  For linear F and intersection curves, the \ac{aa} critical point computation is ill-posed. However, there is no need to add any point in this case, since the intersection is linear.  The curve-plane intersection is ill-posed if the curve is contained in the plane (linear edge); analogously for the surface-line intersection (linear face). These intersections can also be disregarded because they are not affecting the \( \mathrm{meas}_2 \) of sets after trimming.
\end{remark}

\begin{remark}\label{rmk:coarsening}
  The partition generated by the refinement above may contain a level of refinement that is not required to fulfill all the requirements for \prop{prop:ref1} to hold. Thus, one can perform a coarsening step of $\mathcal F^\mathrm{ref}$ that reduces the non-essential vertices, edges, and faces, while maintaining a diffeomorphically equivalent surface.
\end{remark}

\begin{remark}\label{rmk:accuracy}
  The presence of \ac{aa} critical points can improve the accuracy of a parametrization step in \sect{sec:surf-approx}, since the resulting curves are strictly monotonic with respect to the edges of the reference space. Additionally, one can consider another refinement of the $\gamma^0_f$-curves to bound the relative chord $\hat \delta_{\max}$ up to a given threshold.
\end{remark}

\subsection{Connection algorithm}\label{sec:connection-algorithm}

Let us consider a face
{$F\in\mathrm{ref}_1(F^0)$} after the refinement in \alg{alg:aa-refinement}.
We stress that $F$ hereafter denotes a refined face
{and ${F^0}$ the original face before refinement in the previous section.}
We remark that the reference space is always the one of the original face, i.e., the map $\mapF$ is inherited from the unrefined face.
Thus, the curves $\gamma_f^0$ are independent of the refinement.

As above, let us consider the set $C_f(\hat F)$ of mutually disconnected components of $\hat \gamma_f^0 \cap \overline{\hat F}$ (now at the refined face).
We omit the face sub-index, i.e., we use $C(\hat F)$, to represent the union of all these sets for all faces $f \in \Lambda^2(K)$.
In the reference space $\hat F$, the curves $\hat \alpha _f \in C_f(\hat F)$ are strictly monotonic smooth curves in $\hat F \setminus \Lambda^0(\hat F)$ (see \prop{prop:monotonicity}); for brevity, we refer as monotonic the curves that are monotonic with respect to $(\hat \xi_1,\hat \xi_2)$.
{In \prop{prop:monotonicity}, we prove that all the curves in $C_f(\hat F)$ are strictly increasing or all curves are strictly decreasing.}
Moreover, two $\hat \alpha$-curves, $\hat \alpha_f \in C_f(\hat F)$ and $\hat \alpha_{f^\prime} \in C_{f^\prime}(\hat F)$, $f,f^\prime \in \Lambda^2(K)$, cannot intersect in $\hat F \setminus \Lambda^0(\hat F)$ (see \prop{prop:ref1}). Let us define the set of intersection points:
\begin{equation}
  I = \{ \hat \gamma _ { f}^0 \cap \partial \hat F  \ : \ f \in \Lambda^2(K) \}.
\end{equation}
There is an injective map between the intersection points in $I$ and the $\alpha$-curves. We assume that $I$ is composed of isolated points, i.e., $\hat \gamma _ { f}^0 \cap \partial \hat F$ is not an edge of $\hat{F}$. This singular case can readily handled, as discussed in \rmk{rmk:ill-posed}.

By construction (clipping of the reference triangle with AA lines), $\hat{F}$ has four edges and the bottom-left corner connects two AA edges. Thus, we can label the edges of $\Lambda^1(\hat F)$ as $\hat \varepsilon_1, ..., \hat \varepsilon_4$ in the counterclockwise direction, starting with the leftmost vertical edge. We note that $\hat\varepsilon_3$ can be non-\ac{aa} and $\hat\varepsilon _4$ empty, but it does not affect the discussion. The refinement rules ensure that the non-\ac{aa} edges are not intersected by $\hat \gamma_f^0$.
One can simply consider the \ac{aabb} that contains the face, which cannot affect the intersection points, since the $\hat\alpha$-curves only connect \ac{aa} edges that do not change by this step.
Thus, we can assume that all edges are \ac{aa} without loss of generality.

Let us consider $C_f(\hat F)$ to be composed of strictly increasing curves. We define $\Gamma^+ \doteq \hat\varepsilon_1   \cup \hat\varepsilon _2$ and $\Gamma^- \doteq \partial \hat F \setminus \Gamma ^{+}$ and $I^{+} \doteq I \cap \Gamma^{+}$ and $I^{-} \doteq I \cap \Gamma^{-}$. We proceed analogously if $C_f(\hat F)$ is composed of strictly decreasing curves, but defining $\Gamma^{+} \doteq \hat\varepsilon_ 1   \cup \hat\varepsilon _4$.  

\begin{proposition} \label{prop:connection-by-sets}
Every node in $I^{+}$ is connected to one and only one node in $I^{-}$. The opposite is also true. Thus, we can connect all the nodes in $I^{+}$ and $I^{-}$.
\end{proposition}
\begin{proof}   
  Let us consider the case in which $C_f(\hat F)$ is composed of strictly increasing curves. A node in $I^{+}$ cannot be connected to a node in $I^{+}$ since one node in $I$ belongs to only one curve, the curves are strictly increasing, and a strictly increasing function that starts in $\varepsilon_1$ or $\varepsilon_2$ cannot intersect again $\varepsilon_1$ or $\varepsilon_2$.
  Using an analogous argument, we can prove the result when $C_f(\hat F)$ is composed of strictly decreasing curves.
\end{proof}

\begin{proposition} \label{prop:connect-cyclic}
If $I^{+}$ and $I^{-}$ are sorted clockwise and anti-clockwise, resp., then each node in $I^{+}$ can only be connected to the node in the same position in $I^{-}$.
\end{proposition}
\begin{proof}
Let us assume that $i^{+}_1 \in I^{+}$ is connected to $i^{-}_k$, $k \ne 1$. The line that connects $i^{+}_1$ and $i^{-}_k$ split $\hat F$ into two non-empty parts. We note that this is due to the fact that $i^+_1$ and $i^+_k$ cannot be on the same edge.
One of the parts contains  nodes $i^{-}_1, ..., i_{k-1}^{-} \in I^{-}$ and the other part contains $I^{+} \setminus i_1^{+}$.
By \prop{prop:connection-by-sets}, nodes contained in the two different regions are connected by $\alpha$-curves. But these $\alpha$-curves then intersect the one connecting $i^{+}_1$ and $i^{-}_k$. This is not possible, since these are disconnected components of a non-self-intersecting curve. Thus, $i^{+}_1$ can only be connected to $i^{-}_1$. We proceed analogously for the other nodes.
\end{proof}

This way, we compute all the connections in the graph with \alg{alg:vertex-connection}. First, in \alglin{alg:vertex-connection-1}, we compute the intersections between $\hat\gamma^0_f$ and $\partial \hat F$ using intersection algorithms (see \fig{fig:vertex-connection-a}).
In \alglin{alg:vertex-connection-2}, we return the straightforward connection of two points. Otherwise, in \alglin{alg:vertex-connection-4},
we compute the derivative sign of the $\alpha$-curves in $C_f(\hat F)$ (e.g., evaluating the gradient $\pmb{\nabla}(\hat \gamma_f^0)$ at any intersection point $i\in I$).
Next, we connect%
\footnote{In \alglin{alg:vertex-connection-10}, we use the notation \texttt{zip} to iterate simultaneously over multiple iterators of the same length, e.g.,  $\pmb{a}$ and $\pmb{b}$. Each iteration returns a tuple of the $i^\mathrm{th}$ value of each interator, $(a_i, b_i)$.}
the nodes in the sorted sets $I^{+}$ and $I^{-}$ according to \prop{prop:connect-cyclic}  (\alglin{alg:vertex-connection-8}-\ref{alg:vertex-connection-10}).
The algorithm returns a set of edges $\mathcal E$ that connect the intersections $I$. These connections define entirely the $\hat\alpha$-curves, which are then used in \alg{alg:split-connect}. We note that the full connection algorithm is barely used in practice if we start with a surface mesh with a small chordal error of its linearisation.

\begin{algorithm}
  \caption{$\mathtt{connect}(\hat F,\hat  \gamma_f)$}\label{alg:vertex-connection}
  \begin{algorithmic}[1]
    \STATE $I \gets \{ \hat \gamma^0_f \cap \partial \hat F \}$
    \label{alg:vertex-connection-1}
    \IF{ $\mathtt{length}(I) = 2$}
      \RETURN $\mathcal E \gets \{ i_1,i_2 \in I \}$\label{alg:vertex-connection-2}
    \ENDIF
    \STATE $s \gets \mathtt{derivative\_sign}(\gamma^0_f)$
    \label{alg:vertex-connection-4}
        \STATE $I^{+},I^{-} \gets \mathtt{split\_by\_position}(I,s)$
        \label{alg:vertex-connection-8}
        \STATE $I^{+}, I^{-} \gets \mathtt{cyclic\_sort\_by\_sign}(I^{-},I^{+},s)$
        \label{alg:vertex-connection-9}
        \RETURN $\mathcal E \gets \{\ (i^+,i^-) \in \mathtt{zip}(\ I^{+} ,\ I^{-})\ \}$
        \label{alg:vertex-connection-10}
  \end{algorithmic}
\end{algorithm}

\begin{figure}
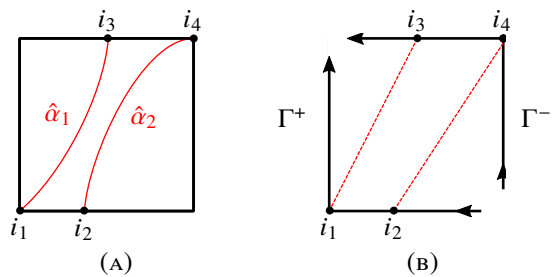

  \centering
  \vspace{0.1cm}
  \begin{subfigure}[b]{.19\textwidth}
    \centering
    \includefig[\small]{.8\textwidth}{conn_1_1}
    \vspace{0.2cm}
    \caption{}
    \label{fig:vertex-connection-a}
  \end{subfigure}
  \hspace{0.05\textwidth}
  \begin{subfigure}[b]{.19\textwidth}
    \centering
    \includefig[\small]{.8\textwidth}{conn_2_1}
    \vspace{0.2cm}
    \caption{}
    \label{fig:vertex-connection-b}
  \end{subfigure}
  \caption[Representation of \alg{alg:nonlinear:vertex-connection} for strictly increasing curves.]{Representation of \alg{alg:vertex-connection} for strictly increasing curves.
  In (a), we present the intersections $i_1,...,i_4 \in I$ of the curves $\hat\alpha_1,\hat\alpha_2 \in C_f(\hat F)$ with $\partial \hat F$. These intersections are classified into $I^+=\{i_3,i_4 \}$ and $I^-=\{i_1,i_2\}$ according to their position in the boundary, $\Gamma^+$ or $\Gamma^-$, resp. These sets are sorted in clockwise and anticlockwise order, resp. This sorting leads to the connection shown in (b) (dashed red). If the curves in  $C_f(\hat F)$ were strictly decreasing, the figures would be symmetric to these ones.  }
  \label{fig:vertex-connection}
\end{figure}

\begin{remark}\label{rmk:filter-intersections}
  In the computation of intersection points in \alg{alg:vertex-connection}, we disregard the ones that do not logically represent an intersection. E.g., if the level set value $\gamma_f$ does not change the sign around a vertex with zero level set value.
\end{remark}

One can observe that for this algorithm to work, $I$ must have an even number of elements (after filtering $I$ with \rmk{rmk:filter-intersections}).
It is true since the $\hat \alpha$-curves are monotonic. Thus, any $\hat \alpha_f\in C_f(\hat F)$ that $\hat \alpha _f \cap \hat F\setminus\Lambda^0(\hat F) \ne \emptyset$ can only intersect twice $\partial \hat F$ (see \thm{th:clipping}).

\subsection{Surface partition}\label{sec:surface-partition}
Let us consider a nonlinear general polygon $\hat F$ in $\mathbb{R}^2$, e.g., $F\in \mathcal F^\mathrm{cut}$. In order to parametrize $F$, we need a partition into regular polygons, e.g., simplices or quadrilaterals.
This parameterization is needed to compute bulk and surface integrals when using (\ref{eq:boundary-integral-simplexify}). However, this is not needed when transforming these integrals into edge integrals using (\ref{eq:boundary-integral-edge}).

Let us recover the definition of a kernel point. A kernel point can be connected by a segment to any other point of a polytope without intersecting its boundary. The union of all possible kernel points is the kernel polytope, which is convex. The kernel polytope of a convex polytope is itself \cite{Ziegler1995}. In a linear polytope $P^\mathrm{ln}$, we can compute a kernel point by finding a point that belongs to the half-spaces defined by the faces of $P^\mathrm{ln}$; see more details in \cite{Sorrente_2021}.

Finding the kernel of a nonlinear polytope $\hat F$ is not trivial. One can find a lower bound of the kernel polytope with the convex hull of the nonlinear faces.
If a kernel point exists for a given polytope $\hat F$, we can compute a simplex partition of $\hat F$ using linear edges. However, the kernel point does not exist in general.

\begin{proposition}\label{prop:partition}
If $\hat F$ has the properties of \thm{th:clipping} then $\hat F$ can be partitioned into a set of nonlinear triangles and quadrilaterals by adding linear edges only.
\end{proposition}
\begin{proof}
  Let us assume that $\Lambda^1(\hat F)$ is composed of two strictly increasing nonlinear curves and a set of \ac{aa} edges. According to \thm{th:clipping}, these curves and edges can only intersect at the boundary of the \ac{aabb} of $\hat F$. Thus, we connect the endpoints of the nonlinear curves with non-increasing linear edges. Since the new edges are non-increasing, they can only intersect once each curve, i.e., at the endpoints. This connection generates a nonlineal quadrilateral $\hat Q\subseteq \hat F$. The remaining parts $\hat F \setminus \hat Q$ are linear and convex (see \thm{th:clipping}), in which a simplex partition is straightforward.
\end{proof}
Therefore, we can compute a hybrid partition if a nonlinear polytope $\hat F$  has no kernel point
We note that the nonlinear quadrilaterals of \prop{prop:partition} are composed of two nonlinear edges and two linear edges. Thus, they can be represented with a tensor product map with anisotropic order, e.g., $q \times 1$ where $q$ is the order of the nonlinear edges. If needed, we can decompose the nonlinear quadrilaterals into triangles within their reference space.
\begin{remark}
  We can consider a coarsening of $\mathcal F^\mathrm{cut}$ that satisfies \prop{prop:ref1} (see \rmk{rmk:coarsening}) if the resulting polygons have a kernel point. We note that the coarsening stated in \rmk{rmk:coarsening} can lead to non-convex polygons.
\end{remark}
\subsection{Surface parametrization}\label{sec:surf-approx}
The intersection $\mathcal{B}^\mathrm{cut}_K \doteq \mathcal{B} \cap K$ is a nonlinear surface mesh
in which the intersection curves are implicitly represented.
Thus, we need a parametrization of the edges and surfaces for further operations.
We utilize a least-squares method \cite{Borges_2002} combined with a sampling strategy \cite{Fries_2015}. This iterative process converges to an accurate parametrization of the intersection curves and trimmed surfaces.

In the least-squares method we can find the B\'ezier control points $X^b=\{x_j^b\}_{j=1}^m$  that minimize $X^\ell - \mathbf{B} X^b$ where $\{\mathbf{B}\}_{ij} = b_j(\xi _i)$ are the B\'ezier basis evaluated at the reference points $\{\xi_j\}_{j=1}^n$ and $X_\xi= \{ x^\ell_j \}_{j=1}^n$ represents the set of sampling points (see \cite{Schmidt_2012} for more details).
When $\mathbf{B}$ is not a square matrix, i.e., $n>m$, $X_b$ is approximated through a linear least-squares operation. To isolate the approximation between the $d$-faces, one can recursively perform a least-squares operation on the interior points, increasing the dimension.
Note that for square matrices, i.e., $n=m$, we are building the B\'ezier extraction operator \cite{Borden_2010}.

\citet{Fries_2015} discussed several sampling strategies and demonstrated similar convergence in \ac{fe} analysis. In each of the sampling strategies, they solve a nonlinear problem for every sampling point. In our case, we aim to parametrize the intersection curves $F\cap G, \ F\in \mathcal B, \ G \in \Lambda^2(K)$.
Thus, we define a sampling strategy to represent the intersection curves
by solving a closest point problem.
Specifically, we find
\[
  \hat{\pmb{\xi}} \in \hat F \ : \  (\pmb{x}_0 - \pmb{\phi}_F(\hat{\pmb{\xi}})) \cdot \pmb{n}_\pi = 0, \qquad  (\pmb{x}_0 - \pmb{\phi}_F(\hat{\pmb{\xi}})) \cdot ( \pmb{n}_\pi \times (\pmb{\partial}_{\hat{\xi}_1}(\mapF(\hat{\pmb{\xi}})) \times \pmb{\partial}_{\hat{\xi}_2}(\mapF(\hat{\pmb{\xi}}))) ) = 0.
\]
When we compute the closest point to a surface, e.g., when parametrizing the interior of a nonlinear face, the problem reads as follows: find
\[
  \hat{\pmb{\xi}} \in \hat F \ : \qquad  (\pmb{x}_0 - \pmb{\phi}_F(\hat{\pmb{\xi}})) \cdot \pmb{\nabla}(\pmb{\phi}_F(\hat{\pmb{\xi}})) = \pmb{0}.
\]
This algorithm requires a seed point in the physical space $\pmb{x}_0$, which does not belong to the surface or curve, and an initial reference point $\hat{\pmb{\xi}}_0 \in \hat F$. This sampling strategy, as well as others, assumes relatively small curvature in the curves and surfaces. The curvature has already been bounded in \sect{sec:surface-trimming}.

The least-squares method with a sampling strategy is insufficient for an accurate parametrization, see the example of \fig{fig:draw-surf-approx-2}. The authors in \cite{Borges_2002} propose a method to optimize the reference points $\{\xi_j\}_{j=1}^n$ iteratively. We use a similar approach that fixes the reference points and recomputes the sampling points at each iteration. This approach allows us to compute the least-squares operator $\mathbf{B}^+$ only once. In addition, the sampling will be evenly spaced in the reference space. The example in \fig{fig:draw-surf-approx} shows how the approximation improves with the iterations. This method applies to the parametrization of curves and surfaces.

\begin{figure}[http]
  \centering
  \includefig[\normalsize]{0.8\textwidth}{draw_approx_legend}

  \begin{subfigure}{0.24\textwidth}
    \includegraphics[width=\textwidth]{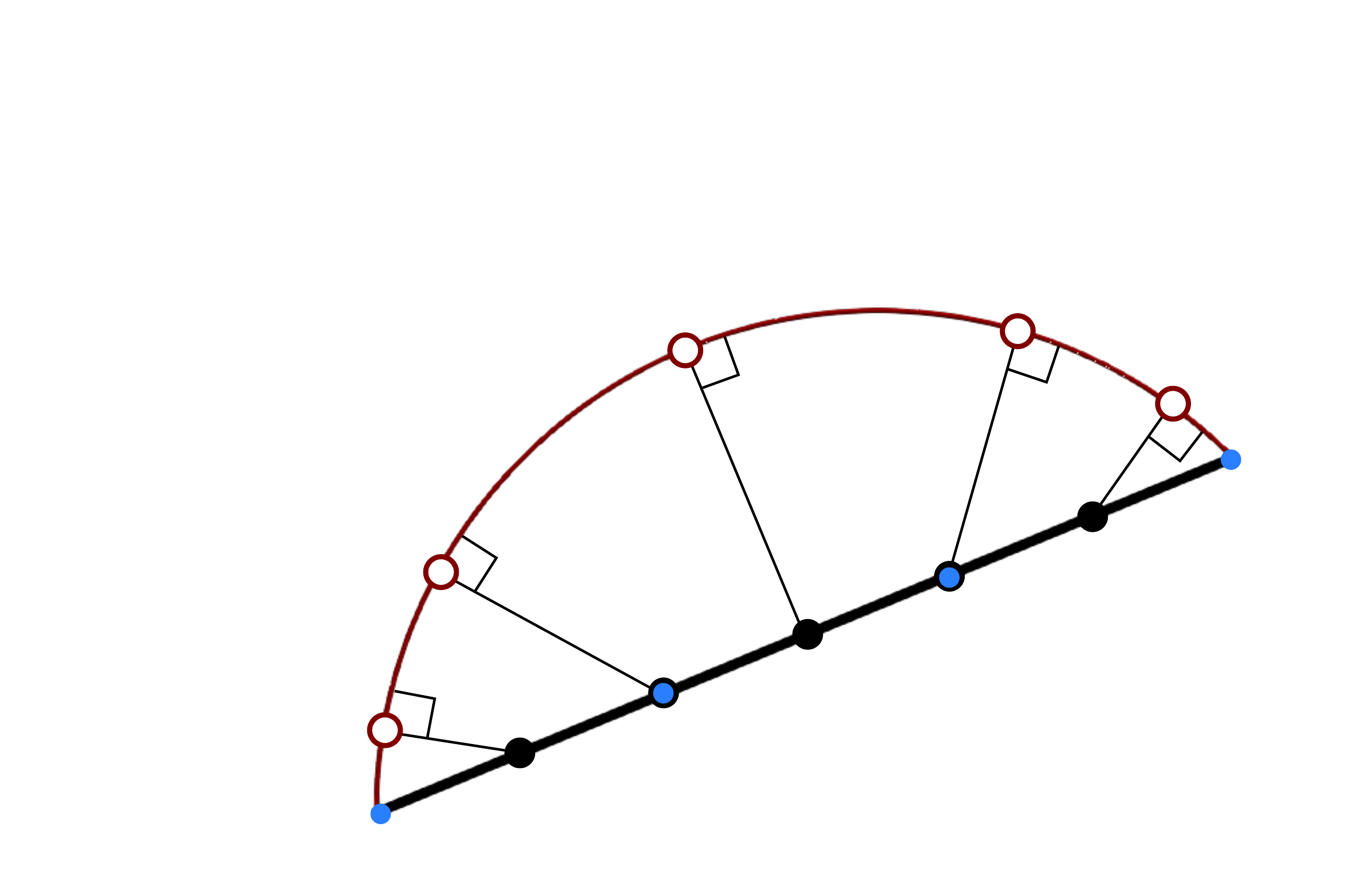}
    \caption{$\mathrm{it} = 0$}
  \end{subfigure}
  \begin{subfigure}{0.24\textwidth}
    \includegraphics[width=\textwidth]{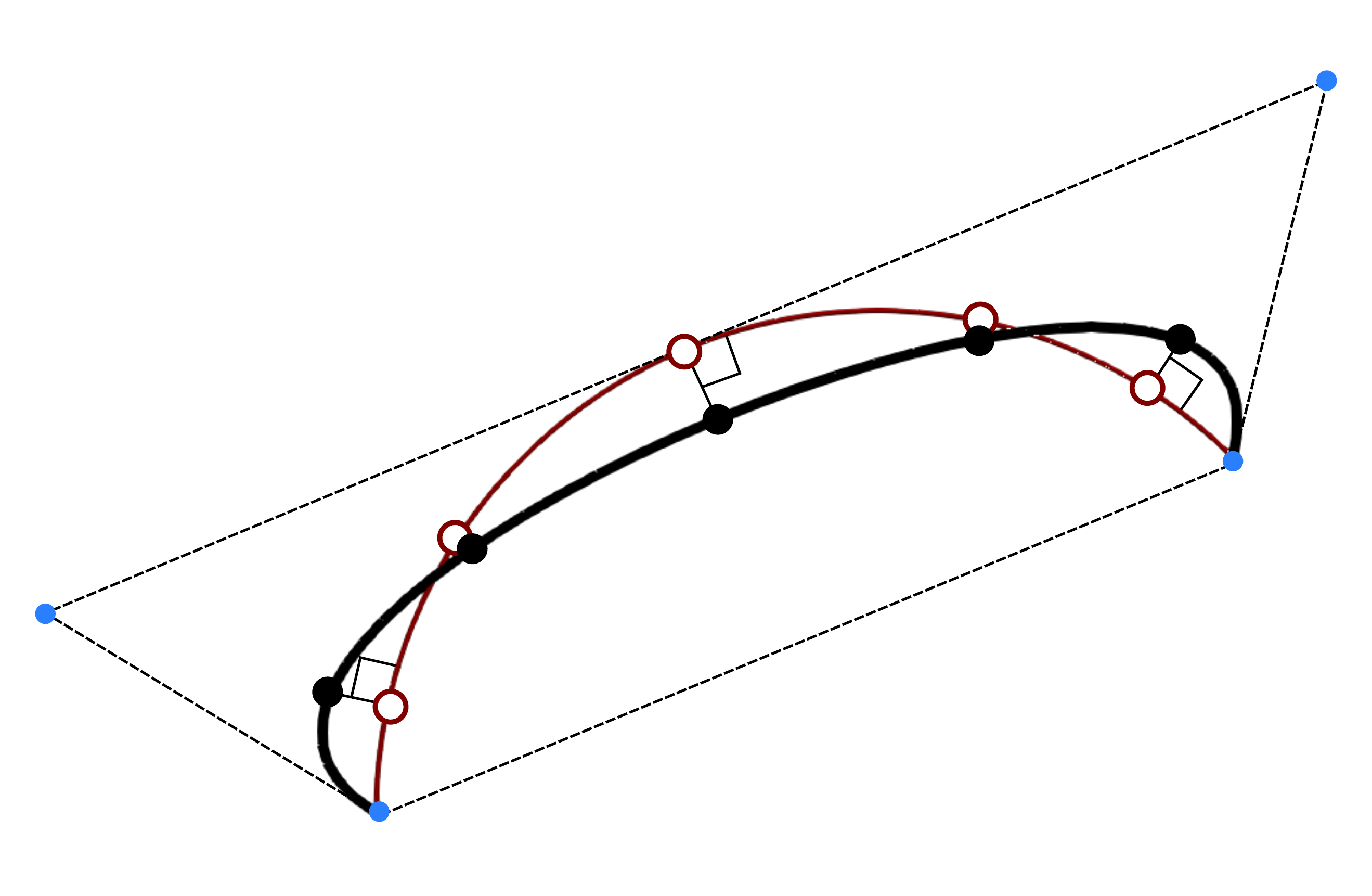}
    \caption{$\mathrm{it} = 1$}
    \label{fig:draw-surf-approx-2}
  \end{subfigure}
  \begin{subfigure}{0.24\textwidth}
    \includegraphics[width=\textwidth]{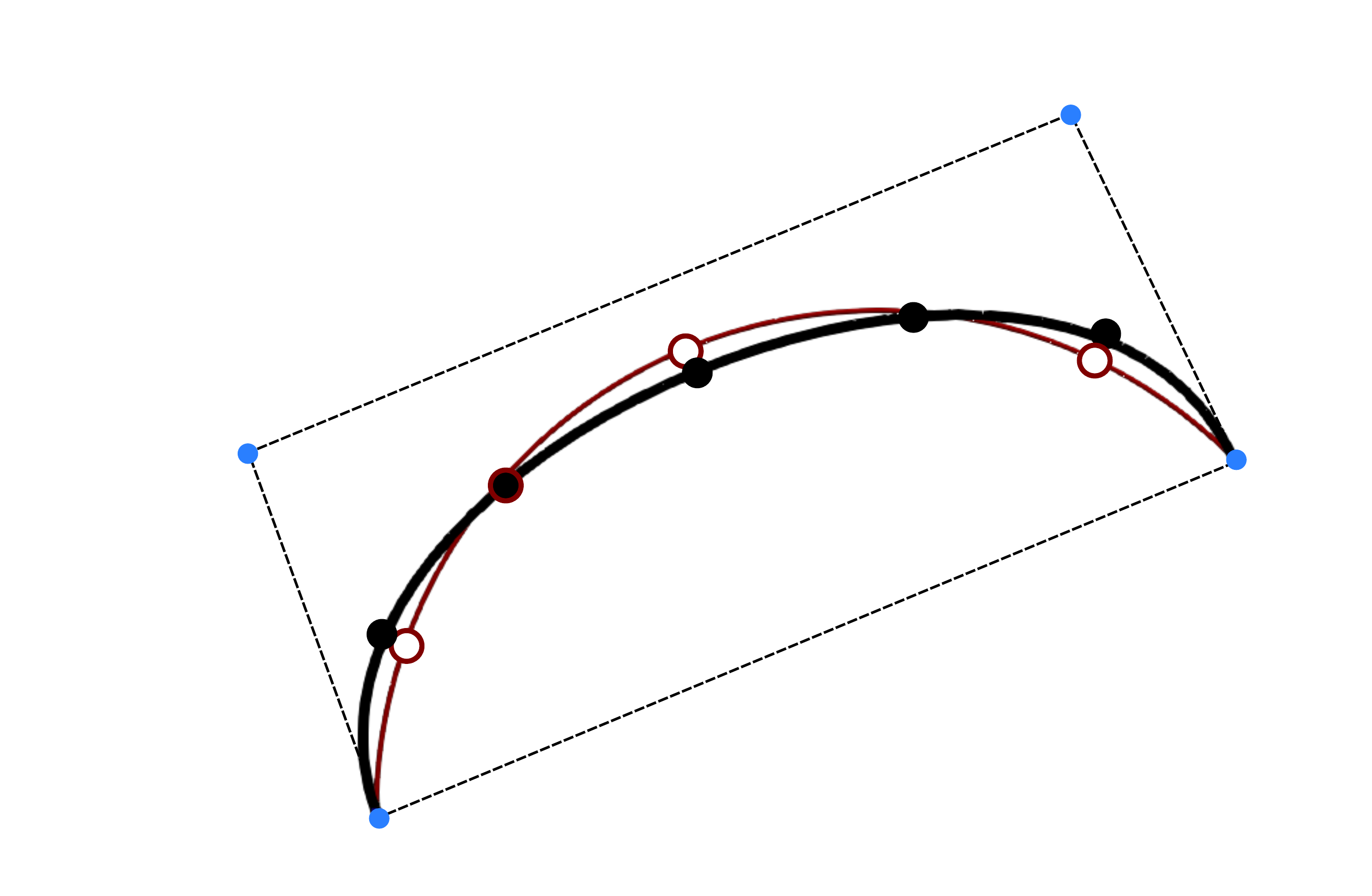}
    \caption{$\mathrm{it} = 5$}
  \end{subfigure}
  \begin{subfigure}{0.24\textwidth}
    \includegraphics[width=\textwidth]{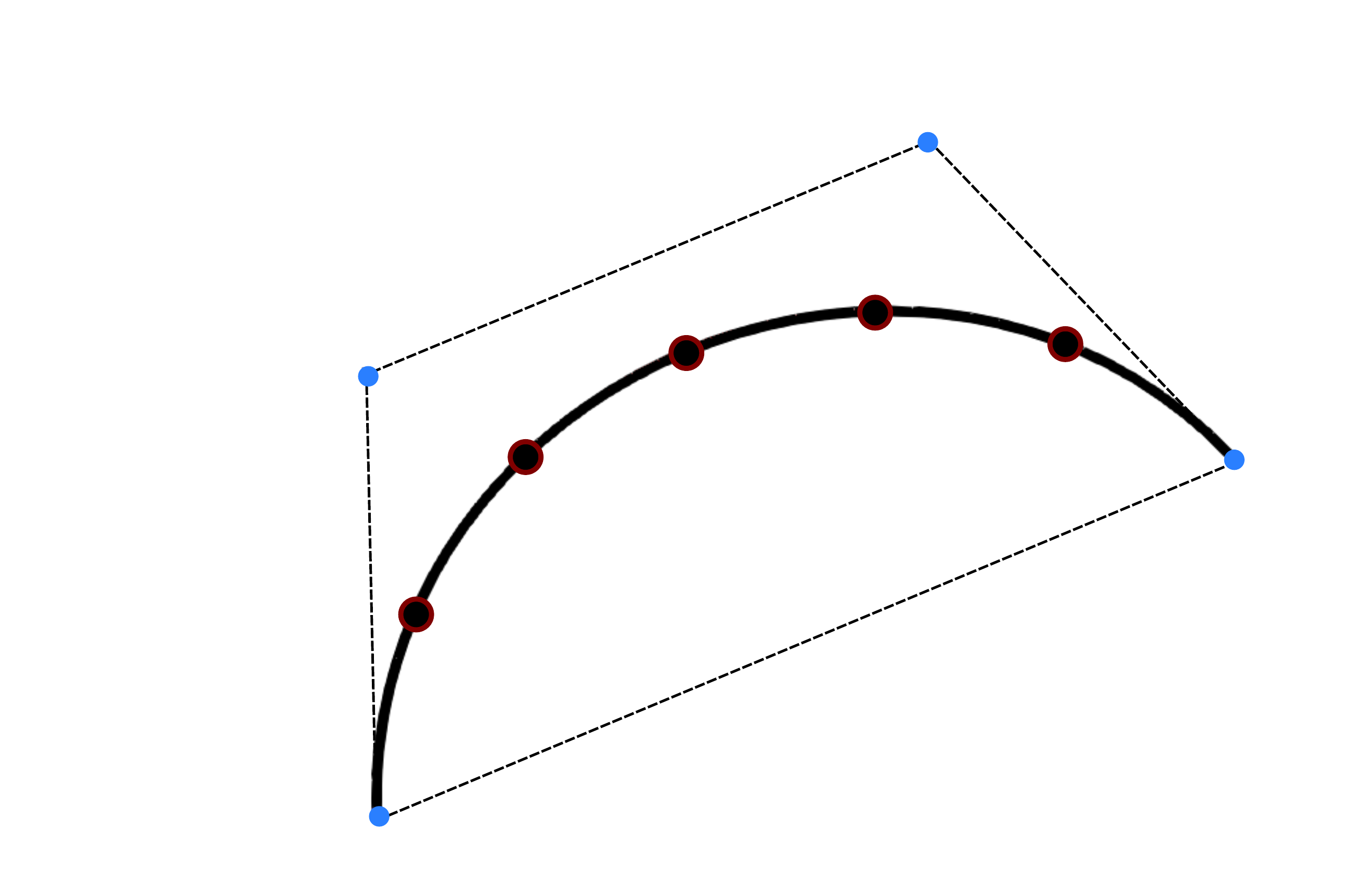}
    \caption{$\mathrm{it} = 15$}
  \end{subfigure}
  \caption[Example of iterative approximation of a curve.]{Example of iterative approximation of a curve. From the linear approximation in (a), we sample the closest points. These sampling points $X^\ell$ are approximated in a B\'ezier curve $\mathcal X^b$ in (b). In (b), we evaluate the approximated seed points $\mathcal T ^b(\mathcal X_\xi)$ to compute the new sampling points. Across the iteration, in (c) and (d), the approximation improves by reducing the distance between the sampling points and the curve. In (d), the sampling points are evenly spaced in the reference space (up to a stopping criterion). Note that this example corresponds to a particular case of the \alg{alg:surf-approx} in which the approximation degree is fixed to $p=3$.}
  \label{fig:draw-surf-approx}
\end{figure}

The parameterization method is described in \alg{alg:surf-approx}. Even though this algorithm is designed for $\mathcal B^\mathrm{cut}_K$, it is general to any nonlinear mesh, namely $\mathcal S$. We first generate a linear approximation of $\mathcal S$ in \alglin{alg:surf-approx-1} with a simplex partition (or a partition into standard polytopes, see \sect{sec:surface-partition}). In
\alglin{alg:surf-approx-2}, we initialize the sampling points through the degree elevation, see \cite{Engvall_2017}. Then, we parametrize the $d$-faces from lower to higher dimension \alglin{alg:surf-approx-3}. Each $d$ is parametrized from lower to higher order \alglin{alg:surf-approx-6}, in simplices, we start with $p^0=d$. The gradual increment of the approximation error improves the convergence. In each iteration, we compute a B\'ezier mesh with the least-squares method then we compute the sampling points from the evaluations of the B\'ezier mesh. When the variation of the error estimator is below a tolerance, we increase the degree. Finally, we return the B\'ezier mesh with the least-squares method of the highest degree.

\begin{algorithm}
  \caption{$\mathtt{parametrize}(\mathcal T,D)$}\label{alg:surf-approx}
  \begin{algorithmic}[1]
    \STATE $ \mathcal T ^\mathrm{lin} \gets \mathtt{simplexify}(\mathcal S,D) $ \label{alg:surf-approx-1}
    \STATE $\mathcal T ^\ell \gets \mathtt{elevation}(\mathcal T ^\mathrm{lin},p^\ell)$\label{alg:surf-approx-2}
    \FOR{ $d  \in \{1,\ldots,D\} $} \label{alg:surf-approx-3}
      \FOR{$p \in \{p^0,\ldots,p^{b}\}$} \label{alg:surf-approx-6}
      \STATE $e^- \gets \infty$; \quad $\Delta e \gets \infty$
      \WHILE{$\Delta e > \epsilon$  }
        \STATE $\mathcal T ^b \gets \mathbf{B}^+_{p} \cdot \mathcal T ^\ell_d$  \label{alg:surf-approx-7}
        \STATE $\mathcal T ^\ell_d \gets \mathtt{sample}(\mathcal T ^b, \mathcal T ^\ell_d, \mathcal S)$  \label{alg:surf-approx-8}
        \STATE $e \gets \| \mathcal T ^b - \mathcal T ^\ell_d  \|_{\ell^2}$; \quad $\Delta e \gets | e - e^- | / e$;\quad $e^- \gets e$ \label{alg:surf-approx-9}
      \ENDWHILE
    \ENDFOR
    \ENDFOR
    \RETURN $\mathbf{B}^+_{p^b} \cdot \mathcal T ^\ell$
  \end{algorithmic}
\end{algorithm}

\subsection{Cell intersection}\label{sec:cell-int}

The intersection of a cell with the domain $K^\mathrm{cut}\doteq K \cap \Omega$ can be represented using its boundary $\partial K ^\mathrm{cut} \doteq (\partial K)^\mathrm{cut} \cup \mathcal B^\mathrm{cut}_K$, as stated in \sect{sec:geo-fe}.
In \alg{alg:boundary-intersection} , we aim to compute $(\partial K)^\mathrm{cut} \doteq \partial K \cap \Omega = \partial K \cap \mathrm{int}(\mathcal B^\mathrm{cut}_K)$ using the linear algorithms from \cite{Badia_2022-stl}. Here, $\mathrm{int}(\mathcal B^\mathrm{cut}_K)$ represents the domain bounded by $\mathcal B^\mathrm{cut}_K$. For this purpose,  we first need a linearized surface $\mathcal B^\mathrm{lin}_K \doteq \mathtt{lin}(\mathcal B_K^\mathrm{cut})$ partitioned into simplices (see \alglin{alg:boundary-intersection-1} and \fig{fig:boundary-intersection-a} to \fig{fig:boundary-intersection-b}).
Since $\mathcal B_K^\mathrm{cut}$ is composed of nonlinear polytopes that are diffeomorphically equivalent to linear and convex polytopes, both linearization and simplex decomposition are straightforward.

The main algorithm in \cite{Badia_2022-stl} provides a linearized partition for $K^\mathrm{cut}$, i.e., $\mathcal T^\mathrm{lin}_K \doteq K \cap \mathcal B^\mathrm{lin}_K$ in \alglin{alg:boundary-intersection-2} and \fig{fig:boundary-intersection-c}.
From $\mathcal T^\mathrm{lin}_K$, we can obtain $(\partial K)^\mathrm{lin}$
by filtering the faces that belong to $\partial K$ in \alglin{alg:boundary-intersection-3}.
Additionally, we merge cells $P\in(\partial K)^\mathrm{lin}$ such that $\Lambda^1(P)$ is composed of cell edges $\varepsilon_K \subseteq \varepsilon \in \Lambda^1(K)$ and surface edges $\varepsilon_{\mathcal B}\in \Lambda^1(\mathcal B_K^\mathrm{lin})$. It is important to note that the surface edges $\varepsilon_{\mathcal B}\in \Lambda^1(\mathcal B_K^\mathrm{lin})$ have bijective map to $\varepsilon_{\mathcal B}^\prime\in \Lambda^1(\mathcal B_K^\mathrm{cut})$, namely $\Phi_{\mathcal B}: \varepsilon_{\mathcal B} \mapsto \varepsilon^\prime_{\mathcal B}$.
Therefore, we recover the nonlinear cell boundary intersection $(\partial K)^\mathrm{cut}$
by replacing the $\varepsilon_{\mathcal{B}}$ edges with the ones parametrized in $\mathcal B_K^\mathrm{cut}$ (see \alglin{alg:boundary-intersection-4} and \fig{fig:boundary-intersection-d}).

\begin{figure}[http]
  \centering
  \begin{subfigure}[b]{0.24\textwidth}
    \includegraphics[trim={12cm 0 17cm 0 0},clip,width=\textwidth]{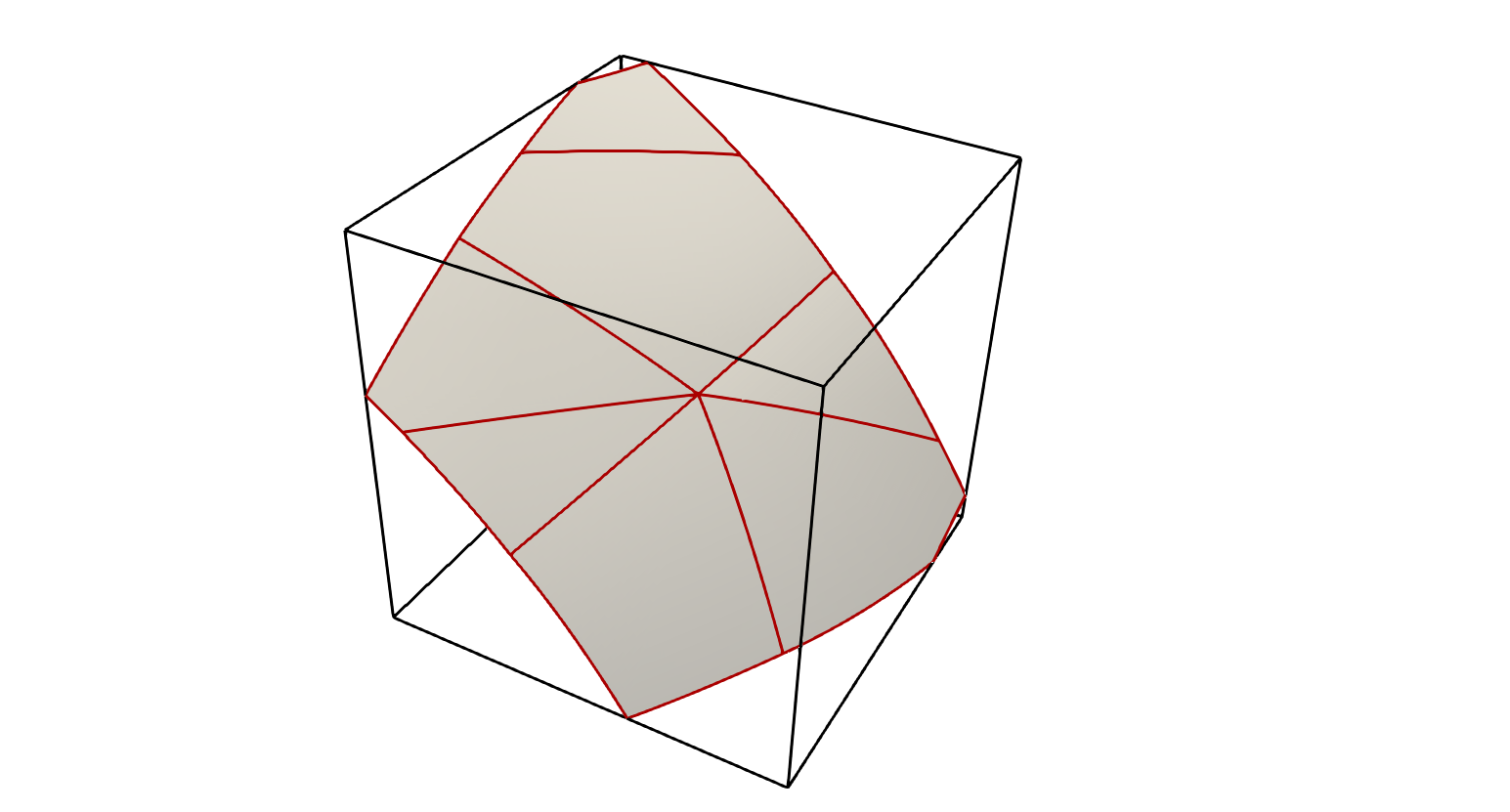}
    \caption{$\mathcal B^\mathrm{cut}_K$}
    \label{fig:boundary-intersection-a}
  \end{subfigure}
  \begin{subfigure}[b]{0.24\textwidth}
    \includegraphics[trim={12cm 0 17cm 0 0},clip,width=\textwidth]{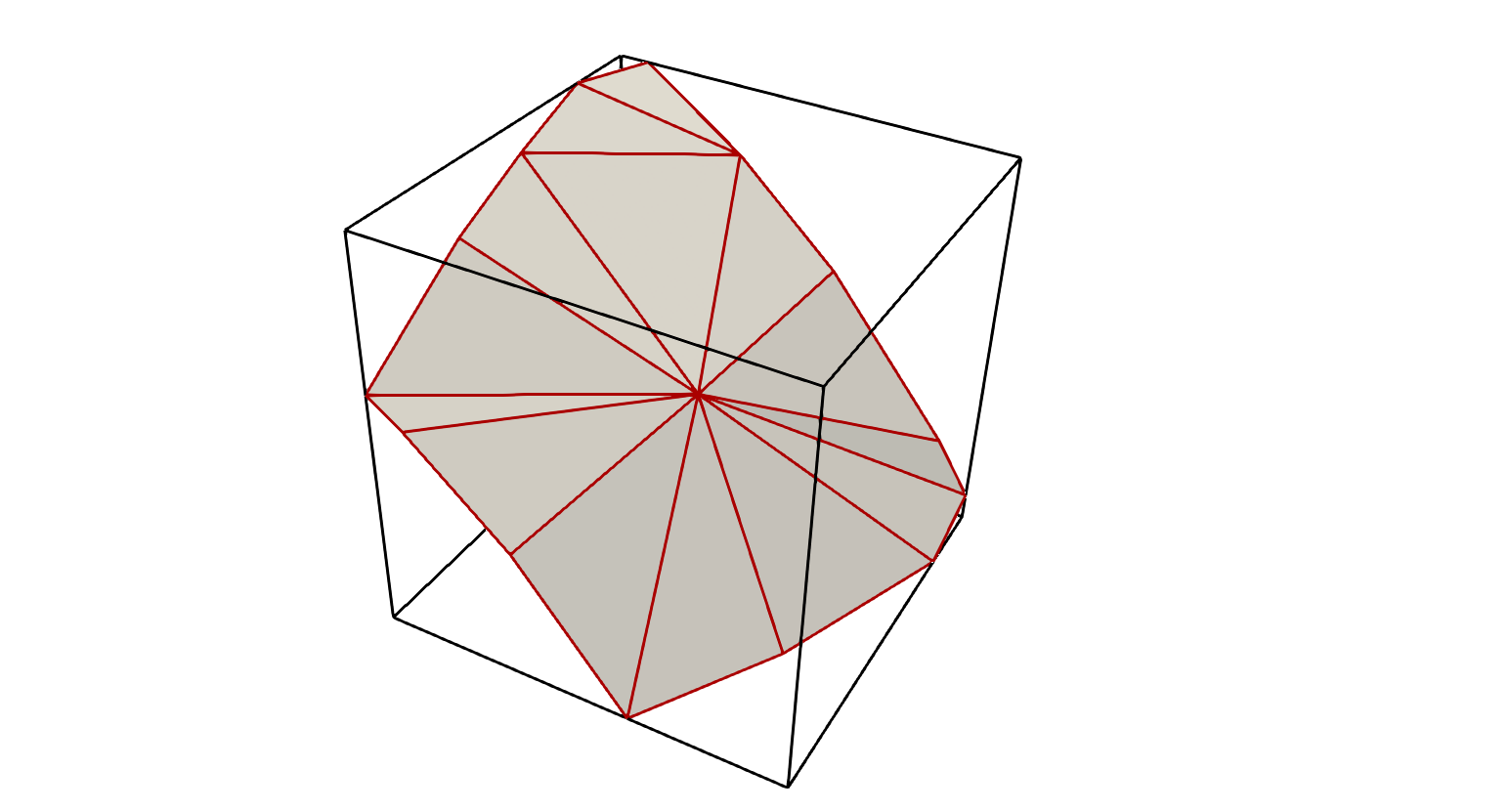}
    \caption{$\mathcal B^\mathrm{lin}_K$}
    \label{fig:boundary-intersection-b}
  \end{subfigure}
  \begin{subfigure}[b]{0.24\textwidth}
    \includegraphics[trim={12cm 0 17cm 0 0},clip,width=\textwidth]{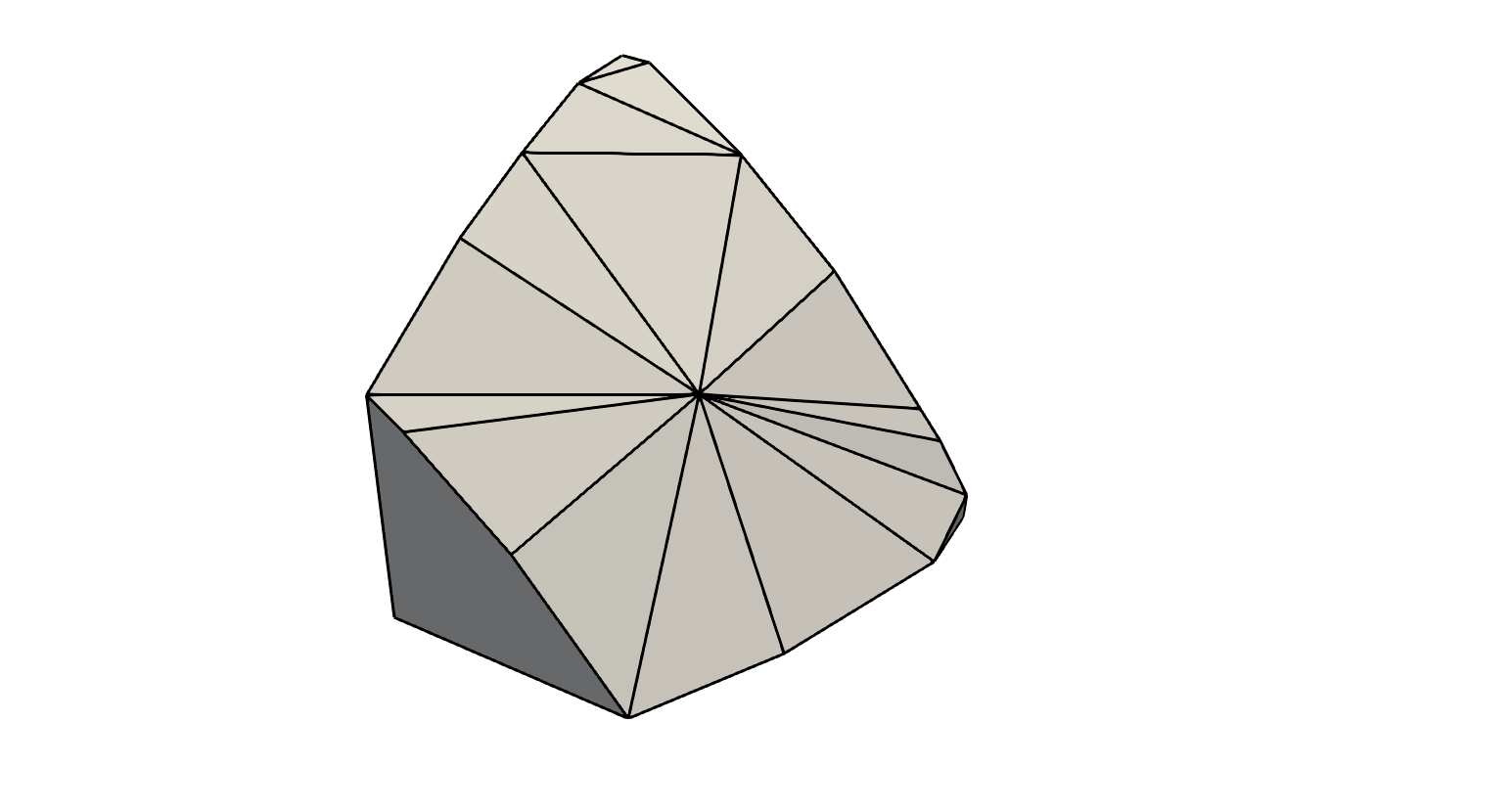}
    \caption{$\mathcal T^\mathrm{lin}_K$}
    \label{fig:boundary-intersection-c}
  \end{subfigure}
  \begin{subfigure}[b]{0.24\textwidth}
    \includegraphics[trim={12cm 0 17cm 0 0},clip,width=\textwidth]{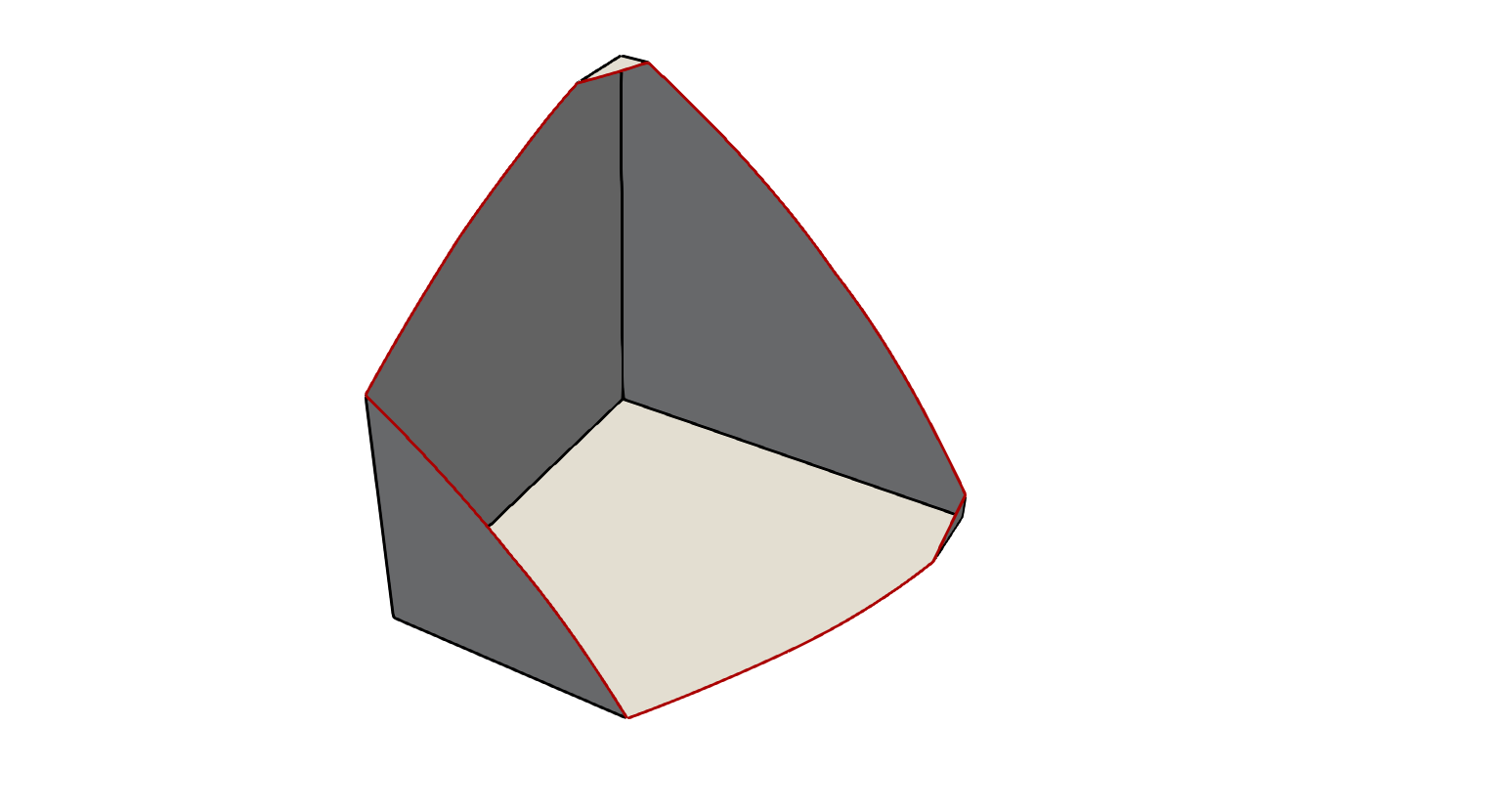}
    \caption{$(\partial K)^\mathrm{cut}$}
    \label{fig:boundary-intersection-d}
  \end{subfigure}
  \caption[Representation of the steps that generate $(\partial K)^\mathrm{cut}$.]{Representation of the steps that generate $(\partial K)^\mathrm{cut}$. First, the surface portion $\mathcal{B}^\mathrm{cut}_K$ (see (a)) is linearized and partitioned into simplices ($\mathcal{B}^\mathrm{lin}_K $ in (b)). Then, in (c), we intersect the background cell $K$ with the half-spaces of the planes of $\mathcal B^\mathrm{lin}$ using linear algorithms, resulting in $\mathcal T^\mathrm{lin}_K$. This linear intersection is possible since the faces $F \in \mathcal B^\mathrm{lin}_K$ are planar. Next, we extract the boundary of $\mathcal T_K^\mathrm{lin}$ that belong to $\partial K$, leading to $(\partial K)^\mathrm{lin}$. Finally, we replace the edges in $(\partial K)^\mathrm{lin}$ by the ones in $\mathcal{ B}^\mathrm{cut}_K$ to obtain $(\partial K)^\mathrm{cut}$ (see (d)).
  The boundary of $K\cap\Omega$ is represented by $(\partial K)^\mathrm{cut}\cup \mathcal B_K^\mathrm{cut}$.}
  \label{fig:boundary-intersection}
\end{figure}

\begin{algorithm}
  \caption{$\partial K \cap \mathrm{int}(\mathcal B_K^\mathrm{cut}) \rightarrow (\partial K)^\mathrm{cut}$}\label{alg:boundary-intersection}
  \begin{algorithmic}[1]
    \STATE $\mathcal B^\mathrm{lin}_K \gets \mathtt{lin}(\mathcal B_K^\mathrm{cut})$
    \label{alg:boundary-intersection-1}
    \STATE $\mathcal T^\mathrm{lin}_K \gets K \cap \mathrm{int}(\mathcal B^\mathrm{lin}_K)$\label{alg:boundary-intersection-2}
    \STATE $(\partial K)^\mathrm{lin} \gets \{ F \in \Lambda^2(P) : P \in T^\mathrm{lin}_K,\ F \subset \partial K \}$
    \label{alg:boundary-intersection-3}
    \RETURN $(\partial K)^\mathrm{cut} \gets \mathtt{replace\_edges}(\ (\partial K)^\mathrm{lin},\  \Phi_{\mathcal B} )$
    \label{alg:boundary-intersection-4}
  \end{algorithmic}
\end{algorithm}

The parametrization of the edges of $(\partial K)^\mathrm{cut}$ is extracted from $\mathcal B^\mathrm{cut}_K$. However, the surface parametrization requires a surface partition into standard polytopes.
We utilize the methods described in \sect{sec:surface-partition} to build such a partition. In this case, the \ac{aa} critical points and \ac{aa} partitions are defined in the reference space of $P \in (\partial K)^\mathrm{cut}$, i.e., in the reference space of $F\in \Lambda^2(K)$. It is worth noting that, in this partition,
the intersections of the nonlinear edges $\varepsilon \in \Lambda^1(\mathcal B_K^\mathrm{cut})$ with \ac{aa} lines are computed in the reference space of $\varepsilon$. This fact ensures local conformity.

\subsection{Global algorithm}\label{sec:global-alg}
The algorithms presented in the previous sections are defined cell-wise. In this section, we describe an algorithm that allows us to integrate \ac{fe} functions in the whole domain and its boundary. Each cell of the background mesh is intersected by the domain bounded by a high-order B\'ezier surface mesh. Therefore, to proceed, we first need to generate this surface mesh from the given \ac{brep}, e.g.,
analytical functions or \ac{cad} representation.
From \ac{cad} models we can generate high-order surface meshes with a third-party library, e.g., \texttt{gmsh} \cite{Geuzaine_2009}. We can convert these meshes into B\'ezier patches with a B\'ezier projection operation or a least-squares approximation.

The intersection of a background cell $K\in \mathcal{T}$ with the whole surface mesh $\mathcal{B}$ would be inefficient. Therefore, we restrict the surface mesh to the faces colliding with the cell $K$. We perform this operation in a preprocessing step similar to \cite{Badia_2022-stl}. The interrogations are approximated by linear operations on the convex hull of each B\'ezier patch.  We can accelerate these queries with a hierarchy of simpler bounding domains, e.g., \acp{aabb}, \acp{obb} \rev{ \cite{ericson2004real} } or \acp{kdop} \cite{Klosowski_1998, Xiao_2019b}. We note that these operations prioritize speed over accuracy as the subsequent operations can deal with false positives in the surface restriction.

The global algorithm is described in \alg{alg:global} and demonstrated through an example in \fig{fig:global-alg}. First, in \alglin{alg:global-alg-2}, the B\'ezier mesh is extracted from the given representation, e.g., \ac{cad} model in \fig{fig:global-alg-a}.
For each cell $K\in \mathcal{T}$, we restrict the faces $F\in \mathcal B$ touching $K$,
see \fig{fig:global-alg-b} step (i) and \alglin{alg:global-alg-4}.
The surface $\mathcal B_K$
is intersected by the walls of $K$, see step (ii) and \alglin{alg:global-alg-5}. Then, the cell boundary is intersected by the half-spaces of the intersected surface $\mathcal B_K^\mathrm{cut}$ (\alglin{alg:global-alg-6}). Both boundary intersections $\mathcal B_K^\mathrm{cut} \cup (\partial K)^\mathrm{cut}$ represent the boundary of the cut cell $K^\mathrm{cut}$, see step (iii). The parameterization of these is stored for integration purposes (\alglin{alg:global-alg-7}).

\begin{algorithm}
  \caption{$\mathcal{T} \cap \mathrm{int}(\mathcal{B}^\mathrm{CAD})$}
  \label{alg:global}
\begin{algorithmic}[1]
  \STATE $\mathcal{T}^\textrm{cut} \gets \emptyset$\label{alg:global-alg-1}
  \STATE $\mathcal{B} \gets \mathtt{extraction}(\mathcal{B}^\mathrm{CAD})$\label{alg:global-alg-2}
  \FOR{$K\in \mathcal{T}$}
    \STATE $\mathcal{B}_K \gets \mathtt{restrict}(\mathcal{B},K)$\label{alg:global-alg-4}
    \STATE $\mathcal{B}^\mathrm{cut}_K \gets  \mathcal{B}_K \cap K$\label{alg:global-alg-5}
    \STATE $ (\partial K)^\mathrm{cut} \gets \partial K \cap \mathrm{int}( \mathcal{B}_K^\mathrm{cut})$\label{alg:global-alg-6}
    \STATE $\mathcal{T}^\textrm{cut} \gets \mathcal{T}^\textrm{cut} \cup \mathtt{parametrize}(\mathcal{B}_K^\mathrm{cut}) \cup \mathtt{parametrize}( (\partial K)^\mathrm{cut} )$\label{alg:global-alg-7}
  \ENDFOR
  \RETURN $\mathcal{T}^\textrm{cut}$
\end{algorithmic}
\end{algorithm}

\begin{figure}[http]
  \centering
  \begin{subfigure}[b]{0.5\textwidth}
    \includegraphics[width=\textwidth]{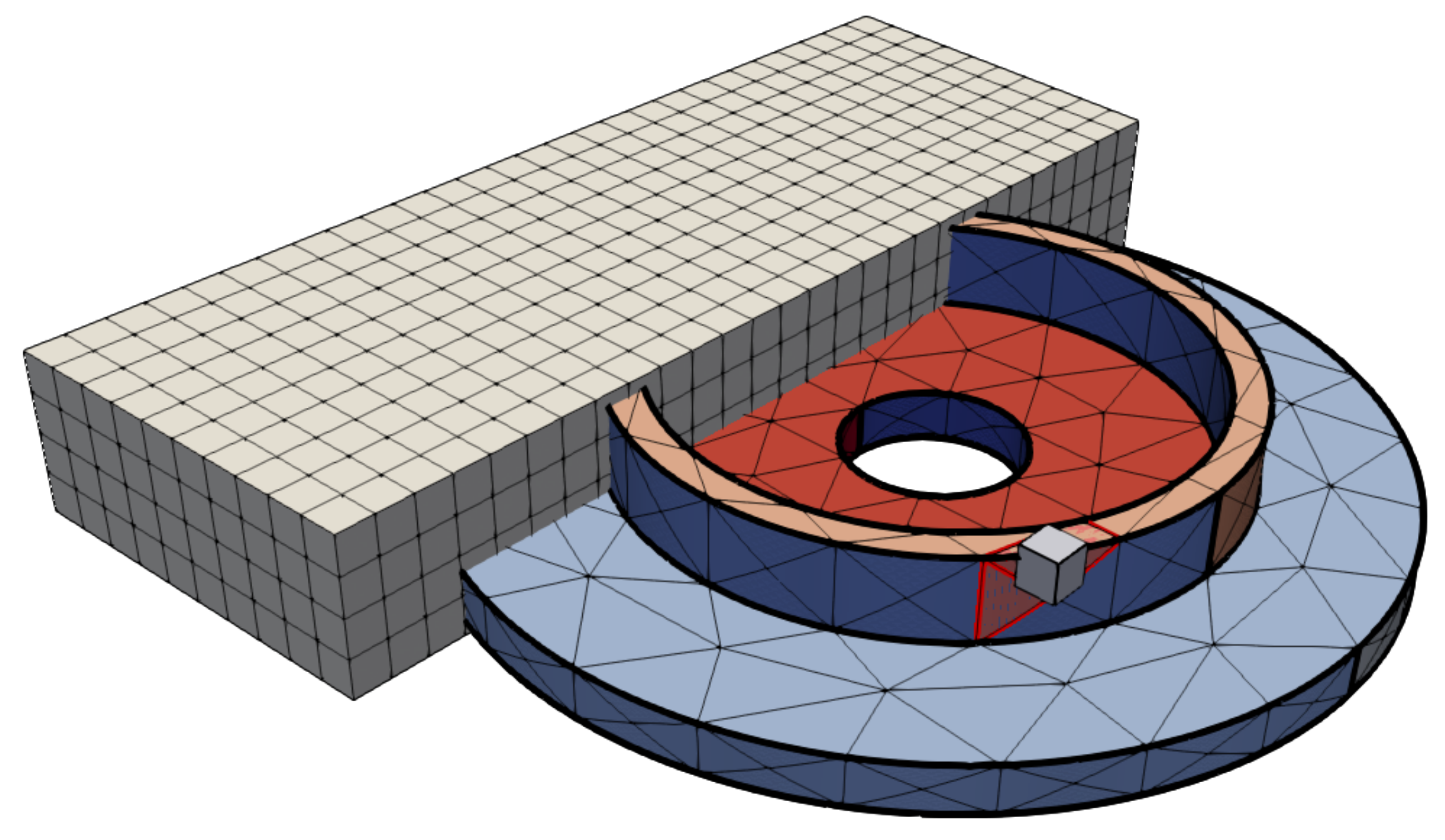}
    \caption{CAD and cell sample}
    \label{fig:global-alg-a}
  \end{subfigure}
  \begin{subfigure}[b]{0.44\textwidth}
    \centering
    \includefig{0.7\textwidth}{local_cut_alg}
    \caption{Local main steps}
    \label{fig:global-alg-b}
  \end{subfigure}
  \caption[CAD geometry and the clipping steps for a background cell.]{In (a) the CAD geometry (colored CAD entities) is first approximated into a high-order surface mesh. In each cell of the background mesh (b), in (i) we restrict the nonlinear faces touching the background cell. Then, in (ii) we clip the nonlinear faces by the background cell walls. Finally, in (iii) we build the polytopal representation of the intersected $\partial K \cap \Omega$. Afterward, we prepare the polytpes for the integration, e.g., with a surface parametrization.}
  \label{fig:global-alg}
\end{figure}

Once we have classified the non-intersected cells as described in  \cite{Badia_2022-stl}, we can proceed with the integration over the entire domain. The integration strategy depends on the dimension of the parametrization used in \alglin{alg:global-alg-7}.
In the methods described in this section, we utilize high-order $2$-faces for parametrization, enabling us to employ Stokes theorem for integration \cite{Chin_2020} in combination with moment-fitting methods \cite{Badia_2022-highorder}.
However, when using edge parametrizations, it is necessary to employ Stokes theorem over trimmed faces \cite{Gunderman_2021}.

\section{Numerical experiments}\label{sec:experiments}
In the numerical experiments, we aim to demonstrate the robustness of the method and optimal convergence of the geometrical and \ac{pde} solution approximations. First,
we demonstrate $hp$-convergence of the intersection and parametrization methods in \sect{sec:conv-approx}. Then, in \sect{sec:robustness-approx}, we show the robustness of the method concerning the relative position of the background mesh and the geometry. Next, in \sect{sec:fe-analysis}, we show the optimal $hp$-convergence of the \ac{fe} analysis for a manufactured solution of the Poisson equation and an elasticity benchmark. Finally, we show the application of the method in real-world examples on \ac{cad} described in terms of \ac{step} files.

\subsection{Experimental setup}
The numerical experiments have been performed on Gadi, a high-end supercomputer at the NCI (Canberra, Australia) with 4962 nodes, 3074 of them
powered by a 2 x 24 core Intel Xeon Platinum 8274 (Cascade Lake) at 3.2 GHz and 192 GB RAM. The algorithms presented in this work have been implemented in the Julia programming language \cite{Julia-2017}.
The unfitted FE computations have been performed using the Julia FE library
\texttt{Gridap.jl} \cite{Badia_2020-gridap,Verdugo_2022} version 0.17.17 and the extension package for unfitted methods \texttt{GridapEmbedded.jl} version 0.8.1 \cite{GridapEmbedded-jl}. \texttt{STLCutters.jl} version 0.1.6 \cite{Martorell_STLCutters_2021} has been used to compute intersection computations on \ac{stl} geometries. The computations of the convex hulls have been performed with \texttt{DirectQhull.jl}  version 0.2.0 \cite{DirectQhull-jl}, a Julia wrapper of the \texttt{qhull} library \cite{Barber_1996}.

The \ac{cad} geometry preprocess is done in \texttt{gmsh} library  \cite{Geuzaine_2009} by using the \texttt{GridapGmsh.jl} Julia wrapper \cite{GridapGmsh-jl}. The \texttt{gmsh} library calls \ac{occt} \cite{occt}  as a parser for \ac{step} files. The sample geometries are extracted from \cite{occt} and \cite{grabcad}.

\subsection{Approximation and parametrization analysis}\label{sec:conv-approx}
In this section, we analyze the approximation of the geometry by a B\'ezier mesh $\mathcal B$. We use a sphere as test geometry. We define the sphere by the boundary of a reference cube bumped by $f(\hat {\pmb{x}}) = \pmb{x}_0 + R(\hat{\pmb{x}} - \hat{\pmb{x}}_0) /\| \hat{\pmb{x}} - \hat{\pmb{x}}_0 \|$ where $\hat{ \pmb{x}}_0$ is the center of the reference cube, $ \pmb{x}_0$ is the center of the sphere and $R=1$ is the radius of the sphere in our experiments. In the following experiments, we consider a triangular surface mesh for the reference cube boundary obtained from the convex decomposition of a Cartesian mesh. We define the relative cell size as $h_\mathrm{surf} = 1 / n_\mathrm{surf}$ where $n_\mathrm{surf}$ is the number of elements in each Cartesian direction. The surface partition of the sphere is approximated by a B\'ezier mesh $\mathcal B$ of order $p$ by using a least-squares operation. This sphere is embedded in a cube of side $L=3$. This domain is discretized with a background Cartesian mesh $\mathcal T$ of relative cell size $h=1/n$ where $n$ is the number of cells in each direction. During the intersection of the surface, the relative chord of the edges is bounded to $\hat\delta_{\max}<0.1$. The surfaces are computed by integrating a unit function on the high-order surface mesh, while we use Stokes theorem in the volume computation.

In the surface approximation experiments of \fig{fig:conv-approx-a}, we test a matrix of surface cell sizes $h_\mathrm{surf}=2^{-\alpha}$, with $\alpha=2,..,5$ and a range of B\'ezier orders $p=1,..,7$. We compute the B\'ezier approximation error as the difference between the surface of $\mathcal B$ and the analytical surface. We observe that the convergence rate of the approximation errors is $p+1$ for odd orders and $p+2$ for even orders. In \cite{Antolin_2019}, one can observe similar convergence rates of the surface and volume errors. \rev{It is worth noting that the convergence rates are affected by machine precision for $p=6,7$ in \fig{fig:conv-approx-a}.}

In \fig{fig:conv-approx-b}, the surface mesh $\mathcal B$ is intersecting by each background cell $K\in\mathcal T$. This generates $\mathcal{B}^\mathrm{cut}$. These intersections are performed for surfaces $\mathcal B$ of order $p$ and parametrized with order $p$, where $p=2,...,6$. In \fig{fig:conv-approx-d}, we present the results of intersecting each background cell $K\in\mathcal{T}$ with the domain bounded by $\mathcal B$, resulting in $\mathcal T^\mathrm{cut}$. In this case, the approximation of $\mathcal B$ and surface parametrization are computed with order $p=2,3,4$ for computational reasons. In both cases, the cell size of the background mesh is fixed to $h=2^{-2}$. We observe in \fig{fig:conv-approx-b} and \fig{fig:conv-approx-d} that the convergence rates of the cut surface error, $|\mathtt{surf}(\mathcal B^\mathrm{cut}) - \mathtt{surf}(\mathcal B)|$, and the cut volume error, $|\mathtt{vol}(\mathcal T^\mathrm{cut}) - \mathtt{vol}(\mathcal B)|$, are similar to the convergence rates of \fig{fig:conv-approx-a}.

\begin{figure}[http]
  \centering
  \includegraphics[width=\textwidth]{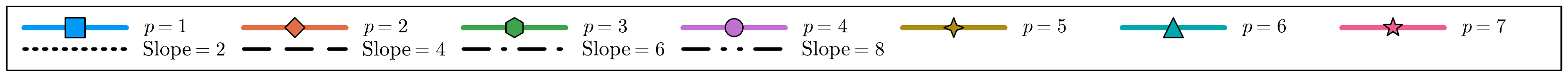}
  \begin{subfigure}{0.24\textwidth}
      \includegraphics[width=\textwidth]{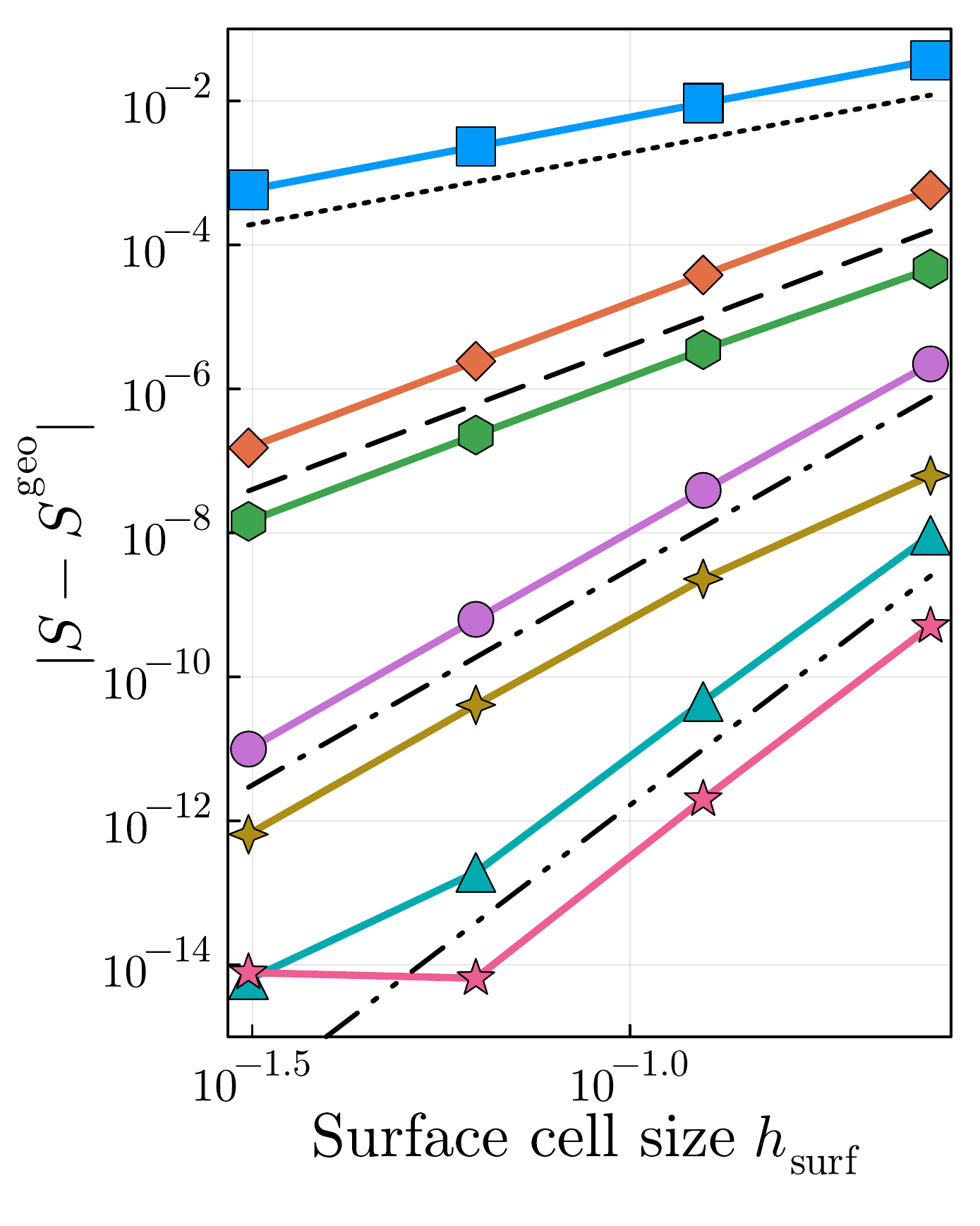}
      \caption{}
      \label{fig:conv-approx-a}
  \end{subfigure}
  \begin{subfigure}{0.24\textwidth}
      \includegraphics[width=\textwidth]{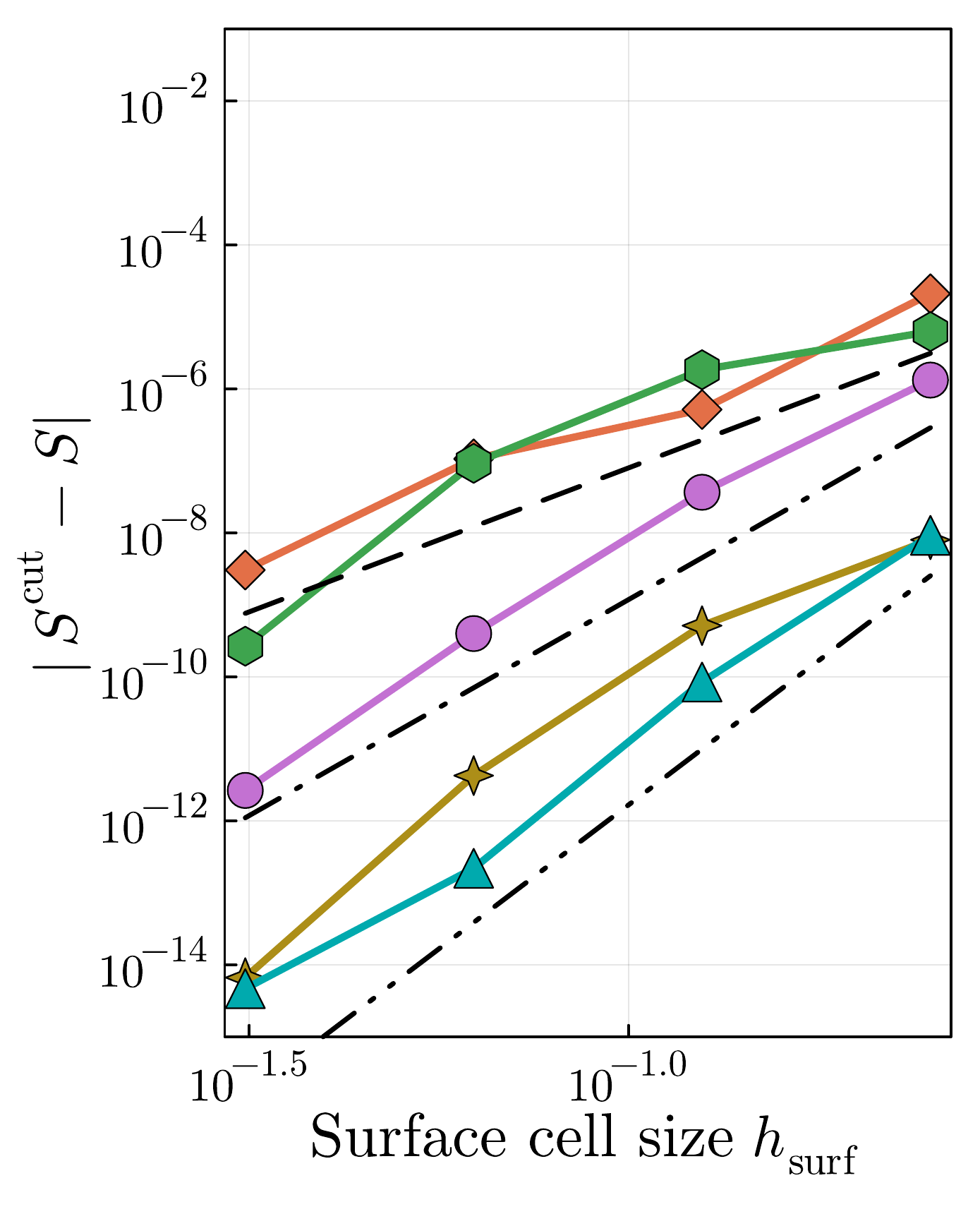}
      \caption{}
      \label{fig:conv-approx-b}
  \end{subfigure}
  \begin{subfigure}{0.24\textwidth}
      \includegraphics[width=\textwidth]{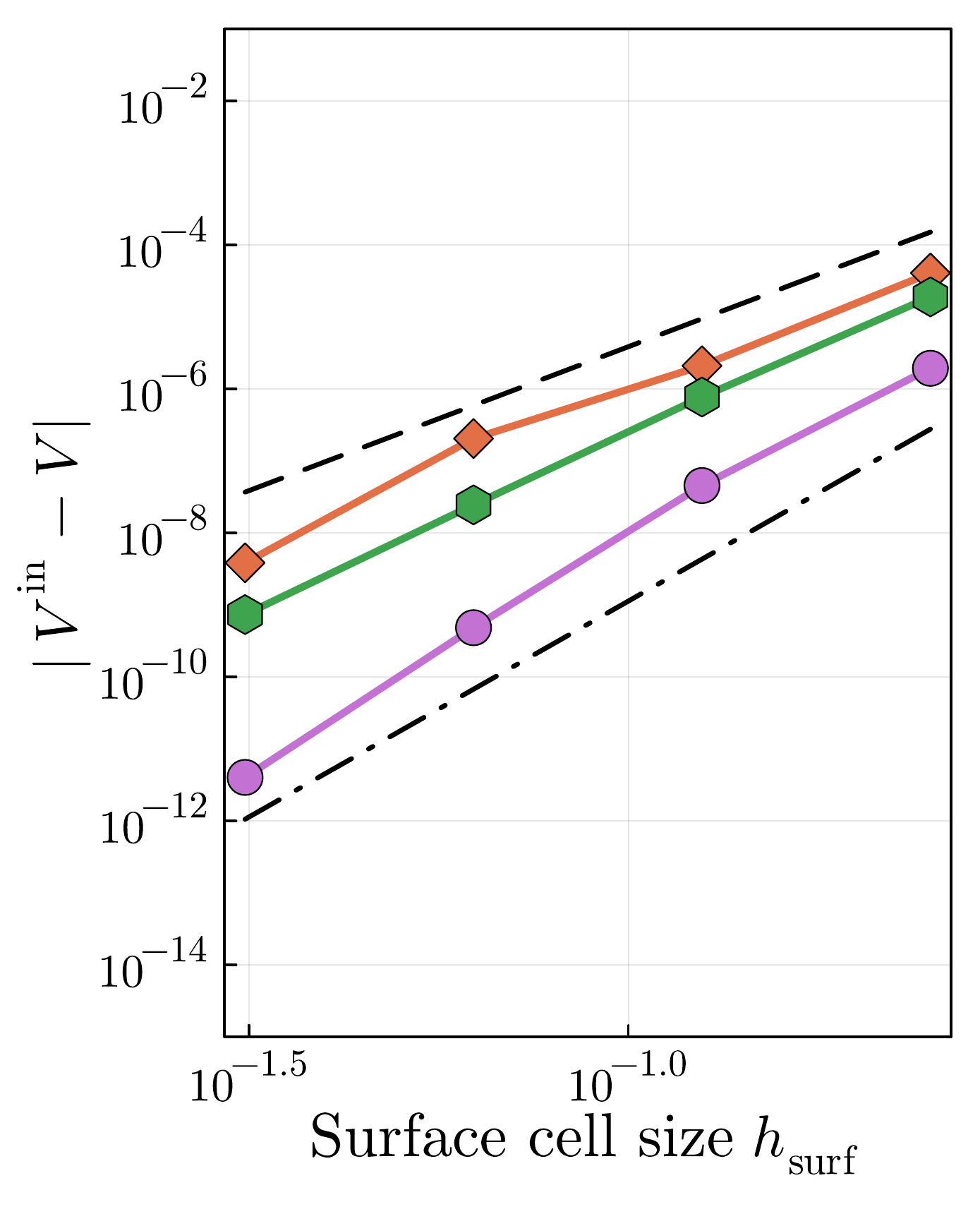}
      \caption{}
      \label{fig:conv-approx-d}
  \end{subfigure}
  \caption[Surface and volume errors of the approximation and the reparametrization of a sphere.]{Surface and volume errors of the approximation and the reparametrization of a sphere. In (a), the surface error of approximating the geometry
  $\mathcal B^\mathrm{geo}$ into a B\'ezier mesh $\mathcal B$, $| S - S^\mathrm{geo} |$ where
  $S^\mathrm{geo}=\mathtt{surf}(\mathcal B^\mathrm{geo})$ and $S=\mathtt{surf}(\mathcal B)$.
  In (b), the parametrization error of $\mathcal{B}^\mathrm{cut}$, the intersection of $\mathcal B$ with each background cell $K\in\mathcal T$,
  $|S^\mathrm{cut} - S|$
  where
  $S^\mathrm{geo}=\mathtt{surf}(\mathcal B^\mathrm{geo})$.
  In (c), the volume integration error of $\mathcal T ^\mathrm{in}$, the intersection of each background cell $K\in\mathcal T$ with the domain bounded by $\mathcal B$,
  $|V^\mathrm{in}-V|$ where
  $V^\mathrm{in}=\mathtt{vol}(\mathcal T ^\mathrm{in})$ and $V=\mathtt{vol}(\mathcal B)$. }
\end{figure}

\subsection{Robustness experiments}\label{sec:robustness-approx}

In this section, we demonstrate robustness concerning the relative position of $\mathcal B$ and $\mathcal T$. For these experiments, we consider a sphere with the same geometrical setup as in \sect{sec:conv-approx}. We perform the intersections with the same relative background cell size and surface cell size $h=h_\mathrm{surf}=2^\alpha$, $\alpha = 3,4$. We test for the approximation and parametrization orders $p=2,...,4$. In each combination, we shift the geometry from the origin a distance $\Delta x = (i/\rev{500}) h$, with $i=1,...,\rev{500}$.

In \fig{fig:disp-approx-a} and \fig{fig:disp-approx-b}, we observe the surface and volume variation of the surface error, $|\mathtt{surf}(\mathcal B^\mathrm{cut}) - \mathtt{surf}(\mathcal B)|$, and the volume error, $|\mathtt{vol}(\mathcal T^\mathrm{in}) - \mathtt{vol}(\mathcal B)|$, resp. Here,  $\mathcal T^\mathrm{in} \doteq \mathcal T ^\mathrm{cut} \cup \{K \in \mathcal T :  K \cap \partial \Omega = \emptyset \}$ is the physical volume mesh.
The variations of the surface error are approximately two orders of magnitude and the variations of volume are one order of magnitude.

However, in \fig{fig:disp-approx-c}, we observe a machine precision error in the domain volume error $|\mathtt{vol}(\mathcal T^\mathrm{in}) + \mathtt{vol}(\mathcal T^\mathrm{out}) - \mathtt{vol}(\mathcal T^\mathrm{bg})|$, where $\mathcal T^\mathrm{out} \doteq \mathcal{T}^\mathrm{bg} \setminus \mathcal{T}^\mathrm{in}$ is complementary of  $\mathcal{T}^\mathrm{in}$. The low error is due to the conformity between inside and outside polytopes of the cut cell described in \sect{sec:cell-int}.

\begin{figure}[http]
  \centering
  \includegraphics[width=0.4\textwidth]{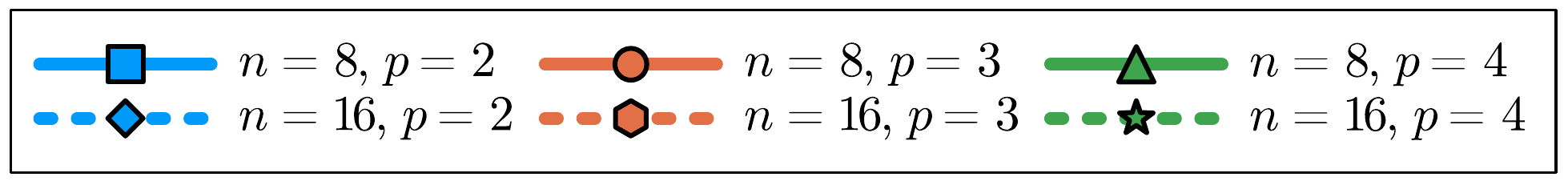}

    \begin{subfigure}{0.24\textwidth}
        \includegraphics[width=\textwidth]{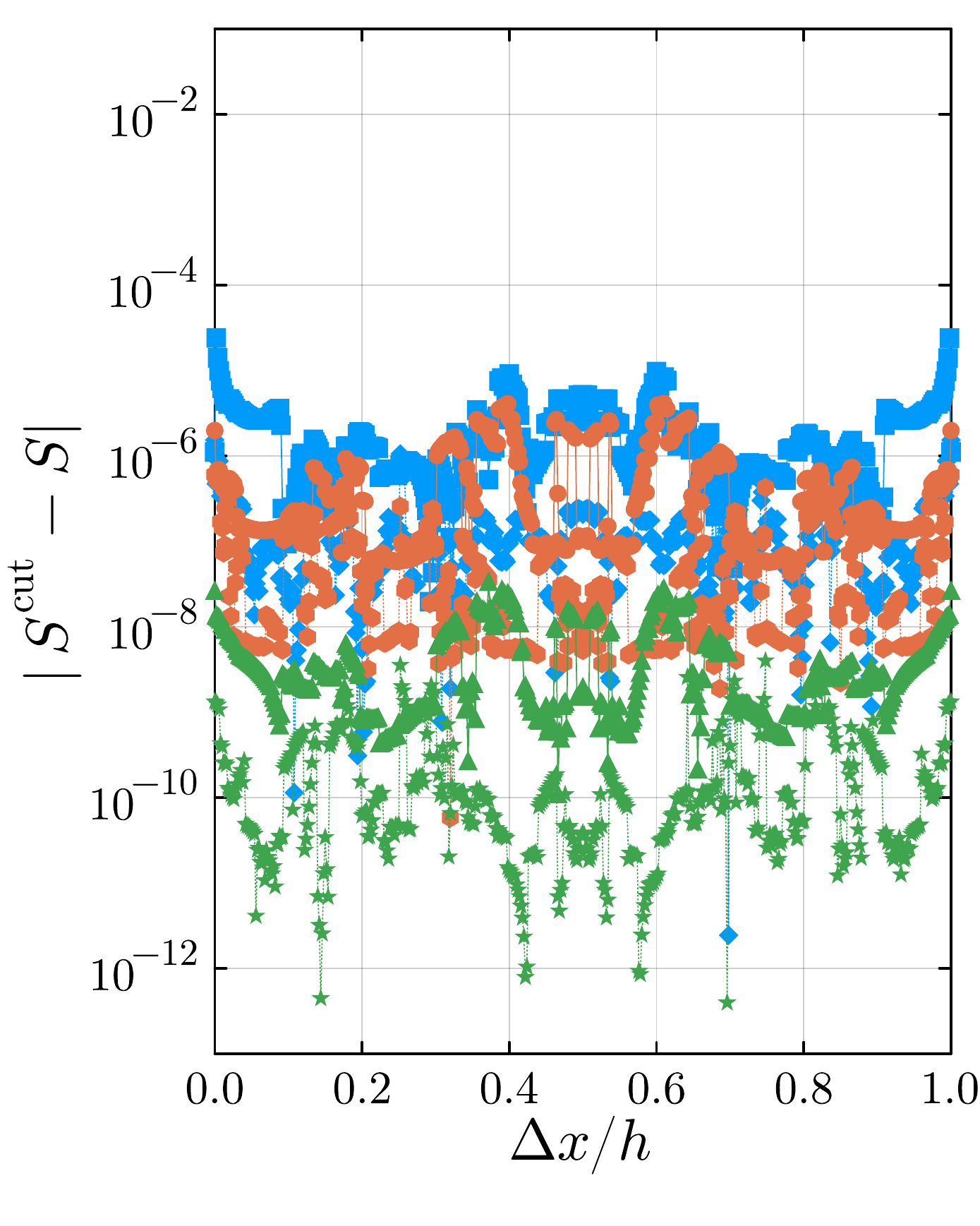}
        \caption{}
        \label{fig:disp-approx-a}
    \end{subfigure}
    \begin{subfigure}{0.24\textwidth}
        \includegraphics[width=\textwidth]{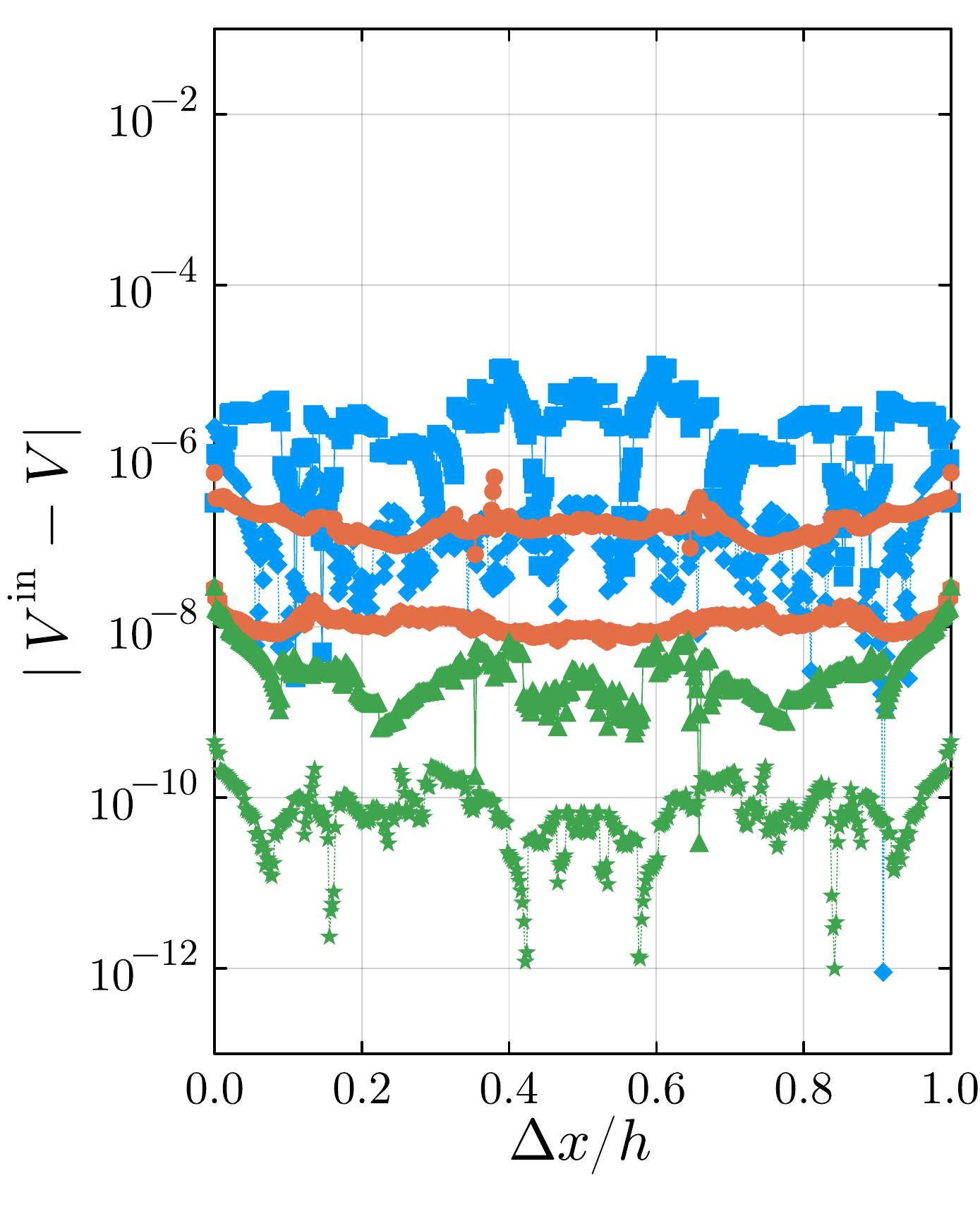}
        \caption{}
        \label{fig:disp-approx-b}
    \end{subfigure}
    \begin{subfigure}{0.24\textwidth}
        \includegraphics[width=\textwidth]{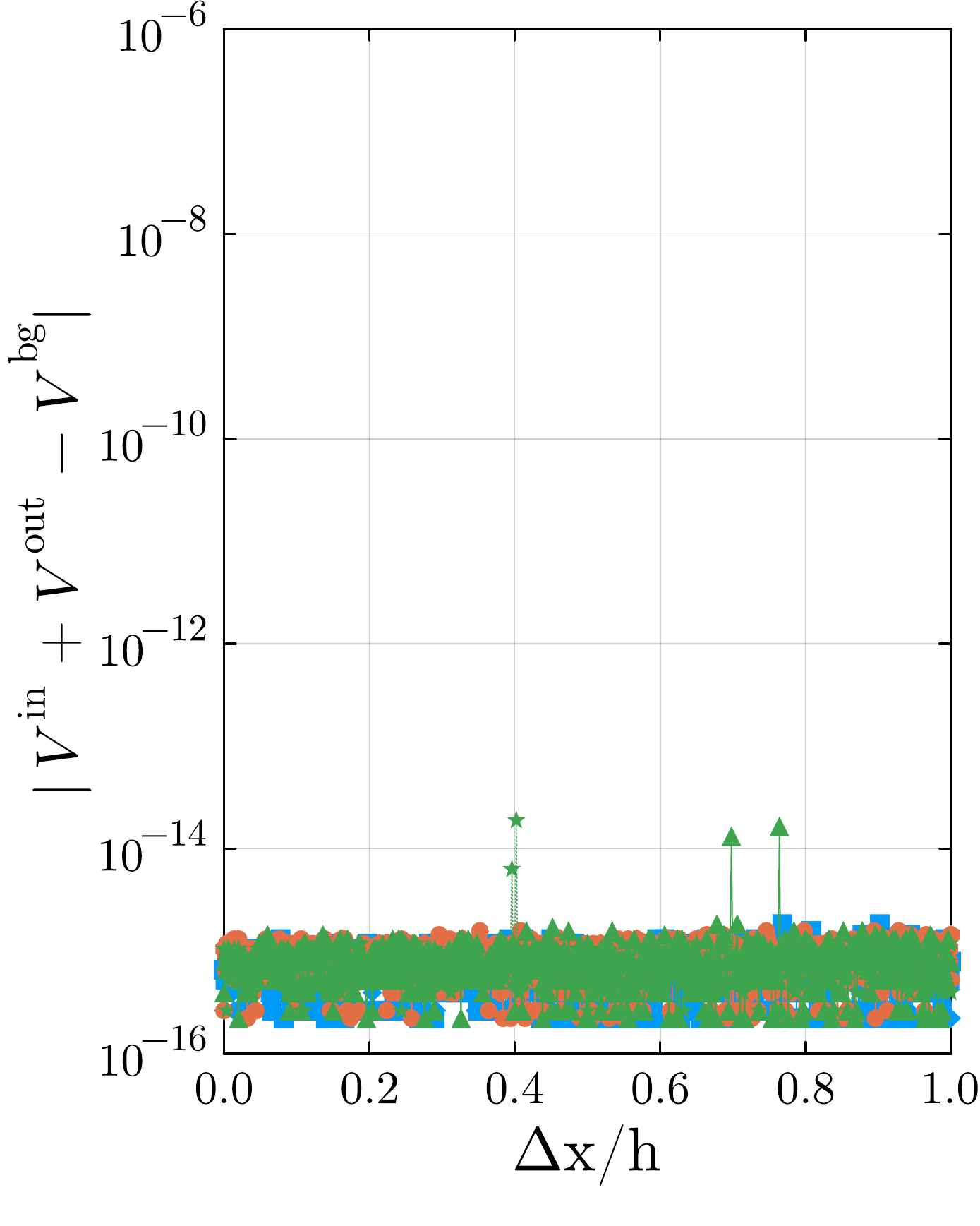}
        \caption{}
        \label{fig:disp-approx-c}
    \end{subfigure}
  \caption[Demonstration of robustnes concerning the relative position of $\mathcal B$ and $\mathcal T$.]{Demonstration of robustnes concerning the relative position of $\mathcal B$ and $\mathcal T$. Both plots have surface and volume errors when shifting a sphere in the embedded domain. Even though the surface and volume errors in (a) and (b) show variations, the errors are bounded. In (c), the domain volume errors are close to machine precision because the inside and outside polytopes of the cut cell are conforming, see \sect{sec:cell-int}}
  \label{fig:disp-approx}
\end{figure}

\subsection{Unfitted FE experiments}\label{sec:fe-analysis}
In these experiments, we explore the behavior of the intersection algorithms by solving \acp{pde} with unfitted \acp{fe} on two analytical benchmarks and two realistic examples. We analyze a Poisson equation with a manufactured solution on a sphere and a spherical cavity benchmark with an analytically derived solution \cite{Bower_2012}. Both experiments are performed with the geometrical setup described in \sect{sec:conv-approx}. However, in the spherical cavity experiments, the domain is the outside of an octant of the sphere. We compute the \ac{fe} solutions using \ac{agfem} with modal $\mathcal C^0$ basis and moment-fitting quadratures, even though the proposed framework can be used with other unfitted methods like ghost penalty stabilization \cite{burman2010ghost}. 
\rev{
We follow the methodology in \cite[Sec. 6.2]{Badia_2022-highorder} for the computation of  moment-fitting quadratures; we leverage an extension of Lasserre's method in \cite{Chin2015}. It is shown in \cite[Fig. 7]{Badia_2022-highorder} that the computed moment-fitted quadratures are much more efficient than marching cube ones while offering the same level of accuracy (at least) up to $m = 5$.
}
In \ac{agfem}, we aggregate all cut cells.

In the Poisson experiments, we consider Dirichlet boundary conditions and a forcing term that satisfies the manufactured solution $u(x,y,z) = x^a+y^a$, with $a=6$. We compute the convergence tests for $h = h_\mathrm{surf} = 2^{-\alpha}$, with $\alpha=3,4,5$, for the \ac{fe} order $p=1,2,3$ and the geometrical order $q=1,2,3$ of the approximation and parametrizations. In \fig{fig:conv-fe-a} and \fig{fig:conv-fe-b}, we can observe that the $L^2$ and $H^1$ errors converge with the optimal rate, $p+1$ and $p$, resp. This convergence is independent of the geometrical order $q$, as expected, the geometry description does not affect the manufactured solution problem.

In the linear elasticity benchmark of the spherical cavity, we derive the potentials with automatic differentiation in Julia. We compute the tests for the same cell sizes $h = h_\mathrm{surf} = 2^{-\alpha}$ and orders $p=1,2,3$ and $q=1,2,3$ than in the Poisson experiments. We set consider a Young modulus $E=10^5$ and Poisson ratio $\nu = 0.3$. We observe in \fig{fig:conv-fe-c} and \fig{fig:conv-fe-d} the expected convergence rates for the $L^2$ and $H^1$ errors, $\min(p,q)+1$ and $p$ resp. This demonstrates that we require high-order discretizations to accurately solve \acp{pde} with high-order unfitted \ac{fe} methods.

\begin{figure}[http]
  \centering
  \includegraphics[width=0.55\textwidth]{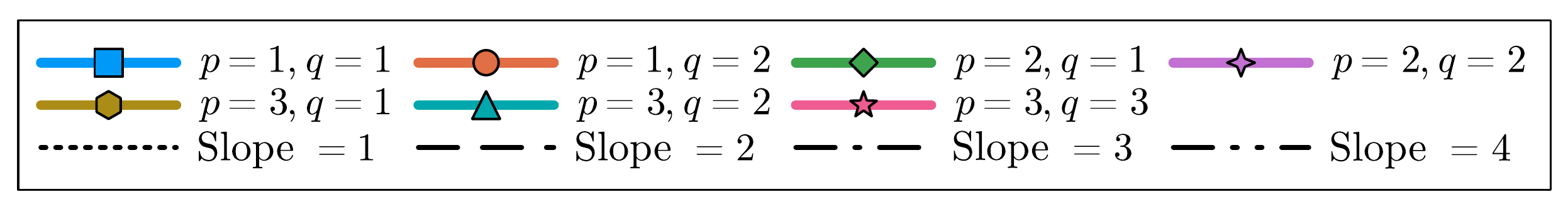}

  \begin{subfigure}{0.24\textwidth}
      \includegraphics[width=\textwidth]{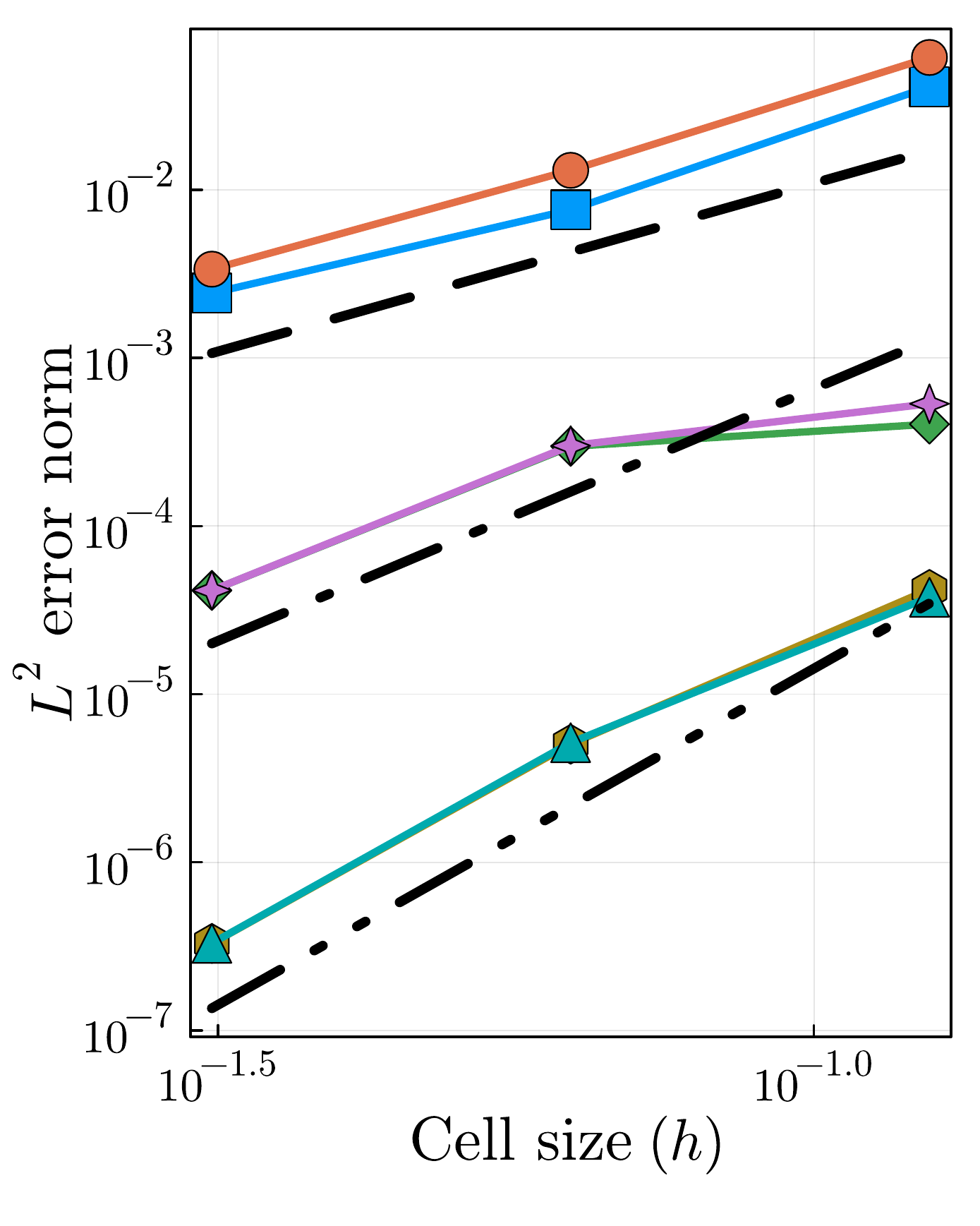}
      \caption{Poisson $L^2$}
      \label{fig:conv-fe-a}
  \end{subfigure}
  \begin{subfigure}{0.24\textwidth}
      \includegraphics[width=\textwidth]{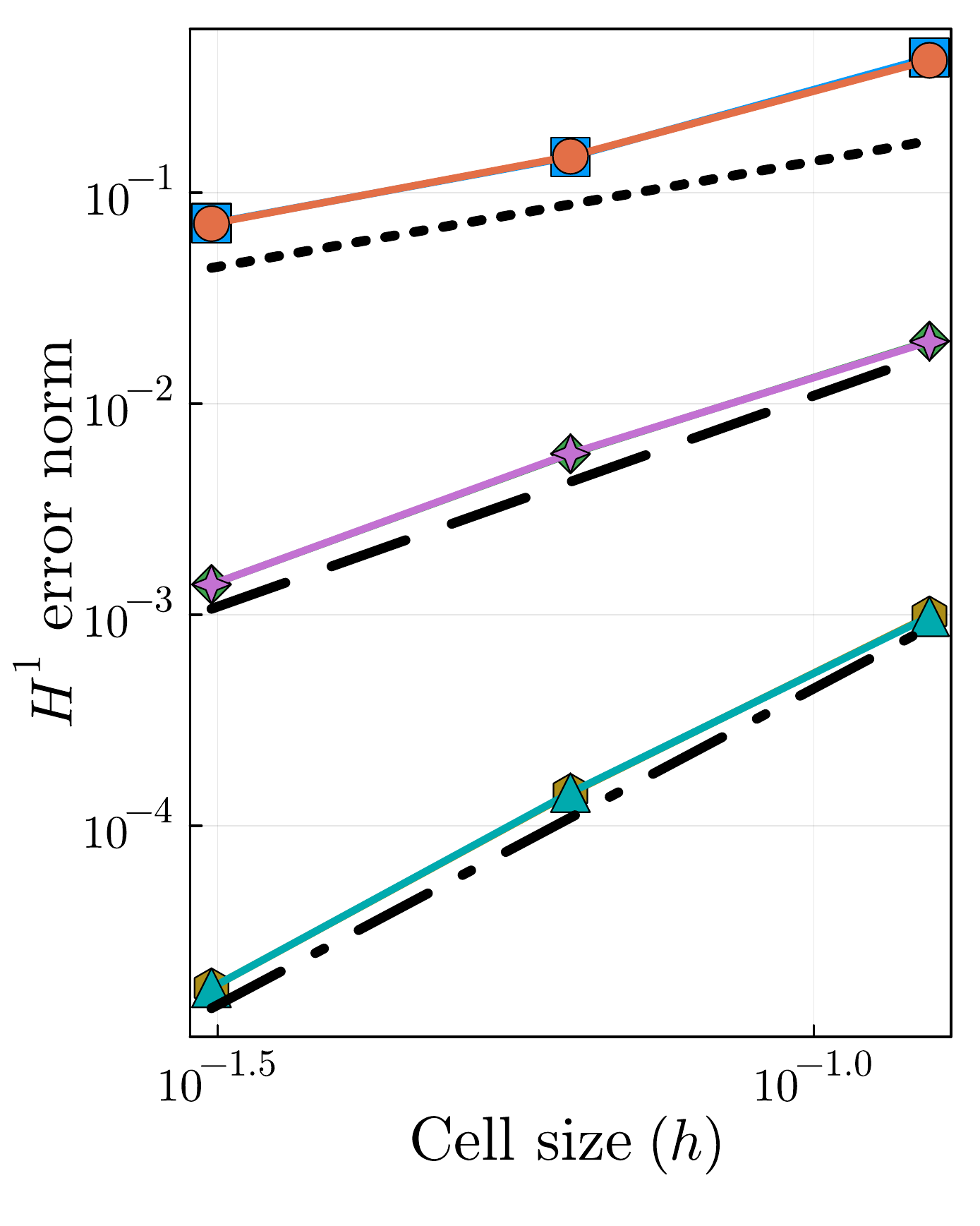}
      \caption{Poisson $H^1$}
      \label{fig:conv-fe-b}
  \end{subfigure}
  \begin{subfigure}{0.24\textwidth}
      \includegraphics[width=\textwidth]{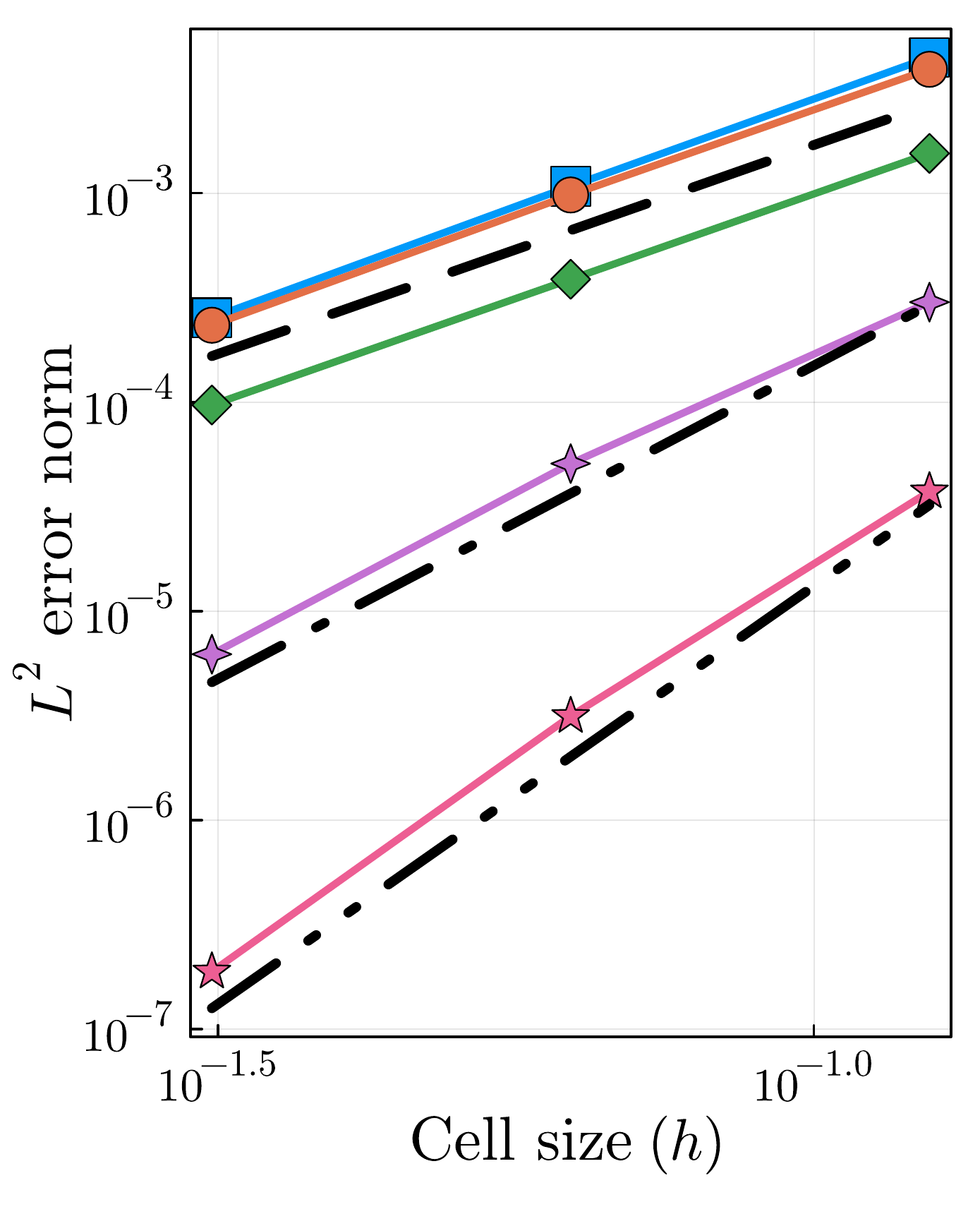}
      \caption{Linear Elasticity $L^2$}
      \label{fig:conv-fe-c}
  \end{subfigure}
  \begin{subfigure}{0.24\textwidth}
      \includegraphics[width=\textwidth]{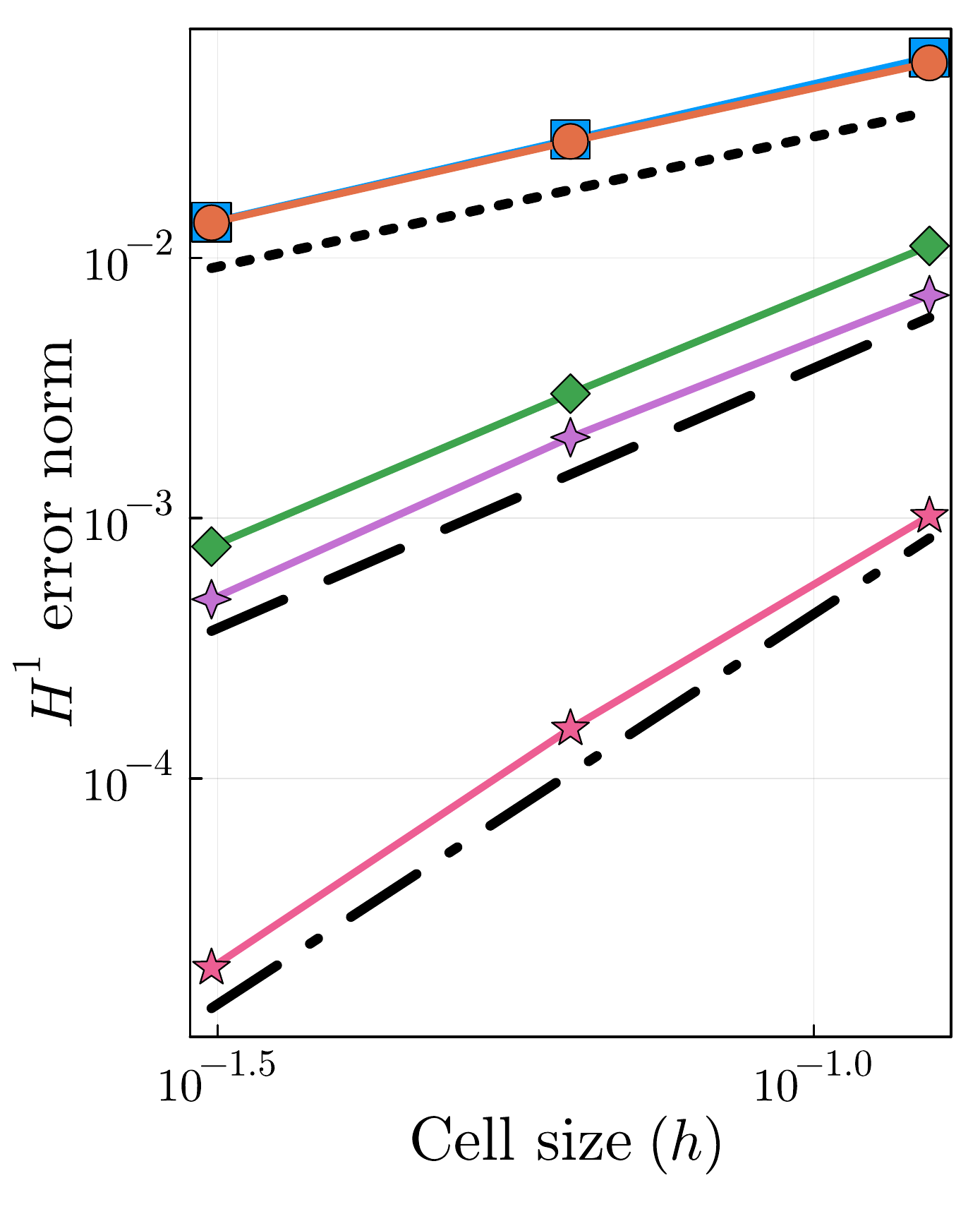}
      \caption{Linear Elasticity $H^1$}
      \label{fig:conv-fe-d}
  \end{subfigure}
  \caption[Convergence tests in \acs{agfem}.]{Convergence tests in \ac{agfem}. Convergence of the manufactured solution with a Poisson equation in a sphere, in (a) and (b). Convergence of the linear elasticity benchmark in a spherical cavity, in (c) and (d). Some combinations of \ac{fe} order $p$ and geometry order $q$ are not shown for the sake of conciseness.}
  \label{fig:conv-fe}
\end{figure}

Finally, we demonstrate the viability of the method in real-world geometries defined by \ac{cad} models in \ac{step} files. These files are extracted from \cite{occt} and \cite{grabcad}, resp. On each \ac{cad} geometry, we generated a high-order surface mesh using \texttt{gmsh}. Then, we converted this mesh into a B\'ezier mesh with a least-squares method.
The indexing of topological entities of the \ac{cad} surface, $\mathcal B^\mathrm{CAD}$, is preserved in the intersected mesh $\mathcal B^\mathrm{cut}$.
Thus, we can impose Dirichlet and Neumann boundary conditions over the entities in $\mathcal B^\mathrm{CAD}$. We consider a heat problem and an elasticity problem in the two different geometries \fig{fig:cad-examples}. The experiments are described in the figure caption.

\begin{figure}
  \centering
  \begin{subfigure}[b]{0.4\textwidth}
      \includegraphics[width=\textwidth]{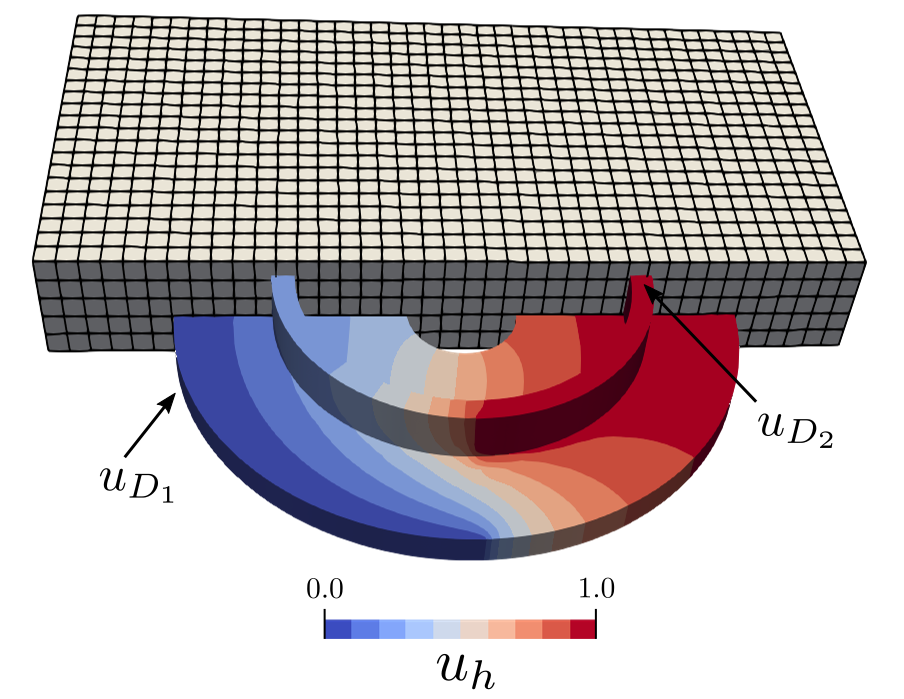}
      \caption{}
      \label{fig:cad-examples-a}
  \end{subfigure}
  \hspace{0.1\textwidth}
  \begin{subfigure}[b]{0.3\textwidth}
      \includegraphics[width=\textwidth]{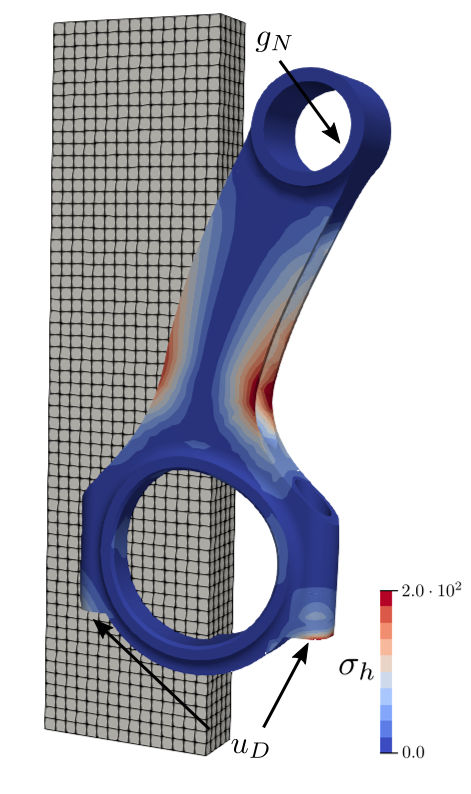}
      \caption{}
      \label{fig:cad-examples-b}
  \end{subfigure}
  \caption[Realistic examples on \acs{cad} geometries.]{Realistic examples on \ac{cad} geometries. In (a) we consider a heat equation with thermal conductivity $k=1.0$ and Dirichlet boundary conditions on two opposite entities ($u_{D_1}=0$ and $u_{D_2}=1$). The background mesh is defined in a \ac{aabb} 20\% larger than the geometry with $40\times 40 \times 5$ elements of order $p=2$. The surface is approximated in 524 quadratic B\'ezier patches. In (b) we consider a linear elasticity problem with Young modulus $E=10^5$ and Poisson ratio $\nu = 0.3$, with Dirichlet and Neumann boundary conditions ($u_D=(0,0,0)$ and $g_N=(1,0,-1)$, resp.). The magnification factor of the deformation is 75. The background mesh is defined in an \ac{aabb} 20\% bigger than the geometry with $30\times 10 \times 60$ elements of order $p=2$. The surface is approximated in 1786 quadratic B\'ezier patches.
  }
  \label{fig:cad-examples}
\end{figure}

\section{Conclusions and future work}\label{sec:conclusions}
In this work, we have designed an automated pipeline for numerically approximating \acp{pde} in complex domains defined by high-order boundary meshes using unfitted \ac{fe} formulations in a structured background mesh. The main challenge of the method lies in the numerical integration of the background cells intersected by the domain boundary. This requires handling the intersection between background cells and complex boundary meshes, including the computation of trimming curves and dealing with nonlinear and nonconvex domains.

We have presented a novel intersection algorithm for general high-order surfaces and polytopal cells, which is accurate and robust. The algorithm is based on mesh partition methods for nonlinear level sets, linear clipping algorithms for general polytopes, multivariate root finding, geometrical least-squares methods and the properties of B\'ezier patches. The result of this algorithm is a set of nonlinear general polyhedra that represent the cut cells of the background mesh. We parametrize the boundary of these polyhedra using sets of B\'ezier patches, taking advantage of concepts like convex hull and kernel point concepts. These B\'ezier patches can be used to integrate the bulk with moment-fitting quadratures.

The implementation, accuracy and robustness of the geometrical algorithm have been tested on high-order geometries defined by analytical methods and \ac{cad} models. In our tests, we observe optimal convergence of the numerical approximations, limited only by rounding errors. Additionally, we have observed the robustness of the method when varying the relative position of the background mesh. Furthermore, we have successfully solved \acp{pde} on geometries defined by nonlinear boundary representations with high-order \ac{fe} methods. We have demonstrated optimal convergence of the solutions in the designed benchmarks.  We have also shown the viability of the method in real-world geometries defined by \ac{cad} models in \ac{step} files. These results position the method as a pioneering computational framework for simulating \acp{pde} on high-order geometries with unfitted \ac{fe} methods. It provides an automatic geometrical and functional discretization that can be especially useful within shape and topology optimization loops, inverse problems with unknown boundaries and interfaces, and transient problems with moving domains.

Future work involves the extension of the method for other background discretizations, e.g., octree meshes for adaptive refinement. We also plan to extend the method to distributed memory \cite{Badia_2022-distributed}, since the method is cell-wise parallel, allowing us to solve larger problems. We can optimize the method by parametrizing only the polytopal edges and using moment-fitting integration on the surfaces \cite{Gunderman_2021}. Further extension of the method involves solving boundary layer problems with a separate discretization, see  \cite{Wei_2021}. Finally, we plan to extend the method to more complex scenarios and practical applications, such as \ac{fsi} and transient problems.

\newcommand{\thethanks}{This research was partially funded by the Australian Government through the Australian Research Council (project numbers DP210103092 and DP220103160). We acknowledge Grant PID2021-123611OB-I00 funded by MCIN/AEI/10.13039/501100011033 and by ERDF ``A way of making Europe''. P.A. Martorell acknowledges the support received from Universitat Politècnica de Catalunya and Santander Bank through an FPI fellowship (FPI-UPC 2019). This work was also supported by computational resources provided by the Australian Government through NCI under the National Computational Merit Allocation Scheme.}

\section*{Acknowledgments}

\thethanks  

\setlength{\bibsep}{0.0ex plus 0.00ex}
\bibliographystyle{myabbrvnat}
\bibliography{refs}

\end{document}